\documentclass[11pt]{article}

\usepackage{fullpage}
\usepackage{amsfonts, amssymb, amsmath, amsthm}
\usepackage{latexsym}
\usepackage{url}
\usepackage[colorlinks=false,breaklinks]{hyperref}
\usepackage{graphicx}
\usepackage{color}
\usepackage[tight]{subfigure}
\usepackage{cancel}
\usepackage{multirow}
\usepackage{caption}	%% caption package is useful for allowing line breaks within figure captions

\usepackage{import}
\usepackage{ltxtable}

\usepackage{calc}	%% The calc package reimplements the LaTeX commands \setcounter, ..., so that these commands accept an infix notation expression.

\def\place #1#2#3{\mspace{#2}\makebox[0pt]{\raisebox{#3}{#1}}\mspace{-#2}}	% the second argument should be in mu, and the third argument in pt

\newcommand{\bra}[1]{{\langle#1|}}
\newcommand{\ket}[1]{{|#1\rangle}}
\newcommand{\braket}[2]{{\langle#1|#2\rangle}}
\newcommand{\ketbra}[2]{{\ket{#1}\!\bra{#2}}}
\newcommand{\abs}[1]{{\lvert #1\rvert}}	% since the delimiters do not scale, it might be a good idea to add a dummy {} at the end, so \abs{big expression}^2 has the superscript at a low height
\newcommand{\norm}[1]{{\| #1 \|}}
\newcommand{\eps}{{\epsilon}}
\newcommand{\binomial}[2]{\ensuremath{\left(\begin{smallmatrix}#1 \\ #2 \end{smallmatrix}\right)}}
\DeclareMathOperator{\Ex}{\operatorname{E}}

\def\tensor {\otimes}
\def\adjoint{\dagger} %{*}
\def\C {{\bf C}}
\def\N {{\bf N}}
\def\R {{\bf R}}
\def\L {{\mathcal L}}
\def\D {{\mathcal D}}

\DeclareMathOperator{\Span}{\operatorname{Span}}
\DeclareMathOperator{\Range}{\operatorname{Range}}
\DeclareMathOperator{\Kernel}{\operatorname{Ker}}

\newcommand{\identity}{\ensuremath{\boldsymbol{1}}} %\mathbb{I}

\newcommand{\ADV} {\mathrm{Adv}}
\newcommand{\ADVpm} {\mathrm{Adv}^{\pm}}

\newcommand{\B}{B}	% \{0,1\}	{{\bf Z}_2}

\DeclareMathOperator{\wsizeop}{{\operatorname{wsize}}}		% short for span program witness size
\newcommand{\wsize}[1]{{\wsizeop({#1})}}
\newcommand{\wsizeD}[1]{\wsize{{#1},\D}}
\newcommand{\wsizeS}[2]{{\wsizeop_{#2}({#1})}}
\newcommand{\wsizeSD}[2]{{\wsizeS{{#1},\D}{#2}}}
\newcommand{\wsizex}[2]{{\wsizeop({#1},{#2})}}
\newcommand{\wsizexS}[3]{{\wsizeop_{#3}({#1},{#2})}}
\def\Ifree{I_\mathrm{free}}		% note: there are no extra brackets around I_free, so \Ifree' will have a prime next to the I
\def\Pifree{{\Pi_\mathrm{free}}}
\def\Pim {{\overline{\Pi}}}
\def\jc {\varsigma}	% this is used in the proof of \thmref{t:spanprogramcomposition} to map i in I_{j,b} to (j,b)

\DeclareMathOperator{\abst}{\operatorname{abs}}
\def\biadj {B}	% A_G is the adjacency matrix, \biadj_G the biadjacency matrix

\def\algorithm {{\mathcal{A}}}

\DeclareMathOperator{\AND}{\ensuremath{\operatorname{AND}}}
\DeclareMathOperator{\OR}{\ensuremath{\operatorname{OR}}}

\newcommand{\size}[1]{{{\operatorname{size}}({#1})}}
\DeclareMathOperator{\signdegree}{\operatorname{sign-degree}}

\newcounter{sprows}

\newlength{\spheight}
\newlength{\spraise}

\newcommand{\comment}[1]{\emph{\color{blue}Comment:\color{black} #1}} % use for simply removing comments
\newlength{\commentslength}
\newcommand{\comments}[1]{
\hspace{-2\parindent}
\addtolength{\commentslength}{-\commentslength}
\addtolength{\commentslength}{\linewidth}
\addtolength{\commentslength}{-\parindent}
\fcolorbox{blue}{white}{\smallskip\begin{minipage}[c]{\commentslength}
\emph{Comments:}\begin{itemize}#1\end{itemize}\end{minipage}}\bigskip
}
\renewcommand{\comment}[1]{}\renewcommand{\comments}[1]{}
\newcommand{\rem}[1]{}

\newtheorem{theorem}{Theorem}[section]
\newtheorem{lemma}[theorem]{Lemma}
\newtheorem{corollary}[theorem]{Corollary}
\newtheorem{claim}[theorem]{Claim}
\newtheorem{proposition}[theorem]{Proposition}
\newtheorem{conjecture}[theorem]{Conjecture}

\newtheorem{definition}[theorem]{Definition}

\newfont{\subsubsecfnt}{ptmri8t at 10pt}
\renewcommand{\subparagraph}[1]{\smallskip{\subsubsecfnt #1.}}

\numberwithin{equation}{section} % makes Eq. numbers (section.number)

\newcommand{\eqnref}[1]{\hyperref[#1]{{(\ref*{#1})}}}
\newcommand{\thmref}[1]{\hyperref[#1]{{Theorem~\ref*{#1}}}}
\newcommand{\lemref}[1]{\hyperref[#1]{{Lemma~\ref*{#1}}}}
\newcommand{\corref}[1]{\hyperref[#1]{{Corollary~\ref*{#1}}}}
\newcommand{\defref}[1]{\hyperref[#1]{{Definition~\ref*{#1}}}}
\newcommand{\secref}[1]{\hyperref[#1]{{Section~\ref*{#1}}}}
\newcommand{\figref}[1]{\hyperref[#1]{{Figure~\ref*{#1}}}}
\newcommand{\tabref}[1]{\hyperref[#1]{{Table~\ref*{#1}}}}
\newcommand{\remref}[1]{\hyperref[#1]{{Remark~\ref*{#1}}}}
\newcommand{\appref}[1]{\hyperref[#1]{{Appendix~\ref*{#1}}}}
\newcommand{\claimref}[1]{\hyperref[#1]{{Claim~\ref*{#1}}}}
\newcommand{\propref}[1]{\hyperref[#1]{{Proposition~\ref*{#1}}}}
\newcommand{\exampleref}[1]{\hyperref[#1]{{Example~\ref*{#1}}}}
\newcommand{\conjref}[1]{\hyperref[#1]{{Conjecture~\ref*{#1}}}}

\newcommand{\threshold}[2]{{T_{{#1}}^{{#2}}}}
\newcommand{\interval}[3]{{I_{{#1},{#2}}^{{#3}}}}

\allowdisplaybreaks[1]
\begin{document}

\title{Span programs and quantum query complexity: \\The general adversary bound is nearly tight\\ for every boolean function}
\author{%
Ben W.~Reichardt%
  \thanks{School of Computer Science and Institute for Quantum Computing, University of Waterloo.}
}
\date{}

\maketitle

\begin{abstract}
The general adversary bound is a semi-definite program (SDP) that lower-bounds the quantum query complexity of a function.  We turn this lower bound into an upper bound, by giving a quantum walk algorithm based on the dual SDP that has query complexity at most the general adversary bound, up to a logarithmic factor.  

In more detail, the proof has two steps, each based on ``span programs," a certain linear-algebraic model of computation.  First, we give an SDP that outputs for any boolean function a span program computing it that has optimal ``witness size."  The optimal witness size is shown to coincide with the general adversary lower bound.  Second, we give a quantum algorithm for evaluating span programs with only a logarithmic query overhead on the witness size.  

\smallskip

The first result is motivated by a quantum algorithm for evaluating composed span programs.  The algorithm is known to be optimal for evaluating a large class of formulas.  The allowed gates include all constant-size functions for which there is an optimal span program.  So far, good span programs have been found in an ad hoc manner, and the SDP automates this procedure.  Surprisingly, the SDP's value equals the general adversary bound.  A corollary is an optimal quantum algorithm for evaluating ``balanced" formulas over any finite boolean gate set.  

The second result broadens span programs' applicability beyond the formula evaluation problem.  We extend the analysis of the quantum algorithm for evaluating span programs.  The previous analysis shows that a corresponding bipartite graph has a large spectral gap, but only works when applied to the composition of constant-size span programs.  We show generally that properties of eigenvalue-zero eigenvectors in fact imply an ``effective" spectral gap around zero.  

\smallskip

A strong universality result for span programs follows.  A good quantum query algorithm for a problem implies a good span program, and vice versa.  Although nearly tight, this equivalence is nontrivial.  Span programs are a promising model for developing more quantum algorithms.   
\end{abstract}

\clearpage

\tableofcontents

\clearpage

\section{Introduction} \label{s:introduction}

Quantum algorithms for evaluating formulas have developed rapidly since the breakthrough AND-OR formula-evaluation algorithm~\cite{fgg:and-or}.  The set of allowed gates in the formula has increased from just AND and OR gates to include all boolean functions on up to three bits, e.g., the three-majority function, and many four-bit functions---with certain technical balance conditions.  Operationally, these new algorithms can be interpreted as evaluating ``span programs," a certain linear-algebraic computational model~\cite{KarchmerWigderson93span}.  Discovering an optimal span program for a function immediately allows it to be added to the gate set~\cite{ReichardtSpalek08spanprogram}.  

This paper is motivated by three main puzzles: 
\begin{enumerate}
\item
Can the gate set allowed in the formula-evaluation algorithm be extended further?  Given that the search for optimal span programs has been entirely ad hoc, yet still quite successful, it seems that the answer must be yes.  How far can it be extended, though?
\item
What is the relationship between span program complexity, or ``witness size," and the adversary lower bounds on quantum query complexity?  There are two different adversary bounds, $\ADV \leq \ADVpm$, but the power of the latter is not fully understood.  Span program witness size appears to be closely connected to these bounds.  For example, so far all known optimal span programs are for functions $f$ with $\ADV(f) = \ADVpm(f)$.  
\item
Aside from their applications to formula evaluation, can span programs be used to derive other quantum algorithms?  
\end{enumerate}

Our first result answers the first two questions.  Unexpectedly, we find that for any boolean function~$f$, the optimal span program has witness size equal to the general adversary bound $\ADVpm(f)$.  This result is surprising because of its broad scope.  It allows us to optimally evaluate formulas over any finite gate set, quantumly.  Classically, optimal formula-evaluation algorithms are known only for a limited class of formulas using AND and OR gates, and a few other special cases.  

This result suggests a new technique for developing quantum algorithms for other problems.  Based on the adversary lower bound, one can construct a span program, and hopefully turn this into an algorithm, i.e., an upper bound.  Unfortunately, it has not been known how to evaluate general span programs.  The second result of this paper is a quantum algorithm for evaluating span programs, with only a logarithmic query overhead on the witness size.  The main technical difficulty is showing that a corresponding bipartite graph has a large spectral gap.  We show that properties of eigenvalue-zero eigenvectors in fact imply an ``effective" spectral gap around zero.  

In combination, the two results imply that the general quantum adversary bound, $\ADVpm$, is tight up to a logarithmic factor for every boolean function.  This is surprising because $\ADVpm$ is closely connected to the nonnegative-weight adversary bound $\ADV$, which has strong limitations.  The results also imply that quantum computers, measured by query complexity, and span programs, measured by witness size, are equivalent computational models, up to a logarithmic factor.  

\smallskip

Some further background material is needed to place the results in context.  

\subsection*{Quantum algorithms for evaluating formulas}

Farhi, Goldstone and Gutmann in 2007 gave a nearly optimal quantum query algorithm for evaluating balanced binary AND-OR formulas~\cite{fgg:and-or, ccjy:and-or}.  This was extended by Ambainis et al.\ to a nearly optimal quantum algorithm for evaluating all AND-OR formulas, and an optimal quantum algorithm for evaluating ``approximately balanced" AND-OR formulas~\cite{AmbainisChildsReichardtSpalekZhang07andor}.  

Reichardt and {\v S}palek gave an optimal quantum algorithm for evaluating ``adversary-balanced" formulas over a considerably extended gate set~\cite{ReichardtSpalek08spanprogram}, including in particular: 
\begin{itemize}
\item All functions $\{0,1\}^n \rightarrow \{0,1\}$ for $n \leq 3$, such as AND, OR, PARITY and MAJ${}_3$.  
\item $69$ of the $92$ inequivalent functions $f: \{0,1\}^4 \rightarrow \{0,1\}$ with $\ADV(f) = \ADVpm(f)$ (\defref{t:adversarydef}).
\end{itemize}

They derived this result by generalizing the previous approaches to consider span programs, a computational model introduced by Karchmer and Wigderson~\cite{KarchmerWigderson93span}.  They then derived a quantum algorithm for evaluating certain concatenated span programs, with a query complexity upper-bounded by the span program witness size (\defref{t:wsizedef}).  Thus in fact the allowed gate set includes all functions $f: \{0,1\}^n \rightarrow \{0,1\}$, with $n = O(1)$, for which we have a span program $P$ computing $f$ and with witness size $\wsize{P} = \ADVpm(f)$ (\defref{t:adversarydef}).  A special case of~\cite[Theorem~4.7]{ReichardtSpalek08spanprogram} is: 

\begin{theorem}[\cite{ReichardtSpalek08spanprogram}] \label{t:spanprogramalgorithm}
Fix a function $f : \{0,1\}^n \rightarrow \{0,1\}$.  For $k \in \N$, define $f^k : \{0,1\}^{n^k} \rightarrow \{0,1\}$ as follows: $f^1 = f$ and $f^k(x) = f \big(f^{k-1}(x_1,\ldots, x_{n^{k-1}}), \ldots, f^{k-1}(x_{n^k-n^{k-1}+1}, \ldots, x_{n^k}) \big)$ for $k > 1$.  If span program $P$ computes $f$, then 
\begin{equation}
Q(f^k) = O(\wsize{P}^k)
 \enspace ,
\end{equation}
where $Q(f^k)$ is the bounded-error quantum query complexity of $f^k$.  
\end{theorem}

\cite{ReichardtSpalek08spanprogram} followed an ad hoc approach to finding optimal span programs for various functions.  Although successful so far, continuing this method seems daunting: 
\begin{itemize}
\item
For most functions $f$, probably $\ADVpm(f) > \ADV(f)$.  Indeed, there are $222$ four-bit boolean functions, up to the natural equivalences, and for only $92$ of them does $\ADVpm = \ADV$ hold.  For no function with a gap has a span program matching $\ADVpm(f)$ been found.  This suggests that perhaps span programs can only work well for the rare cases when $\ADVpm = \ADV$.  
\item
Moreover, for all the functions for which we know an optimal span program, it turns out that an optimal span program can be built just by using AND and OR gates with optimized weights.  (This fact has not been appreciated; see \appref{s:thresholdappendix}.)  On the other hand, there is no reason to think that optimal span programs will in general have such a limited form.  
\item
Finally, it can be difficult to prove a span program's optimality.  For several functions, we have found span programs whose witness sizes match $\ADV$ numerically, but we lack a proof.  
\end{itemize}

In any case, the natural next step is to try to automate the search for good span programs.  A main difficulty is that there is considerable freedom in the span program definition, e.g., span programs are naturally continuous, not discrete.  The search space needs to be narrowed down.  

We show that it suffices to consider span programs written in so-called ``canonical" form.  This form was introduced by~\cite{KarchmerWigderson93span}, but its significance for developing quantum algorithms was not at first appreciated.  We then find a semi-definite program (SDP) for varying over span programs written in canonical form, optimizing the witness size.  This automates the search for span programs.  

Remarkably, the SDP has a value that corresponds exactly to the general adversary bound $\ADVpm$, in a new formulation.  Thus we characterize optimal span program witness size: 

\begin{theorem} \label{t:spanprogramSDPintro}
For any function $f : \{0,1\}^n \rightarrow \{0,1\}$,
\begin{equation} %\label{e:spanprogramSDP}
\inf_{P} \wsize P  = \ADVpm(f)
 \enspace ,
\end{equation}
where the infimum is over span programs $P$ computing $f$.  Moreover, this infimum is achieved.  
\end{theorem}

This result greatly extends the gate set over which the formula-evaluation algorithm of~\cite{ReichardtSpalek08spanprogram} works optimally.  In fact, it allows the algorithm to run on formulas with any finite gate set.  A factor is lost that depends on the gates, but for a finite gate set, this will be a constant.  As another corollary, \thmref{t:spanprogramSDPintro} also settles the question of how the general adversary bound behaves under function composition, and it implies a new upper bound on the sign-degree of boolean functions.  

\subsection*{Quantum algorithm for evaluating span programs}

Now that we know there are span programs with witness size matching the general adversary bound, it is of considerable interest to extend the formula-evaluation algorithm to evaluate arbitrary span programs.  Unfortunately, though, a key theorem from~\cite{ReichardtSpalek08spanprogram} does not hold general span programs.  

The~\cite{ReichardtSpalek08spanprogram} algorithm works by plugging together optimal span programs for the individual gates in a formula $\varphi$ to construct a composed span program $P$ that computes $\varphi$.  Then a family of related graphs $G_P(x)$, one for each input $x$, is constructed.  For an input $x$, the algorithm starts at a particular ``output vertex" of the graph, and runs a quantum walk for about $1/\wsize P$ steps in order to compute $\varphi(x)$.  The algorithm's analysis has two parts.  First, for completeness, it is shown that when $\varphi(x) = 1$, there exists an eigenvalue-zero eigenvector of the weighted adjacency matrix $A_{G_P(x)}$ with large support on the output vertex.  Second, for soundness, it is shown that if $\varphi(x) = 0$, then $A_{G_P(x)}$ has a spectral gap of $\Omega(1/\wsize P)$ for eigenvectors supported on the output vertex.  This spectral gap determines the algorithm's query complexity.  

The completeness step of the proof comes from relating the definition of $G_P(x)$ to the witness size definition.  Eigenvalue-zero eigenvectors correspond exactly to span program ``witnesses," with the squared support on the output vertex corresponding to the witness size.  This argument straightforwardly extends to arbitrary span programs.  

For soundness, the proof essentially inverts the matrix $A_{G_P(x)} - \rho \identity$ gate by gate, span program by span program, starting at the inputs and working recursively toward the output vertex.  In this way, it roughly computes the Taylor series about $\rho = 0$ of the eigenvalue-$\rho$ eigenvectors in order eventually to find a contradiction for $\abs \rho$ small.  One would not expect this method to extend to arbitrary span programs, because it loses a constant factor that depends badly on the individual span programs used for each gate.  Indeed, it fails in general.  Span programs can be constructed for which the associated graphs simply do not have an $\Omega(1/\wsize P)$ spectral gap in the $0$ case.  (For example, take a large span program and add an AND gate to the top whose other input is $0$.  The composed span program computes the constant $0$ function and has constant witness size, but the spectral gaps of the associated large graphs need not be $\Omega(1)$.)  

On the other hand, it has not been understood why the~\cite{ReichardtSpalek08spanprogram} analysis works so well when applied to balanced compositions of constant-size optimal span programs.  In particular, the correspondence between graphs and span programs by definition relates the witness size to properties of eigenvalue-zero eigenvectors.  Why does the witness size quantity also appear in the spectral gap?   

We show that this is not a coincidence, that in general an eigenvalue-zero eigenvector of a bipartite graph implies an ``effective" spectral gap for a perturbed graph.  Somewhat more precisely, the inference is that the total squared overlap on the output vertex of small-eigenvalue eigenvectors is small.  This argument leads to a substantially more general small-eigenvalue spectral analysis.  It also implies simpler proofs of \thmref{t:spanprogramalgorithm} as well as of the AND-OR formula-evaluation result in~\cite{AmbainisChildsReichardtSpalekZhang07andor}.  

This small-eigenvalue analysis is the key step that allows us to evaluate span programs on a quantum computer.  Besides showing an effective spectral gap, though, we would also need to bound $\norm{A_{G_P}}$ in order to generalize~\cite{ReichardtSpalek08spanprogram}.  However, recent work by Cleve et al.\ shows that this norm does not matter if we are willing to concede a logarithmic factor in the query complexity~\cite{CleveGottesmanMoscaSommaYongeMallo08discretize}.  We thus obtain: 

\begin{theorem} \label{t:generalspanprogramalgorithmintro}
Let $P$ be a span program computing $f: \{ 0, 1\}^n \rightarrow \{0, 1\}$.  Then
\begin{equation}
Q(f) = O\bigg( \wsize P \frac{\log \wsize P}{\log \log \wsize P} \bigg)
 \enspace .
\end{equation}
\end{theorem}

We can now prove the main result of this paper, that for any boolean function $f$ the general adversary bound on the quantum query complexity is tight up to a logarithmic factor: 

\begin{theorem} \label{t:querycomplexitytightintro}
For any function $f : \{0,1\}^n \rightarrow \{0,1\}$, the quantum query complexity of $f$ satisfies 
\begin{equation}
Q(f) = \Omega(\ADVpm(f)) \quad \text{and} \quad Q(f) = O\bigg(\ADVpm(f) \, \frac{\log \ADVpm(f)}{\log \log \ADVpm(f)}\bigg)
 \enspace .
\end{equation}
\end{theorem}

\begin{proof}
The lower bound is due to~\cite{HoyerLeeSpalek07negativeadv} (see \thmref{t:advquerycomplexity}).  For the upper bound, use the SDP from \thmref{t:spanprogramSDPintro}, to construct a span program $P$ computing $f$, with $\wsize P = \ADVpm(f)$.  Then apply \thmref{t:generalspanprogramalgorithmintro} to obtain a bounded-error quantum query algorithm that evaluates $f$.  
\end{proof}

Thus the $\ADVpm$ semi-definite program is in fact an SDP for quantum query complexity, up to a logarithmic factor.  Previously, Barnum et al.\ have already given an SDP for quantum query complexity~\cite{BarnumSaksSzegedy03adv}, and have shown that the nonnegative-weight adversary bound $\ADV$ can be derived by strengthening it, but their SDP is quite different.  In particular, the $\ADVpm$ SDP is ``greedy," in the sense that it considers only how much information can be learned using a single query; see \defref{t:adversarydef} below.  The \cite{BarnumSaksSzegedy03adv} SDP, on the other hand, has separate terms for every query.  It is surprising that a small modification to $\ADV$ can not only break the certiÞcate complexity and property testing barriers~\cite{HoyerLeeSpalek07negativeadv}, but in fact be nearly optimal always.  For example, for the Element Distinctness problem with the input in $[n]^n$ specified in binary, $\ADV(f) = O(\sqrt n \log n)$~\cite{SpalekSzegedy04advequivalent} but $Q(f) = \Omega(n^{2/3})$ by the polynomial method~\cite{AaronsonShi04collisioned, Ambainis05polynomialmethod}.  \thmref{t:querycomplexitytightintro} implies that $\ADVpm(f) = \Omega(n^{2/3} / \log n)$.

\section{Definitions} \label{s:definitions}

For a natural number $n \in \N$, let $[n] = \{1, 2, \ldots, n\}$.  Let $\B = \{0,1\}$.  For a bit $b \in \B$, let $\bar b = 1-b$ denote its complement.  
A function $f$ with codomain $\B$ is a (total) boolean function if its domain is $\B^n$ for some $n \in \N$; $f$ is a partial boolean function if its domain is a subset $\D \subseteq \B^n$.  

The complex and real numbers are denoted by $\C$ and $\R$, respectively.  For a finite set $X$, let $\C^X$ be the inner product space $\C^{\abs X}$ with orthonormal basis $\{ \ket x : x \in X \}$.  We assume familiarity with ket notation, e.g., $\sum_{x \in X} \ketbra x x = \identity$ the identity on $\C^X$.  For vector spaces $V$ and $W$ over $\C$, let $\L(V, W)$ denote the set of all linear transformations from $V$ into $W$, and let $\L(V) = \L(V, V)$.  For $A \in \L(V, W)$, $\norm{A}$ is the operator norm of $A$.  

The union of disjoint sets is sometimes denoted by $\sqcup$.  

In the remainder of this section, we will define span programs, from~\cite{KarchmerWigderson93span}, and the ``witness size" span program complexity measure from~\cite{ReichardtSpalek08spanprogram}.  We will then define the quantum adversary bounds and state some of their basic properties, including composition, lower bounds on quantum query complexity, and the previously known lower bound on span program witness size.

\subsection{Span programs} \label{s:spanprogramdef}

A span program $P$ is a certain linear-algebraic way of specifying a boolean function $f_P$~\cite{KarchmerWigderson93span, GalPudlak03spanprogram}.  Roughly, a span program consists of a target $\ket t$ in a vector space $V$, and a collection of subspaces $V_{j,b} \subseteq V$, for $j \in [n]$, $b \in B$.  For an input $x \in \B^n$, $f_P(x) = 1$ when the target can be reached using a linear combination of vectors in $\cup_{j \in [n]} V_{j, x_j}$.  For our complexity measure on span programs, however, it will be necessary to fix a set of ``input vectors" that span 
each subspace $V_{j, b}$.  We desire to span the target using a linear combination of these vectors with small coefficients.  

Formally we therefore define a span program as follows: 

\begin{definition}[Span program~\cite{KarchmerWigderson93span}] \label{t:spanprogramdef}
Let $n \in \N$.  A span program $P$ consists of a ``target" vector $\ket t$ in a finite-dimensional inner-product space $V$ over $\C$, together with ``input" vectors $\ket{ v_i } \in V$ for $i \in I$.  Here the index set $I$ is a disjoint union $I = \Ifree \sqcup \bigsqcup_{j \in [n], b \in \B} I_{j,b}$.  

To $P$ corresponds a function $f_P : \B^n \rightarrow \B$, defined by 
\begin{equation} \label{e:spanprogramdef}
f_P(x) = \begin{cases}
1 & \text{if $\ket t \in \Span(\{ \ket{ v_i } : i \in \Ifree \cup \bigcup_{j \in [n]} I_{j, x_j} \})$} \\ 
0 & \text{otherwise}
\end{cases}
\end{equation}
\end{definition}

We say that $\Ifree$ indexes the set of ``free" input vectors, while $I_{j,b}$ indexes input vectors ``labeled by" $(j,b)$.  We say that $P$ ``computes" the function $f_P$.  For $x \in \B^n$, $f_P(x)$ evaluates to $1$, or true, when the target can be reached using a linear combination of the ``available" input vectors, i.e., input vectors that are either free or labeled by $(j, x_j)$ for $j \in [n]$. 

Some additional notation will come in handy.  Let $\{ \ket i : i \in I \}$ be an orthonormal basis for $\C^{\abs I}$.  Let $A : \C^{\abs I} \rightarrow V$ be the linear operator 
\begin{equation}
A = \sum_{i \in I} \ketbra{v_i}{i}
 \enspace .
\end{equation}
Written as a matrix, the columns of $A$ are the input vectors of $P$.  For an input $x \in \B^n$, let $I(x)$ be the set of available input vector indices and $\Pi(x) : \C^{\abs I} \rightarrow \C^{\abs I}$ the projection thereon, 
\begin{align}
I(x) &= \Ifree \cup \bigcup_{j \in [n]} I_{j, x_j} \\
\Pi(x) &= \sum_{i \in I(x)} \ketbra i i
 \enspace .
\end{align}

\begin{lemma} \label{t:witnesses}
For a span program $P$, $f_P(x) = 1$ if and only if $\ket t \in \Range(A \Pi(x))$.  Equivalently, $f_P(x) = 0$ if and only if $\Pi(x) A^\dagger \ket t \in \Range\!\Big[ \Pi(x) A^\dagger \Big(\identity - \frac{\ketbra t t}{\norm{t}^2}\Big)\Big]$.  
\end{lemma}

\lemref{t:witnesses} follows from Eq.~\eqnref{e:spanprogramdef}.  Therefore exactly when $f_P(x) = 1$ is there a ``witness" $\ket w \in \C^{\abs I}$ satisfying $A \Pi(x) \ket w = \ket t$.  Exactly when $f_P(x) = 0$, there is a witness $\ket{w'} \in V$ satisfying $\braket{t}{w'} \neq 0$ and $\Pi(x) A^\dagger \ket{w'} = 0$, i.e., $\ket{w'}$ has nonzero inner product with the target vector and is orthogonal to the available input vectors.  

The complexity measure we use to characterize span programs is the witness size~\cite{ReichardtSpalek08spanprogram}:

\begin{definition}[Witness size with costs~\cite{ReichardtSpalek08spanprogram}] \label{t:wsizedef}
Consider a span program $P$, and a vector $s \in [0, \infty)^n$ of nonnegative ``costs."  Let $S = \sum_{j \in [n], b \in \B, i \in I_{j,b}} \sqrt{s_j} \ketbra i i$.  For each input $x \in \B^n$, define the witness size of $P$ on $x$ with costs $s$, $\wsizexS P x s$, as follows: 
\begin{itemize}
\item If $f_P(x) = 1$, then $\ket{t} \in \Range(A \Pi(x))$, so there is a witness $\ket{w} \in \C^{\abs I}$ satisfying $A \Pi(x) \ket{w} = \ket{t}$.  Then $\wsizexS P x s$ is the minimum squared length of any such witness, weighted by the costs $s$: 
\begin{equation} \label{e:wsizetrue}
\wsizexS P x s 
= \min_{\ket w : \, A \Pi(x) \ket w = \ket t} \norm{S \ket w}^2
 \enspace .
\end{equation}
\item If $f_P(x) = 0$, then $\ket t \notin \Range(A \Pi(x))$.
Therefore there is a witness $\ket{w'} \in V$ satisfying $\braket{t}{w'} = 1$ and $\Pi(x) A^\adjoint \ket{w'} = 0$.  Then 
\begin{equation} \label{e:wsizefalse}
\wsizexS P x s 
= \min_{\substack{\ket{w'} : \, \braket{t}{w'} = 1 \\ \Pi(x) A^\adjoint \ket{w'} = 0}} \norm{S A^\adjoint \ket{w'}}{}^2
 \enspace .
\end{equation}
\end{itemize}

The witness size of $P$ with costs $s$, restricted to domain $\D \subseteq \B^n$, is 
\begin{equation}
\wsizeSD P s =  \max_{x \in \D} \wsizexS P x s
 \enspace .
\end{equation}
\end{definition}

The $\wsizeSD P s$ notation is for handling partial boolean functions.  For the common case that $\D = \B^n$, let $\wsizeS P s = \wsizeS {P, \B^n} s$.  For $j \in [n]$, $s_j$ can intuitively be thought of as the charge for evaluating the $j$th input bit.  When the subscript $s$ is omitted, the costs are taken to be uniform, $s = \vec 1 = (1, 1, \ldots, 1)$, e.g., $\wsize P = \wsizeS P {\vec 1}$.  In this case, note that $S = \identity - \sum_{i \in \Ifree} \ketbra i i$.  The extra generality of allowing nonuniform costs is necessary for considering unbalanced formulas.  

Before continuing, let us remark that the above definition of span programs differs slightly from the original definition due to Karchmer and Wigderson~\cite{KarchmerWigderson93span}.  Call a span program \emph{strict} if $\Ifree = \emptyset$.  Ref.~\cite{KarchmerWigderson93span} considers only strict span programs.  For the witness size complexity measure, we will later prove that span programs and strict span programs are equivalent (\propref{t:spanprogramrelaxed}).  Allowing free input vectors is often convenient for defining and composing span programs, though, and may be necessary for developing efficient quantum algorithms based on span programs.  Ref.~\cite{ReichardtSpalek08spanprogram} uses an even more relaxed span program definition than \defref{t:spanprogramdef}, letting each input vector to be labeled by a subset of $[n] \times B$.  This definition is convenient for terse span program constructions, and is also easily seen to be equivalent to ours.  

Classical applications of span programs have used a different complexity measure, the ``size" of $P$ being the number of input vectors, $\abs I$.  This measure has been characterized in~\cite{Gal01spanprogramsize}.  

Note that replacing the target vector $\ket t$ by $c \ket t$, for $c \neq 0$, changes the witness sizes by a factor of $\abs{c}^2$ or $1/\abs{c}^2$, depending on whether $f_P(x) = 1$ or $0$.  Thus we might just as well have defined the witness size as 
\begin{equation} \label{e:wsizescalet}
\sqrt{ \max_{x : f_P(x) = 0} \wsizexS P x s \max_{x : f_P(x) = 1} \wsizexS P x s } \enspace ,
\end{equation}
provided that $f_P$ is not the constant $0$ or constant $1$ function on $\D$.  Explicit formulas for $\wsizexS P x s$ can be written in terms of Moore-Penrose pseudoinverses of certain matrices, and are given in~\cite[Lemma~A.3]{ReichardtSpalek08spanprogram}.  \thmref{t:generalspanprogramalgorithmnonblackbox} will give an alternative, related criterion for comparing span programs.

\subsection{Adversary lower bounds} \label{s:adversarydef}

There are essentially two techniques, the polynomial and adversary methods, for lower-bounding quantum query complexity.  The polynomial method was introduced in the quantum setting by Beals et al.~\cite{BealsBuhrmanCleveMoscaWolf98}.  It is based on the observation that after running a quantum algorithm for $q$ oracle queries to an input $x$, the probability of any measurement result is a polynomial of degree at most $2q$ in the variables~$x_j$.  The first of the adversary bounds, $\ADV$, was introduced by Ambainis~\cite{Ambainis00adversary}.  Adversary bounds are a generalization of the classical hybrid argument, that considers the entanglement of the system when run on a superposition of input strings.  Both methods have classical analogs; see~\cite{Beigel93polynomialclassical} and~\cite{Aaronson03localsearch}

The polynomial method and $\ADV$ are incomparable.  {\v S}palek and Szegedy~\cite{SpalekSzegedy04advequivalent} proved the equivalence of a number of formulations for the adversary bound $\ADV$, and also showed that $\ADV$ is subject to a certificate complexity barrier.  For example, for $f$ a total boolean function, $\ADV(f) \leq \sqrt{C_0(f) C_1(f)}$, where $C_b(f)$ is the best upper bound over those $x$ with $f(x) = b$ of the size of the smallest certificate for $f(x)$.  The polynomial method can surpass this barrier.  In particular, for the Element Distinctness problem, the polynomial method implies an $\Omega(n^{2/3})$ lower bound on the quantum query complexity~\cite{AaronsonShi04collisioned, Ambainis05polynomialmethod}, and this is tight~\cite{Ambainis03elementdistinctness, Szegedy04walkfocs}.  However, displaying two list elements that are the same is enough to prove that the list does not have distinct elements, so $C_0(f) = 2$ and $\ADV(f) = O(\sqrt n)$.  $\ADV$ also suffers a ``property testing barrier" on partial functions.  

On the other hand, the polynomial method can also be loose.  Ambainis gave a total boolean function $f^k$ on $n = 4^k$ bits that can be represented exactly by a polynomial of degree only $2^k$, but for which $\ADV(f^k) = 2.5^k$~\cite{Ambainis06polynomial}, and see~\cite{HoyerLeeSpalek07negativeadv} for other examples.  

Thus both lower bound methods are limited.  In 2007, though, H{\o}yer et al.\ discovered a strict generalization $\ADVpm$ of $\ADV$~\cite{HoyerLeeSpalek07negativeadv}.  For example, for Ambainis's function, $\ADVpm(f^k) \geq 2.51^k$.  $\ADVpm$ also breaks the certificate complexity and property testing barriers.  No similar limits on its power have been found.  In particular, for no function $f$ is it known that the quantum query complexity of $f$ is $\omega(\ADVpm(f))$.  

In this section, we define the two adversary bounds.  On account of how their definitions differ, we call $\ADV$ the ``nonnegative-weight" adversary bound, and $\ADVpm$ the ``general" adversary bound.  We also state some previous results.  

\begin{definition}[Adversary bounds with costs~\cite{HoyerLeeSpalek05compose, HoyerLeeSpalek07negativeadv}] \label{t:adversarydef}
For finite sets $C$ and $E$, and $\D \subseteq C^n$, let $f: \D \rightarrow E$ and let $s \in [0, \infty)^n$ be a vector of nonnegative costs.  An adversary matrix for $f$ is a nonzero, $\abs{\D} \times \abs{\D}$ real, symmetric matrix $\Gamma$ that satisfies $\bra x \Gamma \ket y = 0$ for all $x, y \in \D$ with $f(x) = f(y)$.  

Define the nonnegative-weight adversary bound for $f$, with costs $s$, as 
\begin{equation} \label{e:adversarydef}
\ADV_s(f) = \max_{\large \substack{\text{adversary matrices $\Gamma$:} \\
\forall x, y \in \D, \, \bra x \Gamma \ket y \geq 0 \\
\forall j \in [n], \, \norm{\Gamma \circ \Delta_j} \leq s_j}} \norm{\Gamma}
 \enspace ,
\end{equation}
where $\Gamma \circ \Delta_j$ denotes the entry-wise matrix product between $\Gamma$ and $\Delta_j = \sum_{x, y \in \D : x_j \neq y_j} \ketbra x y$, and the norm is the operator norm.  

The general adversary bound for $f$, with costs $s$, is 
\begin{equation}
\ADVpm_s(f) = \max_{\large \substack{\text{adversary matrices $\Gamma$:} \\
\forall j \in [n], \, \norm{\Gamma \circ \Delta_j} \leq s_j}} \norm{\Gamma}
 \enspace .
\end{equation}
In this maximization, the entries of $\Gamma$ need not be nonnegative.  In particular, $\ADVpm_s(f) \geq \ADV_s(f)$.  

Letting $\vec{1} = (1, 1, \ldots, 1)$, the nonnegative-weight adversary bound for $f$ is $\ADV(f) = \ADV_{\vec 1}(f)$ and the general adversary bound for $f$ is $\ADVpm(f) = \ADVpm_{\vec 1}(f)$.  
\end{definition}

One special case is when $s_{j^*} = 0$ for some $j^* \in [n]$.  In this case, since $\Gamma \circ \Delta_{j^*}$ must be zero, letting $s' = (s_1, \ldots, \widehat{s_{j^*}}, \ldots, s_n)$ and $f_b$ be the restriction of $f$ to inputs $x$ with $x_{j^*} = b$, we have $\ADV_s(f) = \max_{b \in C} \ADV_{s'}(f_b)$ and $\ADVpm_s(f) = \max_{b \in C} \ADVpm_{s'}(f_b)$.  Provided $s_j > 0$ for all $j \in [n]$, we can write 
\begin{align}
\ADV_s(f) 
&= \max_{\substack{\text{adversary matrices $\Gamma$:} \\
\forall x, y \in \D, \, \bra x \Gamma \ket y \geq 0}} \min_{j \in n} s_j \frac{\norm{\Gamma}}{\norm{\Gamma \circ \Delta_j}} \\
\ADVpm_s(f) 
&= \max_{\text{adversary matrices $\Gamma$}} \min_{j \in n} s_j \frac{\norm{\Gamma}}{\norm{\Gamma \circ \Delta_j}} 
 \enspace ,
\end{align}
which are the expressions used in Refs.~\cite{HoyerLeeSpalek05compose, HoyerLeeSpalek07negativeadv}.  Furthermore, \thmref{t:adversarydual} and \thmref{t:adversarydualgeneral} will state dual semi-definite programs for $\ADV$ and $\ADVpm$. 

The adversary bounds are primarily of interest because, with uniform costs $s = \vec 1$, they give lower bounds on quantum query complexity.  

\begin{definition}
For $f : \D \rightarrow E$, with $\D \subseteq C^n$, let $Q_\epsilon(f)$ be the $\epsilon$-bounded-error quantum query complexity of $f$, $Q(f) = Q_{1/10}(f)$, and, when $E = \{0,1\}$, let $Q^1(f)$ be the one-sided bounded-error quantum query complexity.  
\end{definition}

\begin{theorem}[\cite{BarnumSaksSzegedy03adv, HoyerLeeSpalek07negativeadv}] \label{t:advquerycomplexity}
For any function $f : \D \rightarrow E$,  with $\D \subseteq C^n$, the $\epsilon$-bounded-error quantum query complexity of $f$ is lower-bounded as 
\begin{equation}\begin{split}
Q_\epsilon(f) 
&\geq \frac{1 - 2\sqrt{\epsilon(1-\epsilon)}}{2} \ADV(f) \\
Q_\epsilon(f) 
&\geq \frac{1 - 2\sqrt{\epsilon(1-\epsilon)} - 2 \epsilon}{2} \ADVpm(f)
 \enspace .
\end{split}\end{equation}
In particular, $Q(f) = \Omega(\ADVpm(f))$.  
Moreover, if $D = \{0,1\}$, then 
\begin{equation}
Q_\epsilon(f) 
\geq \frac{1 - 2\sqrt{\epsilon(1-\epsilon)}}{2} \ADVpm(f)
 \enspace .
\end{equation}
\end{theorem}

For boolean functions, the nonnegative-weight adversary bound composes multiplicatively, but this was not known to hold for the general adversary bound~\cite{HoyerLeeSpalek07negativeadv}:  

\begin{theorem}[Adversary bound composition {\cite{HoyerLeeSpalek07negativeadv, Ambainis06polynomial, LaplanteLeeSzegedy06adversary, HoyerLeeSpalek05compose}}] \label{t:weakadversarycomposition}
Let $f : \{0,1\}^n \rightarrow \{0,1\}$ and, for $j \in [n]$, let $f_j : \{0,1\}^{m_j} \rightarrow \{0,1\}$.  Define $g : \{0,1\}^{m_1} \times \cdots \times \{0,1\}^{m_n} \rightarrow \{0,1\}$ by 
\begin{equation}
g(x) = f\big(f_1(x_1), \ldots, f_n(x_n)\big)
 \enspace .
\end{equation} 
Let $s \in [0, \infty)^{m_1} \times \cdots \times [0, \infty)^{m_n}$, and let $\alpha_j = \ADV_{s_j}(f_j)$ and $\beta_j = \ADVpm_{s_j}(f_j)$ for $j \in [n]$.  Then 
\begin{align}
\ADV_s(g) &= \ADV_\alpha(f) \\
\ADVpm_s(g) &\geq \ADVpm_\beta(f) \label{e:weakadversarycomposition}
 \enspace .
\end{align}
In particular, if $\ADV_{s_1}(f_1) = \cdots = \ADV_{s_n}(f_n) = \alpha$, then $\ADV_s(g) = \alpha \, \ADV(f)$, and if $\ADVpm_{s_1}(f_1) = \cdots = \ADVpm_{s_n}(f_n) = \beta$, then $\ADVpm_s(g) \geq \beta \, \ADVpm(f)$.  
\end{theorem}

Reichardt and {\v S}palek~\cite{ReichardtSpalek08spanprogram} show that the adversary bounds lower-bound the witness size of a span program:

\begin{theorem}[\cite{ReichardtSpalek08spanprogram}] \label{t:wsizeadvbound}
For any span program $P$ computing $f_P : \{0,1\}^n \rightarrow \{0,1\}$, 
\begin{equation}
\wsize{P} \geq \ADVpm(f_P) \geq \ADV(f_P)
 \enspace .
\end{equation}
\end{theorem}

There is a direct proof that $\wsize{P} \geq \ADV(f_P)$ in \cite[Sec.~5.3]{ReichardtSpalek08spanprogram}, but the inequality $\wsize{P} \geq \ADVpm(f_P)$ is only implicit in~\cite{ReichardtSpalek08spanprogram}.  The argument is as follows.  Letting $f^{k}: \{0,1\}^{n^k} \rightarrow \{0,1\}$ be the $k$-times-iterated composition of $f$ on itself, $Q(f_P^{k}) = O_k(\wsize{P}^k)$ by \thmref{t:spanprogramalgorithm}.  Now by \thmref{t:weakadversarycomposition},  $\ADVpm(f)^k \leq \ADVpm(f^{k}) = O(Q(f_P^{k}))$.  Putting these results together and letting $k \rightarrow \infty$ gives $\ADVpm(f) \leq \wsize{P}$.  A full and direct proof will be given below in \thmref{t:spanprogramSDP}.

\section{Example: Span programs based on one-sided-error quantum query algorithms} \label{s:spanprogramuniversal}

Span programs have proved useful in~\cite{ReichardtSpalek08spanprogram} for evaluating formulas.  There, span programs for constant-size gates are composed to generate a span program for a full formula.  In this section, we give an explicit construction of asymptotically large span programs that are interesting from the perspective of quantum algorithms and that do not arise from the composition of constant-size span programs.  We relate span program witness size to {one-sided} bounded-error quantum query complexity.  \thmref{t:spanprogramuniversalboundederror} below will strengthen the results in this section, but the construction there will be less explicit.  

Formally, we show:

\begin{theorem} \label{t:spanprogramuniversal}
Consider a quantum query algorithm $\algorithm$ that evaluates $f: \{0,1\}^n \rightarrow \{0,1\}$, with bounded one-sided error on false inputs, using $q$ queries.  Then there exists a span program $P$ computing $f_P = f$, with 
\begin{equation}
\wsize{P} = O(q)
 \enspace .
\end{equation}
In particular, $\inf_{P : f_P = f} \wsize{P} = O(Q^1(f))$.  
\end{theorem}

This example should be illustrative for \defref{t:spanprogramdef} and \defref{t:wsizedef}, but is not needed for the rest of this article.  Another nontrivial span program example is given in \appref{s:thresholdappendix}.

Many known quantum query algorithms have one-sided error, as required by \thmref{t:spanprogramuniversal}, or can be trivially modified to have one-sided error.  Examples include algorithms for Search, Ordered Search, Graph Collision, Triangle Finding, and Element Distinctness.  There are exceptions, though.  For example, the formula-evaluation algorithms discussed above and implicit in \thmref{t:spanprogramalgorithm} all have bounded two-sided error.  In particular, for AND-OR formula evaluation, the algorithm from~\cite{AmbainisChildsReichardtSpalekZhang07andor} outputs the formula's evaluation but not a witness to that evaluation~\cite[Sec.~5]{ReichardtSpalek08spanprogram}.  For AND-OR formula evaluation, a witness can be extracted from the $\lambda = 0$ graph eigenstate, but it is not known how far this generalizes~\cite{AmbainisChildsLegallTani09witness}.  We certainly expect that there are functions $f$ with bounded two-sided-error quantum query complexity, $Q(f)$, strictly less than the bounded one-sided-error quantum query complexity, $Q^1(f)$.  

\begin{proof}[Proof of \thmref{t:spanprogramuniversal}]
Assume that the quantum algorithm $\algorithm$ has a workspace of $m$ qubits, and an $n$-dimensional query register.  Starting in the state $\ket{0^m,1}$, it alternates between applying unitaries independent of the input string $x$ and oracle queries to $x$.   The evolution of the system is given by 
\begin{equation}
\begin{array}{rlclcl}
\ket{\varphi_0} = \ket{0^m,1} & \;\, \overset{V_1}{\rightarrow} & \ket{\varphi_1} = \sum_{j=1}^{n} \ket{\varphi_{1,j}} \ket j & \overset{O_x}{\rightarrow} & \ket{\varphi_2} = \sum_{j=1}^{n} (-1)^{x_j} \ket{\varphi_{1,j}} \ket j & \rightarrow \cdots \\
\cdots & \overset{V_{2r-1}}{\rightarrow} & \ket{\varphi_{2r-1}} = \sum_{j=1}^{n} \ket{\varphi_{2r-1,j}} \ket j & \overset{O_x}{\rightarrow} & \ket{\varphi_{2r}} = \sum_{j=1}^{n} (-1)^{x_j} \ket{\varphi_{2r-1,j}} \ket j & \rightarrow \cdots \\
\cdots & \overset{V_{2q+1}}{\rightarrow} & \ket{\varphi_{2q+1}}
\end{array}
\end{equation}
Here, for $r \in [q+1]$, $V_{2r-1}$ is the unitary independent of $x$ that is applied at odd time step $2r-1$, while $O_x : \ket y \ket j \mapsto (-1)^{x_j} \ket y \ket j$ is the phase-flip input oracle applied at even time steps.  (To allow conditional queries, prepend a constant bit $0$ to the input string $x$.)  The state of the system after $s$ time steps is $\ket{\varphi_\tau} = \sum_{j=1}^{n} \ket{\varphi_{s,j}} \otimes \ket j$; for $s \geq 2$, these states depend on $x$.  

On inputs $x$ evaluating to $f(x) = 1$, the algorithm $\algorithm$ does not make errors.  Thus for these $x$ we may assume without loss of generality that $\ket{\varphi_{2q+1}} = \ket{0^m,1}$, by at most doubling the number of queries to clean the algorithm's workspace.  On inputs evaluating to $f(x) = 0$, then, $\abs{\braket{0^m,1}{\varphi_{2q+1}}} \leq \epsilon$ for some $\epsilon$ bounded away from one.  

Recall that $\B = \{0,1\}$.  We construct a span program $P$ as follows: 
\begin{itemize}
\item The inner product space is $V = \C^{(2q+2) 2^m}$, spanned by the orthonormal basis $\{ \ket{s, y, j} : s \in \{0,1,\ldots,2q+1\}, y \in \B^m, j \in [n]\}$.  
\item The target vector is $\ket t = - \ket{0, 0^m, 1} + \ket{2q+1, 0^m, 1}$.   
\item There are free input vectors for each odd time step $s$: $\Ifree = \{ 2r-1 : r \in [q+1] \} \times \B^m \times [n]$, with 
\begin{equation}
\ket{v_{s,y,j}} = -\ket{s-1,y,j} + \ket s \otimes V_s \ket{y, j}
\end{equation}
for $(s, y, j) \in \Ifree$.   
\item For $j \in [n]$ and $b \in \B$, $I_{j,b} = \{ 2r : r \in [q] \} \times \B^m \times [n] \times \{ b \}$, with, for $(s, y, j, b) \in I_{j, b}$, 
\begin{equation}
\ket{v_{s, y, j, b}} = -(\ket{s-1} + (-1)^b \ket s) \otimes \ket{y, j}
 \enspace .
\end{equation}
\end{itemize}

For analyzing this span program, it will be helpful to set up some additional notation.  Let $U_s$ be the unitary applied at time step $s$: 
\begin{equation}
U_s = \begin{cases} 
V_s & \text{if $s$ is odd} \\
O_x & \text{if $s$ is even}
\end{cases}
\end{equation}
For an input $x$, the available input vectors, i.e., those indexed by $I(x)$, are then 
\begin{equation}
\ket{v_{s,y,j}} := -\ket{s,y,j} + \ket{s+1} \otimes U_{s+1} \ket{y, j}
\end{equation}
for all $y \in \B^m$, $j \in [n]$ and $s = 0, 1, \ldots, 2q$ even or odd.  Let $A(x) = \sum_{s=0}^{2q} \sum_{y, j} \ketbra{v_{s,y,j}}{s,y,j}$.  Then $f_P(x) = 1$ if and only if $\ket t \in \Range(A(x))$.  

\begin{claim} \label{t:spanprogramuniversaltrue}
If $f(x) = 1$, then $f_P(x) = 1$ and $\wsizex{P}{x} \leq q$.  
\end{claim}

\begin{proof}
Letting $\ket w = \sum_{s=0}^{2q} \ket s \otimes \ket{\varphi_s}$, then 
\begin{equation}\begin{split}
A(x) \ket w
&= \sum_{s=0}^{2q} \sum_{y, j} \ket{v_{s,y,j}} \braket{y,j}{\varphi_s} \\
&= \sum_{s=0}^{2q} -\ket s \otimes \ket{\varphi_s} + \ket{s+1} \otimes U_{s+1} \ket{\varphi_s} \\
&= \sum_{s=0}^{2q} -\ket s \otimes \ket{\varphi_s} + \ket{s+1} \otimes \ket{\varphi_{s+1}} \\
&= -\ket 0 \otimes \ket{\varphi_0} + \ket{2q+1} \otimes \ket{\varphi_{2q+1}} \\
&= \ket t
 \enspace ,
\end{split}\end{equation}
where we have used for the second equality that $\sum_{y, j} \ketbra{y, j}{y, j}$ is a resolution of the identity, and for the third equality that $U_{s+1} \ket{\varphi_s} = \ket{\varphi_{s+1}}$ in order to get a telescoping series.  Thus $\ket w$ is a witness to $f_P(x) = 1$.  Since the input vectors $\ket{v_{s,y,j}}$ for $s$ even are free, the witness size is 
\begin{align}
\wsizex{P}{x}
&\leq \left\|{ \Bigg(\sum_{k=1}^{q} \ketbra{2k+1}{2k+1} \otimes \identity \Bigg) \ket w }\right\|^2 \\
&= \sum_{k=1}^{q} \norm{\ket{2k+1} \otimes \ket{\varphi_{2k+1}}}^2 \nonumber \\
&= q
 \enspace . \qedhere
\end{align}
\end{proof}

\begin{claim} \label{t:spanprogramuniversalfalse}
If $f(x) = 0$, then $f_P(x) = 0$ and $\wsizex{P}{x} \leq 4q/(1-\epsilon)^2$.  
\end{claim}

\begin{proof}
Let $\ket{w'} = \sum_{s=0}^{2q+1} \ket s \otimes \ket{\varphi_s}$.  Then $\abs{\braket{t}{w'}} = \abs{1 - \braket{0^m,1}{\varphi_{2q+1}}} \geq 1 - \epsilon > 0$.  Moreover, since 
\begin{equation}
\bra{v_{s,y,j}} (\ket \sigma \otimes \ket{\varphi_\sigma}) = \begin{cases}
- \braket{y, j}{\varphi_s} & \text{if $\sigma = s$} \\
\bra{y,j} U_{s+1}^\dagger \ket{\varphi_{s+1}} = \braket{y, j}{\varphi_s} & \text{if $\sigma = s+1$} \\
0 & \text{otherwise}
\end{cases}
\end{equation}  
we compute 
\begin{equation}
A(x)^\dagger \ket{w'} 
= \sum_{s=0}^{2q} \sum_{\sigma=0}^{2q+1} \sum_{y, j} \ketbra{s,y,j}{v_{s,y,j}} (\ket \sigma \otimes \ket{\varphi_\sigma})
= 0
 \enspace .
\end{equation}
Thus $\ket{w'}$ is a witness to $f_P(x) = 0$.  Now the input vectors associated with false inputs are, for odd $s$ between $1$ and $2q-1$, $y \in \{0,1\}^m$ and $j \in [n]$, $\ket{v_{s,y,j}'} := -\ket{s,y,j} - \ket{s+1} \otimes U_{s+1} \ket{y, j}$.  
Now $\braket{v_{s,y,j}'}{w'} = -\braket{y,j}{\varphi_s} - \bra{y,j} U_{s+1}^\dagger \ket{\varphi_{s+1}} = -2 \braket{y,j}{\varphi_s}$.  
The witness size therefore satisfies 
\begin{align}
(1-\epsilon)^2 \wsizex{P}{x} 
&\leq \sum_{k=1}^{q} \sum_{y, j} \abs{\braket{v_{2k-1,y,j}'}{w'}}^2 \\
&= 4 \sum_{k=1}^{q} \sum_{y, j} \abs{\braket{y,j}{\varphi_{2k+1}}}^2 \nonumber \\
&= 4 q
 \enspace . \qedhere
\end{align}
\end{proof}

After scaling the target vector appropriately---see Eq.~\eqnref{e:wsizescalet}---\claimref{t:spanprogramuniversaltrue} and \claimref{t:spanprogramuniversalfalse} together give $\wsize{P} \leq 2 q / (1-\epsilon)$, proving \thmref{t:spanprogramuniversal}.  
\end{proof}

\section{Span program manipulations} \label{s:spanprogrammanipulations}

This section presents several useful manipulations of span programs.  First, we develop span program complementation and composition.  The essential ideas for both manipulations have already been proposed in~\cite{ReichardtSpalek08spanprogram}, but the ideas there were not fully translated into the span program formalism, which we do here.  \secref{s:spanprogramcomposition} also introduces a new construction of composed span programs, tensor-product composition, which appears be useful for designing more efficient quantum algorithms for evaluating formulas~\cite{Reichardt09andorfaster}.  

Both techniques take as inputs span programs computing certain functions and output a span program computing a different function.  In \secref{s:strictrealspanprograms}, we give two ways of simplifying a span program $P$ that do not change $f_P$ nor increase the witness size.  \secref{s:canonicalspanprograms} will present a more dramatic simplification, though.

\subsection{Span program complementation} \label{s:dualspanprogram}

Although \defref{t:spanprogramdef} seems to have asymmetrical conditions conditions for when $f_P(x) = 1$ versus when $f_P(x) = 0$, this is misleading.  In fact, span programs can be complemented freely.  This is important for composing span programs that compute non-monotone functions.  

\begin{lemma} \label{t:dualspanprogram}
For every span program $P$, there exists a span program $P^\dagger$, said to be ``dual" to $P$, that computes the negation of $f_P$, $f_{P^\dagger}(x) = \neg f_P(x)$, with witness size $\wsizexS{P^\dagger} x s = \wsizexS P x s$ for all $x \in \B^n$ and $s \in [0,\infty)^n$.  
\end{lemma}

\begin{proof}
There are different constructions of dual span programs~\cite{CramerFehr02secretshare, NikovNikovaPreneel05spanprogram, ReichardtSpalek08spanprogram}.  Here we more or less follow \cite[Sec.~2.3]{ReichardtSpalek08spanprogram}, as the other constructions may not preserve witness size.  

As in \defref{t:spanprogramdef}, let $P$ have target vector $\ket t$ and input vectors $\ket{v_i}$, for $i \in I = \Ifree \sqcup \bigsqcup_{j \in [n], b \in \B} I_{j,b}$, in the inner product space $V = \C^d$.  Recall that $A = \sum_{i \in I} \ketbra{v_i}{i}$, $I(x) = \Ifree \cup \bigcup_{j \in [n]} I_{j, x_j}$ and $\Pi(x) = \sum_{i \in I(x)} \ketbra i i$.  
Let $\tilde \Pi(x) = \sum_{i \in I(x) \smallsetminus \Ifree} \ketbra i i$, and fix an orthonormal basis $\{ \ket k : k \in [d] \}$ for $V$.  

\begin{definition}
The dual span program $P^\dagger$, with target vector $\ket{t'}$ and input vectors $\ket{v_k'}$ for $k \in I' = \Ifree' \sqcup \bigsqcup_{j \in [n], b \in \B} I_{j, b}'$ in the inner product space $V'$, is defined by: 
\begin{itemize}
\item
$V' = \C^{1 + \abs{I}}$, with orthonormal basis $\{ \ket 0 \} \sqcup \{ \ket i : i \in I \}$.  
\item 
$\ket{t'} = \ket 0$.  
\item
$\Ifree' = [d]$, with free input vectors, for $k \in \Ifree'$, 
\begin{equation}
\ket{v_k'} = (\ketbra 0 t + A^\dagger) \ket k = \ket 0 \!\braket{t}{k} + \sum_{i \in I} \ket i \!\braket{v_i}{k}
 \enspace .
\end{equation}  
\item
For $j \in [n]$ and $b \in \B$, $I_{j, b}' = I_{j, \bar b}$ with $\ket{v_i'} = \ket i$ for $i \in I_{j, b}'$.  
\end{itemize}
\end{definition}

Fix $s \in [0,\infty)^n$, and let $A' = \sum_{k \in I'} \ketbra{v_k'}{k} = \ketbra 0 t + A^\dagger + \sum_{i \in I \smallsetminus \Ifree} \ketbra i i$.  Let $I'(x) = \Ifree' \cup \bigcup_{j \in [n]} I_{j, x_j}'$ and $\Pi'(x) = \sum_{i \in I'(x)} \ketbra i i$.  

If $f_P(x) = 1$, then there exists a witness $\ket w \in \C^{\abs{I}}$ such that $A \Pi(x) \ket w = \ket t$.  Assume that $\ket w$ is an optimal witness, i.e., $\wsizexS P x s = \norm{S \ket w}^2$ (see \defref{t:wsizedef}).  Let $\ket{w'} = \ket 0 - \Pi(x) \ket w$.  Then $\braket{t'}{w'} = 1$ and
\begin{equation} \label{e:dualspanprogramtruecase}
\begin{split}
A'^\dagger \ket{w'} 
&= \Big( \ketbra{t}{0} + A + \sum_{i \in I \smallsetminus \Ifree} \ketbra i i \Big) (\ket 0 - \Pi(x) \ket w) \\
&= \ket t - A \Pi(x) \ket w - \sum_{i \in I \smallsetminus \Ifree} \ketbra i i \Pi(x) \ket w \\
&= - \Big(\identity - \sum_{i \in \Ifree} \ketbra i i \Big) \Pi(x) \ket w 
 \enspace .
\end{split}
\end{equation}
Therefore, $\ket{w'}$ is orthogonal to the available input vectors of $P^\dagger$ ($\Pi'(x) A'^\dagger \ket{w'} = 0$), implying that $\ket{w'}$ is a witness to $f_{P^\dagger}(x) = 0$.  Moreover, Eq.~\eqnref{e:dualspanprogramtruecase} also implies $\wsizexS P x s = \norm{S A'^\dagger \ket{w'}}^2 \geq \wsizexS{P^\dagger}{x}{s}$.  

Conversely, if $f_{P^\dagger}(x) = 0$, then there is a witness $\ket{w'} \in V'$, with $\braket{t'}{w'} = 1$, orthogonal to the available input vectors of $P^\dagger$.  Assume that $\ket{w'}$ is an optimal witness, i.e., $\wsizexS{P^\dagger}{x}{s} = \norm{S A'^\dagger \ket{w'}}{}^2$.  The two conditions $\braket{0}{w'} = 1$, and $\braket{i}{w'} = 0$ for all $i \in \cup_{j \in [n]} I_{j, \bar x_j}$, imply $(\identity - \Pi(x)) \ket{w'} = \ket 0$.  The condition that $\ket{w'}$ is orthogonal to the free input vectors then implies
\begin{equation}\begin{split}
0
&= (\ketbra t 0 + A) \ket{w'} \\
&= (\ketbra t 0 + A) (\ket 0 + \Pi(x) \ket{w'}) \\
&= \ket t + A \Pi(x) \ket{w'}
 \enspace .
\end{split}\end{equation}
Thus $f_P(x) = 1$, with witness $\ket w = -\Pi(x) \ket{w'}$.  Moreover, the equalities of Eq.~\eqnref{e:dualspanprogramtruecase} still hold, so $\wsizexS{P^\dagger}{x}{s} = \norm{S \Pi(x) \ket w}^2 \geq \wsizexS P x s$.  

So far we have shown that $f_{P^\dagger}(x) = \neg f_P(x)$ for all $x \in \B^n$.  It remains to show that $\wsizexS P x s = \wsizexS{P^\dagger}{x}{s}$ in the case $f_P(x) = 0$.  

Assume that $f_P(x) = 0$.  Then there exists an optimal witness $\ket{w'}$ satisfying $\braket{t}{w'} = 1$, $\Pi(x) A^\dagger \ket{w'} = 0$ and $\wsizexS P x s = \norm{S A^\dagger \ket{w'}}^2$.  Let $\ket w = \ket{w'} - (\identity-\Pi(x)) A^\dagger \ket{w'}$.  Then $\ket w$ is supported only on the available input vector indices of $P^\dagger$; the first term, $\ket{w'}$, is supported on $\Ifree'$, while the second term is supported only on $\cup_{j \in [n]} I_{j \bar x_j}$.  Furthermore, 
\begin{equation}\begin{split}
A' \ket w
&= \bigg(\ketbra 0 t + A^\dagger + \sum_{i \in I \smallsetminus \Ifree} \ketbra i i \bigg) (\ket{w'} - (\identity-\Pi(x)) A^\dagger \ket{w'}) \\
&= (\ketbra 0 t + A^\dagger) \ket{w'} - \bigg( \sum_{i \in I \smallsetminus \Ifree} \ketbra i i - \tilde\Pi(x) \bigg) A^\dagger \ket{w'} \\
&= \ket 0 = \ket{t'}
\end{split}\end{equation}
since $\sum_{i \in I \smallsetminus \Ifree} \ketbra i i - \tilde\Pi(x) = \identity-\Pi(x)$.  Therefore $\ket w$ is a witness to $f_{P^\dagger}(x) = 0$.  Moreover, the squared length of $S \big(\identity - \sum_{k \in \Ifree'} \ketbra k k\big) \ket w$ is $\norm{S (1-\Pi(x)) A^\dagger \ket{w'}}^2 = \norm{S A^\dagger \ket{w'}}^2 = \wsizexS P x s$, so $\wsizexS{P^\dagger}{x}{s} \leq \wsizexS P x s$.  

\def\Pipfree{{\Pi'_\mathrm{free}}}
To show the converse statement, $\wsizexS P x s \leq \wsizexS{P^\dagger}{x}{s}$, let $\Pipfree = \sum_{k \in \Ifree'} \ketbra k k$ be the projection onto the free columns of $A'$.  Let $\ket w$ be an optimal witness to $f_{P^\dagger}(x) = 1$, i.e., $\wsizexS{P^\dagger}{x}{S} = \norm{S(\identity-\Pipfree) \ket w}^2$.  Then $\Pi'(x) \ket w = \big(\Pipfree + \sum_{i \in I'(x) \smallsetminus \Ifree'} \ketbra i i \big) \ket w = \ket w$ and 
\begin{align}
\ket{t'} = \ket 0
&= A' \ket w \nonumber \\
&= (\ketbra 0 t + A^\dagger) \Pipfree \ket w + \sum_{i \in I'(x) \smallsetminus \Ifree'} \ket i \!\braket i w
 \enspace .
\end{align}
This implies that $\bra t \Pipfree \ket w = 1$ and also $A^\dagger \Pipfree \ket w + \sum_{j \in n, i \in I_{j, \bar x_j}} \ket i \!\braket i w = 0$.  Multiplying by $\Pi(x)$, the latter equation implies that $\Pi(x) A^\dagger \Pipfree \ket w = 0$, so $\ket{w'} = \Pipfree \ket w$ is a witness for $f_P(x) = 0$.  Therefore, 
\begin{equation}\begin{split}
\wsizexS P x s
&\leq \norm{S A^\dagger \ket{w'}}^2 \\
&= \Big\lVert{S \sum_{j \in [n], i \in I_{j, \bar x_j}} \ket i \!\braket i w}\Big\rVert^2 \\
&= \norm{S (\identity - \Pipfree) \ket w}^2 \\
&= \wsizexS{P^\dagger}{x}{s}
 \enspace .
\end{split}\end{equation}
Thus $\wsizexS{P^\dagger}{x}{s} = \wsizexS P x s$ always.  
\end{proof}

\subsection{Tensor-product and direct-sum span program composition} \label{s:spanprogramcomposition}

\def\rtensor {{r\tensor}}

We will now show that the best span program witness size for a function composes sub-multiplicatively, in the following sense: 

\begin{theorem}[Span program composition] \label{t:spanprogramcomposition}
Consider functions $f : \B^n \rightarrow \B$ and, for $j \in [n]$, $f_j : \B^m \rightarrow \B$.  Let $g : \B^m \rightarrow \B$ be defined by 
\begin{equation}
g(x) = f\big( f_1(x), f_2(x), \ldots, f_n(x) \big)
 \enspace .
\end{equation}
Let $P$ be a span program computing $f_P = f$ and, for $j \in [n]$, let $P_j$ be a span program computing $f_{P_j} = f_j$.  

Then there exists a span program $Q$ computing $f_Q = g$, and such that, for any $s \in [0, \infty)^m$ and $r_j = \wsizeS{P_j}{s}$, 
\begin{equation} \label{e:spanprogramcomposition}
\wsizeS Q s \leq \wsizeS P r
 \enspace .
\end{equation}
In particular, $\wsizeS Q s \leq \wsize P \max_{j \in [n]} \wsizeS{P_j}{s}$.
\end{theorem}

The ease with which span programs compose is one of their nicest features.  To prove \thmref{t:spanprogramcomposition}, we will give two constructions of composed span programs, a tensor-product-composed span program $Q^\tensor$ and a direct-sum-composed span program $Q^\oplus$, that each satisfy Eq.~\eqnref{e:spanprogramcomposition}.  The method of composing span programs used in~\cite{ReichardtSpalek08spanprogram} is a special case of direct-sum composition, but tensor-product composition is new.  Below the proof of \thmref{t:spanprogramcomposition}, we will define a third composition method, reduced-tensor-product span program composition, that is closely related to tensor-product composition.  

Of course, only one proof of \thmref{t:spanprogramcomposition} is needed, so the definitions and proofs for $Q^\oplus$ and $Q^\rtensor$ can be safely skipped over.  We include here multiple composition methods because the different constructions have different tradeoffs when it comes to designing efficient quantum algorithms for formula evaluation.  In particular, we believe that an intermediate construction, in which some inputs are composed in a reduced-tensor-product fashion and other inputs in the direct-sum fashion, should be useful for designing a slightly faster quantum algorithm for evaluating AND-OR formulas~\cite{Reichardt09andorfaster}.  \appref{s:compositionexamples} gives examples of the three different span program composition methods for AND-OR formulas, using the correspondence between span programs and bipartite graphs that will be developed in \secref{s:bipartite}.  

\begin{proof}[Proof of \thmref{t:spanprogramcomposition}]
Let span program $P$ be in inner-product space $V$, with target vector $\ket t$ and input vectors indexed by $\Ifree$ and $I_{jc}$ for $j \in [n]$ and $c \in \B$.  For $j \in [n]$, let $P^{j1} = P_j$ and let $P^{j0}$ be a span program computing $f_{P^{j0}} = \neg f_{P^{j1}}$ with $\wsizeS{P^{j0}}{s} = \wsizeS{P^{j1}}{s}$.  Such span programs exist by \lemref{t:dualspanprogram}.  For $j \in [n]$ and $c \in \B$, let $P^{jc}$ be in the inner product space $V^{jc}$ with target vector $\ket{t^{jc}}$ and input vectors indexed by $\Ifree^{jc}$ and $I^{jc}_{kb}$ for $k \in [m]$, $b \in \B$.  

Some more notation will be convenient.  For $x \in \B^m$, let $y = y(x) \in \B^n$ be given by $y(x)_j = f_{P^{j1}}(x) = f_j(x)$ for $j \in [n]$.  Thus $g(x) = f(y(x))$.  Also let $I(y)' = I(y) \smallsetminus \Ifree = \cup_{j \in [n]} I_{j y_j}$.  Define $\jc : I \smallsetminus \Ifree \rightarrow [n] \times \B$ by $\jc(i) = (j,c)$ if $i \in I_{jc}$.  The idea is that $\jc$ maps $i$ to the index of the span program that must evaluate to true in order for $\ket{v_i}$ to be available.  

\begin{definition} \label{t:tensorproductcomposedef}
The tensor-product-composed span program $Q^\tensor$ is defined by: 
\begin{itemize}
\item
The inner product space is $V^\tensor = V \otimes \bigotimes_{j \in [n], c \in \B} V^{jc}$.  
\item
The target vector is $\ket{t^\tensor} = \ket{t}_V \otimes \bigotimes_{j \in [n], c \in \B} \ket{t^{jc}}_{V^{jc}}$.  
\item
The free input vectors are indexed by $\Ifree^\tensor = \Ifree \sqcup \bigsqcup_{j \in [n], c \in \B} (I_{jc} \times \Ifree^{jc})$ with, for $i \in \Ifree^\tensor$, 
\begin{equation}
\ket{v^\tensor_i} = \begin{cases}
\ket{v_i}_V \otimes \bigotimes_{j \in [n], c \in \B} \ket{t^{jc}}_{V^{jc}} & \text{if $i \in \Ifree$} \\
\ket{v_{i'}}_V \otimes \ket{v_{i''}}_{V^{jc}} \otimes \bigotimes_{\substack{j' \in [n], c' \in \B : \\ (j',c') \neq (j,c)}} \ket{t^{j'c'}}_{V^{j'c'}} & \text{if $i = (i', i'') \in I_{jc} \times \Ifree^{jc}$} 
\end{cases}
\end{equation}
\item
The other input vectors are indexed by $I^\tensor_{kb} = \sqcup_{j \in [n], c \in \B} (I_{jc} \times I^{jc}_{kb})$ for $k \in [m]$, $b \in \B$.  For $i \in I_{jc}$, $i' \in I^{jc}_{kb}$, let 
\begin{equation}
\ket{v^\tensor_{ii'}} = \ket{v_i}_V \otimes \ket{v_{i'}}_{V^{jc}} \otimes \bigotimes_{\substack{j' \in [n], c' \in \B : \\ (j',c') \neq (j,c)}} \ket{t^{j'c'}}_{V^{j'c'}}
 \enspace .
\end{equation}
\end{itemize}
\end{definition}

\begin{definition} \label{t:directsumcomposedef}
The direct-sum-composed span program $Q^\oplus$ is defined by: 
\begin{itemize}
\item
The inner product space is $V^\oplus = V \oplus \bigoplus_{j \in [n], c \in \B} (\C^{I_{jc}} \otimes V^{jc})$.  Any vector in $V^\oplus$ can be uniquely expressed as $\ket{u}_V + \sum_{i \in I \smallsetminus \Ifree} \ket i \otimes \ket{u_i}$, where $\ket u \in V$ and $\ket{u_i} \in V^{\jc(i)}$.  
\item
The target vector is $\ket{t^\oplus} = \ket{t}_V$.  
\item
The free input vectors are indexed by $\Ifree^\oplus = I \sqcup \bigsqcup_{j \in [n], c \in \B} (I_{jc} \times \Ifree^{jc})$ with, for $i \in \Ifree^\oplus$, 
\begin{equation}
\ket{v^\oplus_i} = \begin{cases}
\ket{v_i}_V & \text{if $i \in \Ifree$} \\
\ket{v_i}_V - \ket i \otimes \ket{t^{jc}} & \text{if $i \in I_{jc}$} \\
\ket{i'} \otimes \ket{v_{i''}} & \text{if $i = (i', i'') \in I_{jc} \times \Ifree^{jc}$} 
\end{cases}
\end{equation}
\item
The other input vectors are indexed by $I^\oplus_{kb}
= \sqcup_{j \in [n], c \in \B} (I_{jc} \times I^{jc}_{kb})$ for $k \in [m]$, $b \in \B$.  For $i \in I_{jc}$, $i' \in I^{jc}_{kb}$, let 
\begin{equation}
\ket{v^\oplus_{ii'}} = \ket{i} \otimes \ket{v_{i'}}
 \enspace .
\end{equation}
\end{itemize}
\end{definition}

For $x \in \B^m$, the indices of the available input vectors for $Q^\tensor$ and $Q^\oplus$ are 
\begin{align}
I^\tensor(x) 
&= \Ifree \cup \bigcup_{j \in [n], c \in \B} I_{jc} \times I^{jc}(x)
\\
I^\oplus(x)
&= I \cup \bigcup_{j \in [n], c \in \B} I_{jc} \times I^{jc}(x)
 \enspace .
\end{align}
Note that if $\Ifree = \Ifree^{j c} = \emptyset$ for $j \in [n]$ and $c \in \B$, then $Q^\tensor$ has no free input vectors either, $\Ifree^\tensor = \emptyset$.  

Assume $g(x) = f_P(y(x)) = 1$.  Then we have witnesses $\ket w \in \C^I$ and $\ket{w^{j y_j}} \in \C^{I^{j y_j}}$, for $j \in [n]$, such that 
\begin{equation}\begin{split}
\ket t &= \sum_{i \in I(y)} w_i \ket{v_i} \\
\ket{t^{j y_j}} &= \sum_{i \in I^{j y_j}(x)} w^{j y_j}_i \ket{v_i}
 \enspace ,
\end{split}\end{equation}
and such that $\wsizexS P y r = \norm{R \ket w}^2$ (where, analogous to the definition of $S$ in \defref{t:wsizedef}, $R = \sum_{j \in [n], c \in \B, i \in I_{jc}} \sqrt{r_j} \ketbra i i$) and $\wsizexS{P^{j y_j}}{x}{s} = \norm{S \ket{w^{j y_j}}}{}^2$.  

Let $\ket{w^\tensor} \in \C^{I^\tensor(x)}$ be 
\begin{equation}
w^\tensor_i 
= \begin{cases}
w_i & \text{if $i \in \Ifree$} \\
w_{i'} w^{\jc(i')}_{i''} & \text{if $i = (i', i'')$ with $i' \in I(y)'$, $i'' \in I^{\jc(i')}(x)$} \\
0 & \text{otherwise}
\end{cases}
\end{equation}
Then 
\begin{align}
\sum_{i \in I^\tensor(x)} w^\tensor_i \ket{v^\tensor_i}
&= 
\sum_{i \in \Ifree} w_i \ket{v_i}_V \otimes \bigotimes_{j \in [n], c \in \B} \ket{t^{jc}}_{V^{jc}} 
 + \sum_{\substack{i \in I(y)', \\ i' \in I^{\jc(i)}(x)}} w_i \ket{v_i}_V \otimes w^{\jc(i)}_{i'} \ket{v_{i'}}_{V^{\jc(i)}} \otimes \bigotimes_{\substack{j \in [n], c \in \B: \\ (j,c) \neq \jc(i)}} \ket{t^{jc}}_{V^{jc}} \nonumber \\
&= \sum_{i \in I(y)} w_i \ket{v_i}_V \otimes \bigotimes_{j \in [n], c \in \B} \ket{t^{jc}}_{V^{jc}} \nonumber \\
&= \ket{t^\tensor}
 \enspace ,
\end{align}
so indeed $f_{Q^\tensor}(x) = 1$.  

Let $\ket{w^\oplus} \in \C^{I^\oplus(x)}$ be 
\begin{equation}
w^\oplus_i 
= \begin{cases}
w_i & \text{if $i \in I(y)$} \\
w_{i'} w^{\jc(i')}_{i''} & \text{if $i = (i', i'')$ with $i' \in I(y)'$, $i'' \in I^{\jc(i')}(x)$} \\
0 & \text{otherwise}
\end{cases}
\end{equation}
Then
\begin{align}
\sum_{i \in I^\oplus(x)} w^\oplus_i \ket{v^\oplus_i}
&= \sum_{i \in \Ifree} w_i \ket{v_i}_V + \sum_{i \in I(y)'} w_i \big( \ket{v_i}_V - \ket i \otimes \ket{t^{\jc(i)}} \big) + \sum_{\substack{i \in I(y)',\\ i' \in I^{\jc(i)}(x)}} w_{i} w^{\jc(i)}_{i'} \ket i \otimes \ket{v_{i'}} \nonumber \\
&= \sum_{i \in I(y)} w_i \ket{v_i}_V + \sum_{i \in I(y)'} w_i \ket i \otimes \bigg[ - \ket{t^{\jc(i)}} + \sum_{i' \in I^{\jc(i)}(x)} w^{\jc(i)}_{i'} \ket{v_{i'}} \bigg] \nonumber \\
&= \ket{t}_V = \ket{t^\oplus}
 \enspace ,
\end{align}
so indeed $f_{Q^\oplus}(x) = 1$.  

Moreover, 
\begin{equation}\begin{split}
\norm{S \ket{w^\tensor}}{}^2 = \norm{S \ket{w^\oplus}}{}^2 
&= \sum_{\substack{j \in [n], i \in I_{j y_j}, \\ k \in [m], i' \in I^{j y_j}_{k x_k}}} s_k \lvert w_i w^{j y_j}_{i'} \rvert^2 \\
&= \sum_{i \in I(y)'} \wsizexS{P^{\jc(i)}}{x}{s} \abs{w_i}^2 \\
&= \wsizexS P y r
 \enspace ,
\end{split}\end{equation}
so $\wsizexS{Q^\tensor}{x}{s} = \wsizexS{Q^\oplus}{x}{s} \leq \wsizexS P y r$.  

Now assume that $g(x) = f_P(y) = 0$.  Then we have witnesses $\ket u \in V$ and $\ket{u^{j \bar y_j}} \in V^{j \bar y_j}$, for $j \in [n]$, such that $\braket t u = \braket{t^{j \bar y_j}}{u^{j \bar y_j}} = 1$, $\braket{v_i}{u} = 0$ for $i \in I(y)$, $\braket{v_i}{u^{j \bar y_j}} = 0$ for $i \in I^{j \bar y_j}(x)$, $\wsizexS P y r = \sum_{j \in [n], i \in I_{j \bar y_j}} r_j \abs{\braket{v_i}{u}}^2$ and $\wsizexS{P^{j \bar y_j}} x s = \sum_{k \in [m], i \in I^{j \bar y_j}_{k \bar x_k}} s_k \abs{\braket{v_i}{u^{j \bar y_j}}}{}^2$.  

Let 
\begin{equation}
\ket{u^\tensor} = \ket{u}_V \otimes \bigotimes_{j \in [n]} \Bigg( \ket{u^{j \bar y_j}}_{V^{j \bar y_j}} \otimes \frac{\ket{t^{j y_j}}_{V^{j y_j}}}{\norm{\ket{t^{j y_j}}}^2} \Bigg)
 \enspace .
\end{equation}
Then $\braket{t^\tensor}{u^\tensor} = 1$.  For $i \in \Ifree$, $\braket{v^\tensor_i}{u^\tensor} = 0$ since $\braket{v_i}{u} = 0$, and similarly for $i \in I_{j y_j}$, $i' \in I^{j y_j}(x)$, $\braket{v^\tensor_{i,i'}}{u^\tensor} = 0$.   We also have that for $j \in [n]$, $i \in I_{j \bar y_j}$ and $i' \in I^{j \bar y_j}(x)$, $\braket{v^\tensor_{ii'}}{u^\tensor} = 0$, since $\braket{v_{i'}}{u^{j \bar y_j}} = 0$.  
Thus $\braket{v^\tensor_i}{u^\tensor} = 0$ for all $i \in I^\tensor(x)$, so $\ket{u^\tensor}$ is a witness for $f_{Q^\tensor}(x) = 0$.  Moreover, 
\begin{align}
\wsizexS{Q^\tensor}{x}{s}
&\leq \sum_{\substack{j \in [n], c \in \B, i \in I_{j c}, \\ k \in [m], i' \in I^{j c}_{k \bar x_k}}} s_k \abs{\braket{v^\tensor_{ii'}}{u^\tensor}}{}^2 \nonumber \\
&= 
\sum_{\substack{j \in [n], i \in I_{j \bar y_j}, \\ k \in [m], i' \in I^{j \bar y_j}_{k \bar x_k}}} s_k \abs{\braket{v^\tensor_{ii'}}{u^\tensor}}{}^2 
 \enspace , \\
\intertext{where we have used $\braket{v^\tensor_{ii'}}{u^\tensor} = 0$ for $i \in I(y)$, since $\braket{v_i}{u} = 0$,}
&= 
\sum_{\substack{i \in I \smallsetminus I(y), \\ k \in [m], i' \in I^{\jc(i)}_{k \bar x_k}}}
s_k \Biggl | 
\Biggl[ \bra{v_i}_V \otimes \bra{v_{i'}}_{V^{\jc(i)}} \otimes \bigotimes_{\substack{j \in [n], c \in \B : \\ (j,c) \neq \jc(i)}} \bra{t^{jc}}_{V^{jc}} \Biggr] \nonumber \\
&\qquad\qquad\qquad\qquad 
\cdot \Biggl[ \ket{u}_V \otimes \bigotimes_{j \in [n]} \Bigg( \ket{u^{j \bar y_j}}_{V^{j \bar y_j}} \otimes \frac{\ket{t^{j y_j}}_{V^{j y_j}}}{\norm{\ket{t^{j y_j}}}{}^2} \Bigg) \Biggr] 
\Biggr |^2 \nonumber \\
&= \sum_{\substack{i \in I \smallsetminus I(y), \\ k \in [m], i' \in I^{\jc(i)}_{k \bar x_k}}} s_k \abs{\braket{v_i}{u}}^2 \cdot \abs{\braket{v_{i'}}{u^{\jc(i)}}}{}^2 \nonumber \\
&= \sum_{i \in I \smallsetminus I(y)} \wsizexS{P^{\jc(i)}}{x}{s} \abs{\braket{v_i}{u}}^2 \nonumber \\
&= \wsizexS P y r
 \enspace ,
\end{align}
where we have substituted the definitions of $\ket{v^\tensor_{ii'}}$ and $\ket{u^\tensor}$, and used $\braket{t^{j \bar y_j}}{u^{j \bar y_j}} = 1$.  We conclude that $f_{Q^\tensor} = g$ and $\wsizeS{Q^\tensor}{s} \leq \wsizeS P r$.  

Let 
\begin{equation}
\ket{u^\oplus} = \ket{u}_V + \sum_{i \in I \smallsetminus I(y)} \braket{v_i}{u} \ket i \otimes \ket{u_i}
 \enspace .
\end{equation}
Then $\braket{t^\oplus}{u^\oplus} = 1$.  For $i \in \Ifree^\oplus$, $\braket{v^\oplus_i}{u^\oplus} = 0$.  Indeed, this follows for $i \in I(y)$ since $\braket{v_i}{u} = 0$, and it holds for $i \in I \smallsetminus I(y)$ since $(\bra{v_i}_V - \bra i \otimes \bra{t^{\jc(i)}}) (\ket{u}_V + \braket{v_i}{u} \ket i \otimes \ket{u_i}) = 0$.  
$\ket{u^\oplus}$ is clearly orthogonal to the entire subspace $\ket i \otimes V^{\jc(i)}$ for $i \in I(y)$.  
Finally, for $i \in I \smallsetminus I(y)$ and $i' \in I^{\jc(i)}(x)$, $\braket{v^\oplus_{ii'}}{u^\oplus} = 0$ since $\braket{v_{i'}}{u_i} = 0$.  
Thus $\braket{v^\oplus_i}{u^\oplus} = 0$ for all $i \in I^\oplus(x)$, so $\ket{u^\oplus}$ is a witness for $f_{Q^\oplus}(x) = 0$.  Moreover, 
\begin{align}
\wsizexS{Q^\oplus}{x}{s}
&\leq \sum_{\substack{i \in I \smallsetminus \Ifree \\ k \in [m], i' \in I^{\jc(i)}_{k \bar x_k}}} s_k \abs{\braket{v^\oplus_{ii'}}{u^\oplus}}{}^2 \nonumber \\
&= \sum_{\substack{i \in I \smallsetminus I(y) \\ k \in [m], i' \in I^{\jc(i)}_{k \bar x_k}}} s_k \big\lvert\big(\bra{i} \otimes \bra{v_{i'}}\big) \big(\braket{v_i}{u} \ket i \otimes \ket{u_i}\big)\big\rvert^2 \nonumber \\
&= \sum_{\substack{i \in I \smallsetminus I(y) \\ k \in [m], i' \in I^{\jc(i)}_{k \bar x_k}}} s_k \abs{\braket{v_i}{u}}^2 \abs{\braket{v_{i'}}{u_i}}^2 \nonumber \\
&= \sum_{i \in I \smallsetminus I(y)} \wsizexS{P^{\jc(i)}}{x}{s} \abs{\braket{v_i}{u}}^2 \nonumber \\
&= \wsizexS P y r
 \enspace .
\end{align}
We conclude that $f_{Q^\oplus} = g$ and $\wsizeS{Q^\oplus}{s} \leq \wsizeS P r$.  
\end{proof}

Tensor-product composition is somewhat extravagant in the dimension of the final inner product space.  This is not a particular concern theoretically, since a set of $m$ vectors can always be embedded isometrically in at most $m$ dimensions.  However, it can be convenient to have an explicit isometric embedding of the composed span program's vectors into a lower dimensional space.  The ``reduced" tensor-product span program composition, which we will define next, is such an embedding.  It is particularly effective when the outer span program has many zero entries in its input vectors.  Canonical span programs, defined below in \secref{s:canonicalspanprograms}, are good examples.  

As in the setup for \thmref{t:spanprogramcomposition}, consider functions $f : \B^n \rightarrow \B$ and, for $k \in [n]$, $f_k : \B^m \rightarrow \B$.  Let $g : \B^m \rightarrow \B$ be defined by 
\begin{equation}
g(x) = f\big( f_1(x), f_2(x), \ldots, f_n(x) \big)
 \enspace .
\end{equation}
Let $P$ be a span program computing $f_P = f$ and, for $j \in [n]$, let $P_j$ be a span program computing $f_{P_j} = f_j$.  

Let span program $P$ be in inner-product space $V$, with target vector $\ket t$ and input vectors indexed by $\Ifree$ and $I_{jc}$ for $j \in [n]$ and $c \in \B$.  For $j \in [n]$, let $P^{j1} = P_j$ and let $P^{j0}$ be a span program computing $f_{P^{j0}} = \neg f_{P^{j1}}$ with $\wsizeS{P^{j0}}{s} = \wsizeS{P^{j1}}{s}$.  For $j \in [n]$ and $c \in \B$, let $P^{jc}$ be in the inner product space $V^{jc}$ with target vector $\ket{t^{jc}}$ and input vectors indexed by $\Ifree^{jc}$ and $I^{jc}_{kb}$ for $k \in [m]$, $b \in \B$.  

Let $d = \dim(V)$ and $\{ \ket l : l \in [d] \}$ be an orthonormal basis for $V$.  

\begin{definition} \label{t:reducedtensorproductcomposedef}
The tensor-product-composed span program, reduced with respect to the basis $\{ \ket l : l \in [d] \}$, is $Q^\rtensor$, defined by: 
\begin{itemize}
\item
For $l \in [d]$, let $Z_l = \{ (j, c) \in [n] \times \B : \forall \, i \in I_{jc}, \braket{l}{v_i} = 0 \}$, and let $\pi_l = \prod_{(j,c) \in Z_l} \norm{\ket{t^{jc}}}$.
\item
The inner product space of $Q^\rtensor$ is $V^\rtensor = \bigoplus_{l \in [d]} \big( \bigotimes_{(j,c) \notin Z_l} V^{jc} \big)$.  Any vector $\ket v \in V^\rtensor$ can be uniquely expressed as $\sum_{l \in [d]} \ket l \otimes \ket{v_l}$, where $\ket{v_l} \in \bigotimes_{(j,c) \notin Z_l} V^{jc}$.  
\item
The target vector is 
\begin{equation}
\ket{t^\rtensor} = \sum_{l \in [d]} \braket l t \ket l \pi_l \otimes \bigotimes_{(j,c) \notin Z_l} \ket{t^{jc}}_{V^{jc}}
 \enspace .
\end{equation}
\item
The free input vectors are indexed by $\Ifree^\rtensor = \Ifree^\tensor = \Ifree \sqcup \bigsqcup_{j \in [n], c \in \B} (I_{jc} \times \Ifree^{jc})$ with, for $i \in \Ifree^\rtensor$, 
\begin{equation}
\ket{v^\rtensor_i} = \begin{cases}
\sum_{l \in [d]} \braket{l}{v_i} \ket l \pi_l \otimes \bigotimes_{(j,c) \notin Z_l} \ket{t^{jc}}_{V^{jc}} & \text{if $i \in \Ifree$} \\
\sum_{l \in [d]} \braket{l}{v_{i'}} \ket l \pi_l \otimes \ket{v_{i''}}_{V^{jc}} \otimes \bigotimes_{\substack{(j',c') \notin Z_l : \\ (j',c') \neq (j,c)}} \ket{t^{j'c'}}_{V^{j'c'}} & \text{if $i = (i', i'') \in I_{jc} \times \Ifree^{jc}$}
\end{cases}
\end{equation}
\item
The other input vectors are indexed by $I^\rtensor_{kb} = I^\tensor_{kb} = \sqcup_{j \in [n], c \in \B} (I_{jc} \times I^{jc}_{kb})$ for $k \in [m]$, $b \in \B$.  For $i \in I_{jc}$, $i' \in I^{jc}_{kb}$, let 
\begin{equation}
\ket{v^\rtensor_{ii'}} = 
\sum_{l \in [d]} \braket{l}{v_i} \ket l \pi_l \otimes \ket{v_{i'}}_{V^{jc}} \otimes \bigotimes_{\substack{(j',c') \notin Z_l : \\ (j', c') \neq (j,c)}} \ket{t^{j'c'}}_{V^{j'c'}}
 \enspace .
\end{equation}
\end{itemize}
\end{definition}

For example, if $P$ is a canonical span program---see \defref{t:spanprogramcanonicaldef} below---with $\{ \ket x : x \in \B^n , f_P(x) = 0 \}$ an orthonormal basis for $V$, then for each $x$ with $f_P(x) = 0$, $\{ (j, x_j) : j \in [n] \} \subseteq Z_x$.  

\begin{proposition} \label{t:reducedtensorproductcompose}
The span program $Q^\rtensor$ computes $f_{Q^\rtensor} = g$, and, for any $s \in [0, \infty)^m$, 
\begin{equation}
\wsizeS {Q^\rtensor} s \leq \wsizeS P \sigma
 \enspace ,
\end{equation}
where $\sigma_j = \wsizeS{P_j}{s}$ for $j \in [n]$.  In particular, $\wsizeS {Q^\rtensor} s \leq \wsize P \max_{j \in [n]} \wsizeS{P_j}{s}$.
\end{proposition}

\begin{proof}
Rather than repeat the proof of \thmref{t:spanprogramcomposition}, it is enough to note that the input vectors of $Q^\rtensor$ are in one-to-one correspondence with the input vectors of $Q^\tensor$, and that the lengths of, and angles between, corresponding vectors are preserved.  Therefore, $f_{Q^\rtensor} = f_{Q^\tensor}$ and for all $s \in [0, \infty)^n$ and $x \in \B^n$, $\wsizexS {Q^\rtensor} x s = \wsizexS {Q^\tensor} x s$.  
\end{proof}

To conclude this section, let us remark that the composed span programs $Q^\oplus$, $Q^\tensor$ and $Q^\rtensor$ from \thmref{t:spanprogramcomposition} and \defref{t:reducedtensorproductcomposedef} are optimal under certain conditions.  

\begin{corollary} \label{t:spanprogramcompositiontight}
In \thmref{t:spanprogramcomposition}, assume that the functions $f_j$, $j \in [n]$, depend on disjoint sets of the input bits.  Assume also that the span programs $P_j$ have witness sizes $r_j = \wsizeS{P_j}{s} = \ADVpm_s(f_j)$ and that $P$ has witness size $\wsizeS P r = \ADVpm_r(f)$.  (By \thmref{t:wsizeadvbound}, these witness sizes are optimal.)   Then the composed span program $Q$ satisfies 
\begin{equation} \label{e:spanprogramcompositiontight}
\wsizeS Q s = \ADVpm_s(g) = \ADVpm_r(f)
 \enspace ,
\end{equation}
which is optimal.  
\end{corollary}

\begin{proof}
We have the inequalities
\begin{equation}\begin{split}
\ADVpm_r(f)
&\leq 
\ADVpm_s(g) \\ 
&\leq
\wsizeS Q s \\
&\leq
\wsizeS P r \\
&=
\ADVpm_r(f) 
 \enspace ,
\end{split}\end{equation}
where the three inequalities are from \thmref{t:weakadversarycomposition}, \thmref{t:wsizeadvbound} and \thmref{t:spanprogramcomposition}, respectively.  Therefore, all inequalities are equalities, and Eq.~\eqnref{e:spanprogramcompositiontight} follows.  
\end{proof}

\thmref{t:spanprogramSDP} below will show that a span program $P$ has optimal witness size with costs $s$ among all span programs computing $f_P$ if and only if $\wsizeS P s = \ADVpm_s(f_P)$.  Therefore, \corref{t:spanprogramcompositiontight} says that $Q$ is optimal if the input span programs are optimal and the $f_j$ depend on disjoint sets of the input bits.

\subsection{Strict and real span programs} \label{s:strictrealspanprograms}

For searching for span programs with optimal witness size, it turns out that \defref{t:spanprogramdef} is more general than necessary.  In fact, it suffices to consider span programs over the reals $\R$, and without any free input vectors.  

\begin{definition} \label{t:strictmonotonerealdef}
Let $P$ be a span program.  
\begin{itemize}
\item
$P$ is \emph{strict} if it has no free input vectors, i.e., $\Ifree = \emptyset$.  
\item
$P$ is \emph{real} if in a basis for $V$ the coefficients of the input and target vectors are all real numbers.  
\item
$P$ is \emph{monotone} if $I_{j,0} = \emptyset$ for all $j \in [n]$.  
\end{itemize}
\end{definition}

As remarked in \secref{s:spanprogramdef}, \cite{KarchmerWigderson93span} considered only strict span programs.  

\begin{proposition} \label{t:spanprogramrelaxed}
For any span program $P$, there exists a strict span program $P'$ with $f_{P'} = f_P$ and $\wsizexS{P'}{x}{s} = \wsizexS P x s$ for all $s \in [0, \infty)^n$ and $x \in \B^n$.  
\end{proposition}

\def\Pfree{{\overline{\Delta}_\mathrm{free}}}

\begin{proof}
Construct $P'$ by projecting $P$'s target vector $\ket t$ and input vectors $\{\ket{v_i} : i \in I \smallsetminus \Ifree\}$ to the space orthogonal to the span of the free input vectors.  That is, let $\Pfree$ be the projection onto the space orthogonal to $\Span(\{\ket{v_i} : i \in \Ifree\})$.  Then the target vector of $P'$ is $\Pfree \ket t$ and the input vectors are $\{ \Pfree \ket{v_i} : i \in I \smallsetminus \Ifree\}$.  

Then $f_{P'} = f_P$.  Indeed, if $f_{P}(x) = 1$, i.e., $\ket t = A \Pi(x) \ket w$ for some witness $\ket w$, then $\ket w$ is also a witness for $f_{P'}(x) = 1$.  Conversely, if $f_{P'}(x) = 1$, i.e., for some $\ket w$, $\Pfree \ket t = \Pfree A \Pi(x) \ket w$, then $\ket t - A \Pi(x) \ket w \in \Range(\{\ket{v_i} : i \in \Ifree\})$, so $f_P(x) = 1$.  

Now fix $s \in [0, \infty)^n$.  We claim that $\wsizexS{P'}{x}{s} = \wsizexS P x s$ for all $x \in \B^n$.  

First, if $f_P(x) = 0$, then by \defref{t:wsizedef},
\begin{equation}\begin{split}
\wsizexS P x s
&= \min_{\substack{\ket{w'} : \braket{t}{w'} = 1 \\ \Pi(x) A^\adjoint \ket{w'} = 0}} \norm{S A^\adjoint \ket{w'}}^2 \\
&= \min_{\substack{\ket{w'} : \braket{t}{w'} = 1 \\ \Pi(x) A^\adjoint \ket{w'} = 0 \\ \Pfree \ket{w'} = \ket{w'}}} \norm{S A^\adjoint \ket{w'}}^2 \\
&= \min_{\substack{\ket{w'} : \bra{t} \Pfree \ket{w'} = 1 \\ \Pi(x) A^\adjoint \Pfree \ket{w'} = 0}} \norm{S A^\adjoint \Pfree \ket{w'}}^2 \\
&= \wsizexS{P'}{x}{s}
 \enspace ,
\end{split}\end{equation}
where the second equality is because $\Pi(x) A^\adjoint \ket{w'} = 0$ implies in particular that $\braket{v_i}{w'} = 0$ for all $i \in \Ifree$.  

If $f_{P}(x) = 1$, then let $\Pifree = \sum_{i \in \Ifree} \ketbra i i$ and $\Pi'(x) = \Pi(x) - \Pifree = \sum_{i \in I(x) \smallsetminus \Ifree} \ketbra i i$.  We have 
\begin{align}
\wsizexS P x s 
&= \min_{\ket w : A \Pi(x) \ket w = \ket t} \norm{S \Pi'(x) \ket w}^2 \nonumber \\
&= \min_{\substack{\ket w : \Pi'(x) \ket w = \ket w \\ A \ket w - \ket t \in \Range(A \Pifree)}} \norm{S \ket w}^2 \\
&= \min_{\ket w : \Pfree A \Pi'(x) \ket w = \Pfree \ket t} \norm{S \ket w}^2 \nonumber \\
&= \wsizexS{P'}{x}{s}
 \enspace . \qedhere
\end{align}
\end{proof}

Span programs may also be taken to be real without harming the witness size:

\begin{lemma} \label{t:spanprogramCtoR}
For any span program $P$, there exists a real span program $P'$ computing the same function $f_{P'} = f_P$, with $\wsizexS{P'}x{s} \leq \wsizexS P x s$ for every cost vector $s \in [0, \infty)^n$ and $x \in \B^n$.  
\end{lemma}

\begin{proof}
\def\i{\iota}	% I use a different index in this proof so I can let i = sqrt{-1}
Let $i = \sqrt{-1}$.  For a complex number $c \in \C$, let $\Re(c), \Im(c) \in \R$ denote its real and imaginary parts, $c = \Re(c) + \Im(c) i$.  Extend this definition entry-wise to complex vectors: for $v \in \C^l$, let $\Re(v) = (\Re(v_1), \ldots, \Re(v_l))$ and $\Im(v) = (\Im(v_1), \ldots, \Im(v_l))$.  Furthermore, define $R : \C^l \rightarrow \R^l \tensor \R^2$ by 
\begin{equation}
R(v) = \Re(v) \otimes \ket 0 + \Im(v) \otimes \ket 1
 \enspace .
\end{equation}
Note that this map satisfies, for any vector $v \in \C^l$ and any scalar $c \in \C$, $\norm{R(v)} = \norm{v}$ and 
\begin{equation} \label{e:CtoRmultiplication}
R(c \, v) = \Re(c) R(v) + \Im(c) R(i v)
 \enspace .
\end{equation}

Let $P$ have target vector $\ket t$, and input vectors $\ket{v_\i}$ for $\iota \in I = \Ifree \sqcup \bigsqcup_{j \in [n], b \in \B} I_{j,b}$.  Fix an arbitrary orthonormal basis for $P$'s inner product space $V$.  

Let the inner product space for $P'$ be $V' = V \otimes \C^2$.  Let $P'$'s target vector be $\ket{t'} = R(\ket t)$, and its input vectors be indexed by $I' = I \times \B$, such that for any $x \in \B^n$, the set of available input vectors is indexed by $I'(x) = I(x) \times \B$.  That is, $\Ifree' = \Ifree \times \B$ and $I'_{j,b} = I_{j,b} \times \B$ for $j \in [n]$ and $b \in \B$.   
For $(\i, b) \in I' = I \times \B$, let the corresponding input vector be 
\begin{equation}
\ket{v_{\i,b}} 
= R(i^b \ket{v_\i})
= \begin{cases}
\Re(\ket{v_\i}) \otimes \ket 0 + \Im(\ket{v_\i}) \otimes \ket 1 & \text{if $b = 0$} \\
-\Im(\ket{v_\i}) \otimes \ket 0 + \Re(\ket{v_\i}) \otimes \ket 1 & \text{if $b = 1$}
\end{cases}
\end{equation}

Fix a cost vector $s \in [0, \infty)^n$.  Let $A = \sum_{\i \in I} \ketbra{v_\i}{\i}$, $A' = \sum_{(\i,b) \in I'} \ketbra{v_{\i,b}}{\i, b}$ and $\Pi(x) = \sum_{\i \in I(x)} \ketbra \i \i$.

The most interesting case to check is when $f_{P'}(x) = 0$.  Let $\ket{w'}$ be an optimal witness, i.e., $\braket{w'}{t'} = 1$, $(\Pi(x) \otimes \identity) A'^\dagger \ket{w'} = 0$ and $\wsizexS {P'} x s = \norm{(\Pi(x) \otimes \identity) A'^\dagger \ket{w'}}{}^2$.  Then since the entries of $P'$'s target and input vectors are real, $\Re(\ket{w'})$ is also a witness for $f_{P'}(x) = 0$, with equal or better witness size, so assume that $\ket{w'} = \Re(\ket{w'})$.  Let $\ket w$ be such that $R(\ket w) = \ket{w'}$.  Then $\Re(\braket w t) = \braket{w'}{t'} = 1$ so $\abs{\braket w t} \geq 1$; there may be a nonzero imaginary part to $\braket w t$.  Also, $A'^\dagger \ket{w'} = R( A^\dagger \ket w )$, so $\ket w$ is a witness for $f_P(x) = 0$ and $\norm{(S \otimes \identity) A'^\dagger \ket{w'}}{}^2 = \norm{S A^\dagger \ket w}{}^2$.  Hence $\wsizexS P x s \leq \wsizexS {P'} x s / \abs{\braket w t} \leq \wsizexS {P'} x s$.  

The arguments in the other cases are similar.  In every case, witnesses for $P$ and for $P'$ have a simple correspondence.  If $\ket w$ is a witness for $f_P(x) = b \in \B$, then $\ket{w'} = R(\ket w)$ will be a witness for $f_{P'}(x) = b$.  If $\ket{w'}$ is a witness for $f_{P'}(x) = b$, then so is $\Re(\ket{w'})$, and letting $\ket w$ be such that $R(\ket w) = \Re(\ket{w'})$, $\ket w$ will be a witness for $f_P(x) = b$.  We omit the details.
\end{proof}

\lemref{t:spanprogramCtoR} implies that there would have been no loss in generality in defining span programs over $\R$ instead of over $\C$.  In some cases, though, it is convenient to work over $\C$ to have smaller span programs.  For example,~\cite{ReichardtSpalek08spanprogram} gives a span program for the three-majority function with three input vectors and optimal witness size two, and one can verify that this is impossible for span programs over $\R$.  

The idea of the construction in \lemref{t:spanprogramCtoR} is essentially to replace every entry $a$ of $A = \sum_{i \in I} \ketbra{v_i}{i}$ by the $2 \times 2$ block $\left(\begin{smallmatrix}\Re a & -\Im a \\ \Im a & \Re a \end{smallmatrix}\right)$ to simulate multiplication of complex numbers over the reals.  The proof can be slightly simplified by assuming, without loss of generality, that $\ket t = \ket 1$, a basis vector for $V$.  We have avoided doing so, though, in order to illustrate a special case of how span programs can be defined over matrices.  For $j, k, l \in \N$, \defref{t:spanprogramdef} and \defref{t:wsizedef} naturally extend to allowing the target to be a vector of $j \times l $ matrices and the input vectors to have entries that are $j \times k$ matrices.  The program evaluates to $1$ if there exists a way of summing available input vectors multiplied by $k \times l$ matrices to reach the target.   Provided that an entry-wise matrix inner product is used in the generalization of \defref{t:wsizedef}, such programs can be simulated over $\R$ without changing the witness size.  This generalization can be useful for finding span programs when we would like to work with a higher-dimensional representation of a function's symmetry group.  For example, this technique has been used to find an optimal span program for a Hamming-weight threshold function in~\cite[Example~5.1]{ReichardtSpalek08spanprogram}.  

Let us conclude this section with one last span program manipulation: 

\begin{lemma} \label{t:wsizescaling}
For $P$ a span program and $M$ any invertible linear transformation on $P$'s inner product space $V$, $f_P$ and the witness size of $P$ are invariant under applying $M$ to the target vector and all input vectors.  
\end{lemma}

\begin{proof}
Let $P'$ be the span program in which $M$ has been applied to $P$'s target and input vectors.  The claim is that for all $x \in \B^n$ and $s \in [0, \infty)^n$, $f_{P'}(x) = f_P(x)$ and $\wsizexS P x s = \wsizexS {P'} x s$.  Indeed, 
the conditions $A \Pi(x) \ket w = \ket t$ and $(M A) \Pi(x) \ket w = M \ket t$ are equivalent.  This implies that $f_P = f_{P'}$ and, when $f_P(x) = 1$, $\wsizexS P x s = \wsizexS {P'} x s$, by definition Eq.~\eqnref{e:wsizetrue}.  To finish the proof, note that when $f_P(x) = 0$, 
\begin{equation}
\min_{\substack{\ket{w'} : \braket{t}{w'} = 1 \\ \Pi(x) A^\adjoint \ket{w'} = 0}} \norm{S A^\adjoint \ket{w'}}{}^2
=
\min_{\substack{\ket{w'} : \bra t M^\adjoint \ket{w'} = 1 \\ \Pi(x) (M A)^\adjoint \ket{w'} = 0}} \norm{S (M A)^\adjoint \ket{w'}}{}^2
\end{equation}
by the change of variables $\ket{w'} \rightarrow M^\dagger \ket{w'}$.  By Eq.~\eqnref{e:wsizefalse}, $\wsizexS P x s = \wsizexS {P'} x s$.   
\end{proof}

\section{Canonical span programs} \label{s:canonicalspanprograms}

For every function $f : \B^n \rightarrow \B$, there exists a span program $P$ computing $f_P = f$.  Indeed, one can, for example, expand $f$ into a circuit that uses OR gates and NOT gates.  Each OR gate can be implemented by a trivial span program with $\ket t = \ket{v_i} = (1) \in \C$.  Appealing to \lemref{t:dualspanprogram} to negate this span program, and using \thmref{t:spanprogramcomposition} to compose the span programs following the circuit, gives a span program for $f$.  However, this naive span program will generally not have the optimal witness size among all span programs computing $f$.  Moreover, there is considerable freedom in the definition of span programs, so unless $f$ is very simple, it can be quite difficult to find an optimal span program.  

In this section we prove that it suffices to search over span programs with a restricted form, so-called canonical span programs.  Combined with \lemref{t:spanprogramCtoR}, this implies that it suffices to look for real, canonical span programs.  In \secref{s:spanprogramsdp}, we will develop a semi-definite program, inspired by this reduction, for computing the optimal span program for a function.  Canonical span programs were originally defined by Karchmer and Wigderson~\cite{KarchmerWigderson93span}, but their significance for developing quantum algorithms was not at first appreciated.  

\begin{definition}[Canonical span program~\cite{KarchmerWigderson93span}] \label{t:spanprogramcanonicaldef}
Let $P$ be a span program computing $f_P : \B^n \rightarrow \B$, with inner product space $V$, target vector $\ket t$ and input vectors $\ket{v_i}$ for $i \in \Ifree \sqcup \bigsqcup_{j \in [n], b \in \B} I_{j,b}$.  $P$ is canonical if:
\begin{itemize}
\item
$\Ifree = \emptyset$.  Thus $P$ is strict (\defref{t:strictmonotonerealdef}).  
\item
$V = \C^{F_0}$ where $F_0 = \{ x \in \B^n : f_P(x) = 0 \}$.  
\item
In the orthonormal basis $\{ \ket x : x \in F_0 \}$ for $V$, the target $\ket t$ is given by $\ket t = \sum_{x \in F_0} \ket x$, and 
\item
For all $x \in F_0$, $j \in [n]$ and $i \in I_{j, x_j}$, $\braket{x}{v_i} = 0$.  
\end{itemize}
\end{definition}

\begin{theorem} \label{t:spanprogramcanonical}
For any cost vector $s \in [0, \infty)^n$, a span program $P$ can be converted to a canonical span program $\hat P$ that computes the same function $f_{\hat P} = f_P$, with $\wsizexS{\hat P}{x}{s} \leq \wsizexS P x s$ for all $x \in \B^n$.  
In fact, for all $x \in \B^n$ with $f_P(x) = 0$, $\wsizexS{\hat P}{x}{s} = \wsizexS P x s$, with $\ket x$ itself an optimal witness for $f_{\hat P}(x) = 0$.  

Moreover, $\hat P$ uses the same input vector index sets $I_{j,b}$ as $P$, so in particular if $P$ is monotone then $\hat P$ is also monotone.  If $P$ is real, then so is $\hat P$.  
\end{theorem}

\begin{proof}
This theorem is analogous to~\cite[Theorem~6]{KarchmerWigderson93span}, and we use the same conversion procedure, except we additionally analyze the witness size.  

Let $P$ have target vector $\ket t \in V$ and input vectors $\ket{v_i}$ for $i \in I = \Ifree \cup \bigcup_{j \in [n], b \in \B} I_{j,b}$.  Recall the definitions $A = \sum_{i \in I} \ketbra{v_i}{i}$, $I(x) = \Ifree \cup \bigcup_{j \in [n]} I_{j, x_j}$ and $\Pi(x) = \sum_{i \in I(x)} \ketbra i i$.  Fix $s \in [0, \infty)^n$ and let $S = \sum_{j \in [n], b \in \B, i \in I_{j,b}} \sqrt{s_j} \ketbra i i$.

For $x \in \B^n$, let $\ket{w(x)}$ or $\ket{w'(x)}$ be optimal witnesses for $f_P(x)$ being 1 or 0, respectively, with costs $s$.  That is, let 
\begin{equation}
\begin{array}{r@{\;=\;}l@{\qquad}l}
\ket{w(x)} & {\arg \min}_{\ket w : A \Pi(x) \ket w = \ket t} \norm{S \ket w}^2 & \text{if $f_P(x) = 1$} \\[6pt]
\ket{w'(x)} & {\arg \min}_{\substack{\ket{w'} : \braket{t}{w'} = 1\\ \Pi(x) A^\dagger \ket{w'} = 0}} \norm{S A^\dagger \ket{w'}}^2 & \text{if $f_P(x) = 0$}
\end{array}
\end{equation}
(See~\cite[Lemma~A.3]{ReichardtSpalek08spanprogram} for explicit formulas for $\ket{w(x)}$ and $\ket{w'(x)}$.)

Let $F_0 = \{ x \in \B^n : f_P(x) = 0 \}$.  To construct $\hat P$ from $P$, simply apply to $P$'s target and input vectors the map $\sum_{x \in F_0} \ketbra{x}{w'(x)} \in \L(V, \C^{F_0})$.  Then 
\begin{itemize}
\item
The target vector becomes $\ket{\hat t} = \sum_{x \in F_0} \ket x \braket{w'(x)}{t} = \sum_{x \in F_0} \ket x \in \C^{F_0}$, as required for a canonical span program.  The input vectors become, for $i \in I$, $\ket{\hat{v}_i} = \sum_{x \in F_0} \ket x \braket{w'(x)}{v_i}$.  
\item
For any $x \in F_0$ and $i \in I(x)$, since $\braket{w'(x)}{v_i} = 0$, $\braket{x}{\hat v_i} = 0$.  
\item
In particular, for $i \in \Ifree$, $\braket{x}{\hat v_i} = 0$ for all $x \in F_0$.  Thus $\ket{\hat v_i} = 0$, so the free input vectors may be discarded.  
\end{itemize}
Therefore $\hat P$ is a canonical span program.  $\hat P$ is monotone if $P$ is monotone, and $\hat P$ is real if $P$ is real.  

Let $\hat A = \sum_{i \in I} \ketbra{\hat v_i}{i} = \sum_{x \in F_0} \ketbra{x}{w'(x)} A$.  

For $x \in F_0$, note that $\braket{\hat t}{x} = 1$ and 
\begin{equation} \label{e:spanprogramcanonicalfalsewitness}
\hat{A}^\dagger \ket x = A^\dagger \ket{w'(x)}
 \enspace .
\end{equation}
In particular, $\Pi(x) \hat{A}^\dagger \ket x = 0$, so $\ket x$ is a witness for $f_{\hat P}(x) = 0$.  Also, $\wsizexS{\hat P}{x}{s} \leq \norm{S \hat{A}^\dagger \ket x}{}^2 = \norm{S A^\dagger \ket{w'(x)}}{}^2 = \wsizexS P x s$.  In fact, $\ket x$ is an optimal witness for $f_{\hat P}(x) = 0$.  Indeed, assume otherwise, and let $\ket{\hat u} = \sum_{y \in F_0} \hat u_y \ket y$ satisfy $\braket{\hat t}{\hat u} = \sum_{y \in F_0} \hat u_y = 1$, $\Pi(x) \hat A^\dagger \ket{\hat u} = 0$ and $\norm{S \hat A^\dagger \ket{\hat u}}{}^2 < \norm{S \hat A^\dagger \ket x}{}^2$.  Let $\ket u = \sum_{y \in F_0} \hat u_y \ket{w'(y)}$, so $A^\dagger \ket u = \hat A^\dagger \ket{\hat u}$.  Then $\braket t u = 1$, $\Pi(x) A^\dagger \ket u = 0$, and $\norm{S A^\dagger \ket u}{}^2 = \norm{S \hat A^\dagger \ket{\hat u}}{}^2 < \wsizexS P x s$, a contradiction.  

Next consider an $x \in \B^n$ such that $f_P(x) = 1$.  Then 
\begin{equation}\begin{split}
\hat A \Pi(x) \ket{w(x)}
&= \sum_{y \in F_0} \ketbra{y}{w'(y)} A \Pi(x) \ket{w(x)} \\
&= \sum_{y \in F_0} \ket y \braket{w'(y)}{t} \\
&= \ket{\hat t}
 \enspace .
\end{split}\end{equation}
Thus $\ket{w(x)}$ is a witness for $f_{\hat P}(x) = 1$, and $\wsizexS{\hat P}{x}{s} \leq \norm{S \ket{w(x)}}^2 = \wsizexS P x s$.  
\end{proof}

Note that the canonical span program $\hat P$ from \thmref{t:spanprogramcanonical} depends on the cost vector $s$.  In contrast, the strict span program $P'$ from \propref{t:spanprogramrelaxed} has witness size equal to that of $P$ for all $s \in [0, \infty)^n$.

\section{Span program witness size and the general adversary bound} \label{s:spanprogramsdp}

In this section, we will use \thmref{t:spanprogramcanonical} to formulate a semi-definite program (SDP) for the optimal span program computing a boolean function $f$.  Remarkably, this SDP turns out to be exactly the dual of the SDP that defines the general adversary bound for $f$ (\defref{t:adversarydef}).  Thus the optimal span program witness size is exactly equal to the general adversary bound.  This result has several corollaries, in quantum algorithms and in complexity theory, that we give in \secref{s:corollaries}.  

This result may be somewhat surprising, because the optimal span programs known previously were all for functions $f$ with $\ADV(f) = \ADVpm(f)$~\cite{ReichardtSpalek08spanprogram}.  It is not clear why earlier attempts to find optimal span programs did not succeed for any function $f$ with $\ADV(f) < \ADVpm(f)$.  

\begin{theorem} \label{t:spanprogramSDP}
For any function $f : \D \rightarrow \B$, with $\D \subseteq \B^n$, and any cost vector $s \in [0, \infty)^n$, 
\begin{equation} \label{e:spanprogramSDP}
\inf_{\large\substack{P : \, f_P\vert_\D = f}} \wsizeSD P s = \ADVpm_s(f)
 \enspace ,
\end{equation}
where the infimum is over span programs $P$ that compute a function agreeing with $f$ on $\D$.  Moreover, this infimum is achieved.  
\end{theorem}

Before proving \thmref{t:spanprogramSDP}, let us show the following dual characterization of the general adversary bound: 

\begin{theorem} \label{t:adversarydual}
For finite sets $\D \subseteq C^n$, and $E$, let $f: \D \rightarrow E$, and let $s \in [0, \infty)^n$ be a vector of nonnegative costs.  
If either $C = \{0,1\}$ or $E = \{0,1\}$, then the general adversary bound for $f$, with costs $s$, equals 
\begin{equation} \label{e:generaladversarydual}
\ADVpm_s(f) = 
\min_{\large \substack{X \succeq 0 : \\ \forall (x,y) \in F, \,  \sum_{j \in [n] : x_j \neq y_j} \bra{x,j} X \ket{y,j} = 1}}
\max_{\large \substack{x \in \D}} \sum_{j \in [n]} s_j \bra{x,j} X \ket{x,j} 
 \enspace .
\end{equation}
Here $X$ is required to be a positive semi-definite, $(n \abs{\D}) \times (n \abs{\D})$ matrix, with coordinates labeled by $\D \times [n]$, and $F = \{ (x, y) \in \D \times \D : f(x) \neq f(y) \}$.  
The optimum is achieved.  
\end{theorem}

\begin{proof}
The proof is by a standard application of duality theory to the semi-definite program given in \defref{t:adversarydef}.  Nonetheless, this expression for $\ADVpm_s(f)$ is new, and is somewhat simpler than the expression that was known before, Eq.~\eqnref{e:generaladversarydualgeneral} below.  Therefore we include a proof, based on the following immediate observation:  

\begin{claim} \label{t:bipartitematrix}
Let $M = \sum_{j, k \in [m]} M_{jk} \ketbra j k \in \L(\C^{[m]})$ be an $m \times m$ Hermitian matrix.  Assume that either $M$ is entry-wise nonnegative, i.e., $M_{jk} \geq 0$ for all $j, k \in [m]$, or that $M$ is bipartite, i.e., for some $l \in [m-1]$, $M = \sum_{j \leq l, k > l} (M_{jk} \ketbra j k + M_{kj} \ketbra k j)$.  Then $M \preceq \identity$ if and only if $\norm M \leq 1$.  
\end{claim}

Taking the dual of the SDP on the right-hand side of Eq.~\eqnref{e:generaladversarydual}, we obtain 
\begin{equation} \label{e:sdp2}
\max_{\large \substack{ \tilde \Gamma = \sum_F \alpha_{xy} \ketbra x y \\ \{ \beta_x \geq 0 \} }} \; \sum_F \alpha_{xy} \quad \text{such that} \quad \sum_x \beta_x \leq 1, \;\; \forall j, \; \tilde \Gamma_j \preceq s_j \sum_x \beta_x \ketbra x x
 \enspace .
\end{equation}
Here $\tilde \Gamma_j = \tilde \Gamma \circ \Delta_j = \sum_{x,y \in \D : x_j \neq y_j} \ketbra x x \tilde \Gamma \ketbra y y$ as in \defref{t:adversarydef}.  Also, $\tilde \Gamma_j \preceq s_j \sum_{x} \beta_x \ketbra x x$ means that the difference $(s_j \sum_{x \in \D} \beta_x \ketbra x x) - \tilde \Gamma_j$ is a positive semi-definite matrix.  In particular, this constraint implies that if $s_j = 0$ then $\alpha_{xy} = 0$ for all $x, y$ with $x_j \neq y_j$; and that if $\alpha_{xy} \neq 0$, then $\beta_x > 0$ and $\beta_y > 0$.  

Thus we can change variables, letting $\Gamma = \sum_{(x,y) \in \Delta : \alpha_{xy} \neq 0} \frac{\alpha_{xy}}{\sqrt{\beta_x \beta_y}} \ketbra x y$.  Like $\tilde \Gamma$, $\Gamma$ can vary over the set of adversary matrices, i.e., symmetric matrices supported only on those $\ketbra x y$ with $f(x) \neq f(y)$.  The objective function becomes $\sum_F \bra x \Gamma \ket y \sqrt{\beta_x \beta_y}$, and, for $j \in [n]$, the constraint on $\tilde \Gamma_j$ becomes $\Gamma_j \preceq s_j \identity$, where $\Gamma_j = \Gamma \circ \Delta_j$.  

Now if $C = \{0,1\}$, then the matrices $\Delta_j$ are bipartite---perhaps in a permuted basis---so each $\Gamma_j$ is also bipartite.  If $E = \{0,1\}$, then $\Gamma$ is bipartite since it is supported only on $F$.  In either case, by \claimref{t:bipartitematrix} the condition $\Gamma_j \preceq s_j \identity$ is equivalent to $\norm{\Gamma_j} \leq s_j$. Therefore, after changing variables, the SDP becomes
\begin{equation} \label{e:sdp3}
\max_{\large \substack{\text{adversary matrices $\Gamma$} \\ \{ \beta_x \geq 0 \} }} \; \sum_F \bra x \Gamma \ket y \sqrt{\beta_x \beta_y}
\quad \text{such that} \quad \sum_{x} \beta_x \leq 1, \;\; \forall j, \; \norm{\Gamma_j} \leq s_j
 \enspace .
\end{equation}
Since any negative signs on the coordinates of the principal eigenvector of $\Gamma$ can be absorbed into the matrix, without affecting the norms of the $\Gamma_j$, the objective function in Eq.~\eqnref{e:sdp3} simplifies to $\norm{\Gamma}$, so we obtain $\ADVpm(f)$.  Since the dual SDP in Eq.~\eqnref{e:sdp2} is clearly strictly feasible, by the duality principle~\cite[Theorem~3.4]{Lovasz00sdp} the primal optimum equals the dual optimum and the primal optimum is achieved.  Eq.~\eqnref{e:generaladversarydual} follows.  
\end{proof}

For completeness, we state without proof the dual forms of the adversary bounds for the case of functions without a binary input alphabet or boolean codomain: 

\begin{theorem} \label{t:adversarydualgeneral}
For finite sets $\D \subseteq C^n$, and $E$, let $f: \D \rightarrow E$, and let $s \in [0, \infty)^n$.  Let $F = \sum_{x, y \in \D : \, f(x) \neq f(y)} \ketbra x y$.  As in \defref{t:adversarydef}, let $\Delta_j = \sum_{x, y \in \D : x_j \neq y_j} \ketbra x y$ for $j \in [n]$, and let $\circ$ denote entry-wise matrix multiplication.  

Then the nonnegative-weight adversary bound for $f$, with costs $s$, equals
\begin{equation} \label{e:adversarydual}
\ADV_s(f) 
= 
\min_{\large \substack{X_j \succeq 0 : \\ \sum_j X_j \circ \Delta_j \circ F \geq F}} \max_{x \in \D} \sum_{j \in [n]} s_j \bra x X_j \ket x
 \enspace .
\end{equation}
The minimization is over $\abs \D \times \abs \D$ positive semi-definite matrices $X_j$, $j \in [n]$, that satisfy the entry-wise inequality $\sum_j X_j \circ \Delta_j \circ F \geq F$.  (Note that Eq.~\eqnref{e:generaladversarydual} has the same form, except with the requirement that $\sum_j X_j \circ \Delta_j \circ F = F$.)  

The general adversary bound for $f$, with costs $s$, equals 
\begin{equation} \label{e:generaladversarydualgeneral}
\ADVpm_s(f) = 
\min_{\large \substack{X_j, Y_j \succeq 0 : \\ \sum_j (X_j - Y_j) \circ \Delta_j \circ F = F}}
\max_{\large \substack{x \in \D}} \sum_{j \in [n]} s_j \bra{x} (X_j+Y_j) \ket{x} 
 \enspace .
\end{equation}
\end{theorem}

\begin{proof}[Proof of \thmref{t:spanprogramSDP}]
\lemref{t:spanprogramSDPconstruction} constructs an SDP whose solution is the optimal witness size of a span program computing $f$.  

\begin{lemma} \label{t:spanprogramSDPconstruction}
Let $f : \D \rightarrow \B$, with $\D \subseteq \B^n$, be a partial boolean function.  For $b \in \B$, let $F_b = \{ x \in \D : f(x) = b \}$.  Then for any cost vector $s \in [0, \infty)^n$, 
\begin{equation}\begin{split} \label{e:spanprogramSDPCholesky}
\inf_{\large\substack{P : \, f_P\vert_\D = f}} \wsizeSD P s
&=
\inf_{\large \substack{m \in \N, \\ \{ \ket{v_{x j}} \in \R^m : \, x \in \D, j \in [n] \} \,: \\ \forall (x,y) \in F_0 \times F_1, \, \sum_{j \in [n] : x_j \neq y_j} \braket{v_{x j}}{v_{y j}} = 1}} \max_{\large \substack{x \in \D}} \; \sum_{j \in [n]} s_j \norm{ \ket{v_{x j}} }^2
 \enspace .
\end{split}\end{equation}
\end{lemma}

\begin{proof}
The proof is by establishing a correspondence between solutions to the constraints on the right-hand side of Eq.~\eqnref{e:spanprogramSDPCholesky} and real, canonical span programs computing $f_P\vert_\D = f$ with $\max_{j \in [n], b \in \B} \abs{I_{j,b}} \leq m$.   

First let us prove the $\leq$ direction.  Given a solution $\{ \ket{v_{xj}} \}$, let $P$ be a span program with target $\ket t = \sum_{x \in F_0} \ket x \in \R^{F_0}$ and $I_{j,b} = [m]$ for all $j \in [n]$, $b \in \B$.  These sets are not disjoint, so for $k \in I_{j,b}$, use $\ket{v_{jbk}}$ to denote the corresponding input vector, defined by $\ket{v_{jbk}} = \sum_{x \in F_0 : x_j \neq b} \braket{v_{xj}}{k} \ket x$.  Thus 
\begin{equation} \label{e:spanprogramSDPcanonical}
\begin{split}
A
&:= \sum_{j \in [n], b \in \B, k \in [m]} \ketbra{v_{j b k}}{j,b,k} \\
&= \sum_{x \in F_0, j \in [n]} \ketbra{x}{j, \bar x_j} \otimes \bra{v_{x j}}
 \enspace .
\end{split}\end{equation}

For $x \in F_0$, $\ket{w'} = \ket x$ is a witness for $f_P(x) = 0$; $\braket x t = 1$ but $\braket{x}{v_{j x_j k}} = 0$ for all $j, k$.  The witness size is $\norm{A^\dagger \ket x}{}^2 = \sum_j s_j \norm{\ket{v_{xj}}}^2$.  

For $x \in F_1$, let $\ket w = \sum_{j} \ket {j, x_j} \otimes \ket {v_{xj}}$.  The condition that $\sum_{j : x_j \neq y_j} \braket{v_{y j}}{v_{x j}} = 1$ implies that $\ket w$ is a witness, $A \Pi(x) \ket w = A \ket w = \ket t$, so $f_P(x) = 1$.  The witness size is $\norm{\ket w}^2 = \sum_{j} s_j \norm{\ket{v_{x j}}}^2$.  

Thus $f_P\vert_\D = f$ and $\wsizeSD P s \leq \max_x \sum_j s_j \norm{\ket{v_{xj}}}^2$.  

Now let us prove the $\geq$ direction.  Let $P$ be a span program computing $f_P$, with $f_P\vert_\D = f$.  By \thmref{t:spanprogramcanonical} and \lemref{t:spanprogramCtoR}, we may assume that $P$ is real and in canonical form, and that for each $x \in F_0$, $\ket x$ is an optimal witness for $f_P(x) = 0$: $\wsizexS P x s = \norm{S A^\dagger \ket x}{}^2$.  

Thus the target vector is $\ket t = \sum_{x \in F_0} \ket x$ and the input vectors lie in the inner product space~$\R^{F_0}$.  Let $m = \max_{j \in [n], b \in \B} \abs{I_{j,b}}$.  Without loss of generality, we may assume that $\abs{I_{j,b}} = m$ for all $j \in [n]$ and $b \in \B$.  Indeed, if some index set $I_{j,b}$ is smaller, then we can pad the span program with zero vectors labeled by $(j,b)$ without affecting the witness size.  Therefore, let $I_{j,b} = [m]$ for all $j \in [n]$ and $b \in \B$.  These sets are not disjoint, so for $k \in I_{j,b}$, use $\ket{v_{jbk}}$ to denote the corresponding input vector.  

For $x \in F_0$, note that since the span program is canonical, $\braket{x}{v_{j x_j k}} = 0$ for all $j \in [n]$ and $k \in [m]$.  For $j \in [n]$, let $\ket{v_{xj}} = \sum_{k \in [m]} \braket{ v_{j \bar x_j k} }{x} \ket k$.  Then Eq.~\eqnref{e:spanprogramSDPcanonical} again holds.   Moreover, $\wsizexS P x s = \norm{S A^\dagger \ket x}{}^2 = \sum_{j \in [n]} s_j \norm{\ket{v_{x j}}}^2$.  Thus $\max_{x \in F_0} \sum_j s_j \norm{\ket{v_{xj}}}^2 \leq \wsizeSD P s$.  

For $x \in F_1$, on the other hand, let $\ket{w_x}$ be an optimal witness vector, i.e., satisfying $\ket{w_x} = \Pi(x) \ket{w_x} = \sum_{j \in [n], k \in [m]} \ket{j, x_j, k} \braket{j, x_j, k}{w_x}$, $A \ket{w_x} = \ket t$ and $\wsizexS P x s = \norm{S \ket{w_x}}^2$.  For $j \in [n]$, let $\ket{v_{x j}} = \sum_{k \in [m]} \ket k \braket{j, x_j, k}{w_x}$.  Then 
\begin{equation}
A \ket{w_x} = \ket t \qquad \Longrightarrow \qquad \forall \, y \in F_0, \; \sum_{j : x_j \neq y_j} \braket{v_{y j}}{v_{x j}} = 1
 \enspace .
\end{equation}
Finally, $\wsizexS P x s = \sum_j s_j \norm{\ket{v_{xj}}}^2$, so $\max_{x \in F_1} \sum_{j} s_j \norm{\ket{v_{xj}}}^2 \leq \wsizeSD P s$.  
\end{proof}

Now the expression on the right-hand side of Eq.~\eqnref{e:spanprogramSDP} is just the Cholesky decomposition of the solution to the SDP in Eq.~\eqnref{e:generaladversarydual}.  We conclude that $\inf_{P : f_P\vert_\D = f} \wsizeS P s = \ADVpm_s(f)$, as claimed.  
\end{proof}

Before stating some corollaries of \thmref{t:spanprogramSDP}, let us make a remark on the proof: 

\begin{lemma} \label{t:spanprogramSDPrank}
For a function $f : \D \rightarrow \B$, with $\D \subseteq \B^n$, assume that there is a rank-$k$ optimal solution $X$ to Eq.~\eqnref{e:generaladversarydual} for $\ADVpm(f)$.  Note that $k \leq n \abs{\D} \leq n 2^n$.   Then by the proof of \lemref{t:spanprogramSDPconstruction} there is an optimal span program computing $f$ with $\abs{I_{j,b}} = k$ for all $j \in [n]$ and $b \in \B$.   
\end{lemma}

\cite[Theorem~18]{HoyerLeeSpalek07negativeadv} states in particular that Eq.~\eqnref{e:adversarydual} always has a rank-one optimal solution.  The proof takes the Cholesky decomposition of a solution $X = \sum_{x, y, j, j'} \ket{x,j} \! \braket{v_{xj}}{v_{yj'}} \! \bra{y,j'}$, and replaces each vector $\ket{v_{xj}}$ by the scalar $\norm{\ket{v_{xj}}}$.  That is, let $X' = \sum_{x, y, j, j'} \norm{\ket{v_{xj}}} \norm{\ket{v_{yj'}}} \ketbra{x,j}{y,j'}$, a rank-one matrix.  Then by the Cauchy-Schwarz inequality, $\bra{x,j} X' \ket{y,j} \geq \bra{x,j} X \ket{y,j}$, with equality when $y = x$, so $X'$ is as good a solution to Eq.~\eqnref{e:adversarydual} as $X$ is.  However, note that even when $\ADV_s(f) = \ADVpm_s(f)$, this argument does not imply that Eq.~\eqnref{e:generaladversarydual} has a rank-one optimal solution~\cite{Spalek09private}.

\section{Consequences of the SDP for optimal witness size} \label{s:corollaries}

This section will state several corollaries of \thmref{t:spanprogramSDP}.  First of all, we can strengthen \thmref{t:spanprogramuniversal}.  

\begin{theorem} \label{t:spanprogramuniversalboundederror}
For any function $f : \D \rightarrow \{0,1\}$, with $\D \subseteq \{0,1\}^n$, there exists a span program $P$ computing $f_P\vert_\D = f$ with witness size upper-bounded by the bounded-error quantum query complexity of~$f$, 
\begin{equation}
\wsizeD P = O(Q(f))
 \enspace .
\end{equation}
\end{theorem}

\begin{proof}
By \thmref{t:advquerycomplexity}, the quantum query complexity of $f$ is lower-bounded by the general adversary bound for $f$, which by \thmref{t:spanprogramSDP} equals the best span program witness size: 
\begin{align}
Q(f)
&= \Omega(\ADVpm(f)) \\
&= \Omega\big( \inf_{P : f_P\vert_\D = f} \wsizeD P \big) \nonumber
 \enspace .
\qedhere
\end{align}
\end{proof}

It is an interesting problem to prove \thmref{t:spanprogramuniversalboundederror} based directly on a quantum query algorithm that evaluates $f$, as in the proof of \thmref{t:spanprogramuniversal} for the one-sided error case.  

As an immediate corollary of \thmref{t:spanprogramSDP} and \thmref{t:spanprogramcomposition}, the general adversary bound composes multiplicatively for boolean functions.  That is, the inequality in Eq.~\eqnref{e:weakadversarycomposition}, from \thmref{t:weakadversarycomposition}, is actually an equality.

\begin{theorem}[General adversary bound composition] \label{t:advpmcomposition}
Let $f : \{0,1\}^n \rightarrow \{0,1\}$ and, for $j \in [n]$, let $f_j : \{0,1\}^{m_j} \rightarrow \{0,1\}$.  
Define $g : \{0,1\}^{m_1} \times \cdots \times \{0,1\}^{m_n} \rightarrow \{0,1\}$ by $g(x) = f\big(f_1(x_1), \ldots, f_n(x_n)\big)$.  Let $s \in [0, \infty)^{m_1} \times \cdots \times [0, \infty)^{m_n}$, and let $\beta_j = \ADVpm_{s_j}(f_j)$ for $j \in [n]$.  Then 
\begin{equation}
\ADVpm_s(g) = \ADVpm_\beta(f)
 \enspace .
\end{equation}
In particular, if $\ADVpm_{s_1}(f_1) = \cdots = \ADVpm_{s_n}(f_n) = \beta$, then $\ADVpm_s(g) = \beta \, \ADVpm(f)$.  
\end{theorem}

\begin{proof}
\thmref{t:weakadversarycomposition} gives the inequality $\ADVpm_s(g) \geq \ADVpm_\beta(f)$.  To obtain the opposite inequality, appeal to \thmref{t:spanprogramSDP} to obtain optimal span programs for the functions, compose these span programs using \thmref{t:spanprogramcomposition}, and appeal to \thmref{t:spanprogramSDP} to upper-bound $\ADVpm_s(g)$.  

This proof is rather indirect.  Based on the new formulation of the general adversary bound in \thmref{t:adversarydual}, we can also give a direct proof of \thmref{t:advpmcomposition} that does not use span programs.  

Recall that $\B = \{0,1\}$.  For $x \in \B^{m_1} \times \cdots \times \B^{m_n}$, let $y(x) = (f_1(x), \ldots, f_n(x))$, so $g(x) = f(y(x))$.  

For $y \in \B^n$ and $j \in [n]$, fix vectors $\ket{v_{yj}} \in V$ that achieve $\ADVpm_\beta(f)$, i.e., $\sum_{j : y_j \neq y_{j'}} \braket{v_{yj}}{v_{y'j}} = 1$ for all $y, y' \in \B^n$ with $f(y) \neq f(y')$, and $\ADVpm_\beta(f) = \max_{y \in \B^n} \sum_{j \in [n]} \beta_j \norm{\ket{v_{yj}}}^2$.  For $j \in [n]$, fix vectors $\ket{v^j_{zk}} \in V^j$ for $z \in \B^{m_j}$, $k \in [m_j]$, that achieve $\ADVpm_{s_j}(f_j)$, i.e., $\sum_{k : z_k \neq z_k'} \braket{v^j_{zk}}{v^j_{z'k}} = 1$ for all $z, z' \in \B^{m_j}$ with $f_j(z) \neq f_j(z')$.  

Based on these solutions, we construct a feasible solution for the dual formulation of $\ADVpm_s(g)$.  For $x \in \B^{m_1} \times \cdots \times \B^{m_n}$, $j \in [n]$ and $k \in [m_j]$, let
\begin{equation}
\ket{w_{xjk}} = \ket{v_{y(x)j}} \otimes \ket{v^j_{x_j k}} \otimes \ket{\delta_{g(x), f_j(x_j)}} \in V \otimes (\oplus_{j \in n} V^j) \otimes \C^2
 \enspace .
\end{equation}
Here, the third register is spanned by the orthonormal basis $\{ \ket 0, \ket 1 \}$, and $\delta_{a,b}$ is $1$ if $a = b$ and $0$ otherwise.  

Consider $x, x' \in \B^{m_1} \times \cdots \times \B^{m_n}$ such that $g(x) \neq g(x')$.  In particular, $y(x) \neq y(x')$.  Then 
\begin{equation}\begin{split}
\sum_{\substack{j \in [n], k \in [m_j] : \\ x_{jk} \neq x'_{jk}}} \braket{w_{xjk}}{w_{x'jk}}
&= \sum_{j \in [n]} \braket{v_{y(x)j}}{v_{y(x')j}} \sum_{k \in [m_j]: x_{jk} \neq x'_{jk}} \braket{v^j_{x_j k}}{v^j_{x'_j k}} (1 - \delta_{f_j(x_j), f_j(x'_j)}) \\
&= \sum_{j \in [n] : y(x)_j \neq y(x')_j} \braket{v_{y(x)j}}{v_{y(x')j}} \sum_{k \in [m_j] : x_{jk} \neq x'_{jk}} \braket{v^j_{x_j k}}{v^j_{x'_j k}} \\
&= \sum_{j \in [n] : y(x)_j \neq y(x')_j} \braket{v_{y(x)j}}{v_{y(x')j}} \\
&= 1
 \enspace .
\end{split}\end{equation}
Hence indeed the vectors $\ket{w_{xjk}}$ give a feasible solution.  We conclude that 
\begin{equation}\begin{split}
\ADVpm_s(g)
&\leq \max_{x \in \B^{m_1} \times \cdots \times \B^{m_n}} \sum_{j \in [n], k \in [m_j]} s_{jk} \norm{\ket{w_{xjk}}}^2 \\
&= \max_x \sum_{j \in [n]} \norm{\ket{v_{y(x)j}}}^2 \sum_{k \in [m_j]} s_{jk} \norm{\ket{v^j_{x_j k}}}^2 \\
&\leq \max_x \sum_{j \in [n]} \beta_j \norm{\ket{v_{y(x)j}}}^2 \\
&= \ADVpm_\beta(f)
 \enspace .
\end{split}\end{equation}
The last step is clearly an inequality, which is all we actually need to finish the proof.  It is in fact an equality, though, because $y(x)$ varies over all strings in $\B^n$ as $x$ varies over $\B^{m_1} \times \cdots \times \B^{m_n}$.  
\end{proof}

By substituting \thmref{t:spanprogramSDP} into \thmref{t:spanprogramalgorithm}, we obtain an exact asymptotic expression for the quantum query complexity of a boolean function $f$ composed on itself.  

\begin{theorem} \label{t:querycomplexitylimit}
For any function $f : \{0,1\}^n \rightarrow \{0,1\}$, define $f^{k} : \{0,1\}^{n^k} \rightarrow \{0,1\}$ as the function $f$ composed on itself repeatedly to a depth of $k$, as in \thmref{t:spanprogramalgorithm}.  Then 
\begin{equation}
\lim_{k \rightarrow \infty} Q(f^{k})^{1/k} = \ADVpm(f)
 \enspace .
\end{equation}
\end{theorem}

\begin{proof}
By the adversary lower bound \thmref{t:advquerycomplexity}, $Q(f^{k}) = \Omega(\ADVpm(f^{k})) = \Omega(\ADVpm(f)^k)$ by \thmref{t:weakadversarycomposition}.  Hence ${\lim \inf}_{k \rightarrow \infty} Q(f^{k})^{1/k} \geq \ADVpm(f)$.  
\thmref{t:spanprogramSDP} together with the formula-evaluation algorithm \thmref{t:spanprogramalgorithm} implies $Q(f^{k}) = O_k(\ADVpm(f)^k)$.  Hence ${\lim \sup}_{k \rightarrow \infty} Q(f^{k})^{1/k} \leq \ADVpm(f)$.  
\end{proof}

\thmref{t:querycomplexitylimit} implies a new asymptotic upper bound on the sign-degree of a boolean function $f$ composed on itself to a depth of $k$, as $k \rightarrow \infty$.  

\begin{definition}[Sign-degree]
Given a function $f : \{0, 1\}^n \rightarrow \{0,1\}$, a real multivariate polynomial $p(x_1, \ldots, x_n)$ is said to be a threshold polynomial that sign-represents $f$ if for all inputs $x \in \{0,1\}^n$, $p(x) \neq 0$ and the signs of $p(x)$ and $(-1)^{f(x)}$ coincide.   

The sign-degree of $f$, $\signdegree(f)$, is defined as the least degree of a polynomial that sign-represents $f$.  
\end{definition}

By the polynomial method~\cite{BealsBuhrmanCleveMoscaWolf98, NielsenChuang00}, $\signdegree(f) = O(Q(f))$ for every boolean function $f$.  (See also Refs.~\cite{MontanaroNishimuraRaymond07unboundederror, BuhrmanVereshchaginDeWolf07unbounded}, which relate the sign-degree of $f$ to the unbounded-error quantum and classical query complexities of $f$.)  Thus we obtain the following corollary of \thmref{t:querycomplexitylimit}: 

\begin{corollary} \label{t:signdegree}
For any function $f : \{0,1\}^n \rightarrow \{0,1\}$, 
\begin{equation}
\underset{k \rightarrow \infty}{{\lim \sup}} \, \signdegree(f^{k})^{1/k} \le \lim_{k \rightarrow \infty} Q(f^k)^{1/k}= \ADVpm(f)
 \enspace .
\end{equation}
\end{corollary}

Lee and Servedio have recently shown that $\signdegree(f)^k \leq \signdegree(f^{k})$~\cite{Lee09private}, based on which \corref{t:signdegree} gives an upper bound of the sign-degree of $f$ itself.  

One special case of interest is when $f$ is a read-once AND-OR formula on $n$ variables.  In this case, $\ADV(f) = \ADVpm(f) = \sqrt n$~\cite{BarnumSaks04readonce}.  Indeed, these bounds can be computed by showing $\ADV_{(s_1, \ldots, s_m)}(\mathrm{AND}_m) = \ADVpm_{(s_1, \ldots, s_m)}(\mathrm{AND}_m) = \sqrt{\sum_{j \in [m]} s_j^2}$, where $\mathrm{AND}_m$ denotes the AND gate on $m$ variables, and then using \thmref{t:weakadversarycomposition} and \thmref{t:advpmcomposition} to compose the nonnegative-weight and general adversary bounds, respectively.  O'Donnell and Servedio~\cite{ODonnellServedio03threshold} asked whether $\signdegree(f) = O(\sqrt n)$?  
This question has consequences in learning theory~\cite{KS01, KSO04}.  Ambainis et al.\ proved that $\signdegree(f) = n^{1/2 + o(1)}$ by giving a quantum algorithm, and, therefore, an explicit threshold polynomial~\cite{AmbainisChildsReichardtSpalekZhang07andor}.  Combined with the unpublished result of Lee and Servedio mentioned above, \corref{t:signdegree} will close this question.  In fact, though, \cite{AmbainisChildsReichardtSpalekZhang07andor} with Lee and Servedio's result already suffices; the composed function $f^k$ is an ``approximately balanced" formula for any fixed AND-OR formula~$f$, and, by another result of~\cite{AmbainisChildsReichardtSpalekZhang07andor}, therefore $\signdegree(f^k) = O(\sqrt{n^k})$.

\thmref{t:spanprogramalgorithm} is only a special case of the formula-evaluation result from~\cite{ReichardtSpalek08spanprogram}.  That article's main result, \cite[Theorem~4.7]{ReichardtSpalek08spanprogram}, can also be extended.  For brevity, we will not repeat all the notation and definitions, but will just state the extension.  \cite{ReichardtSpalek08spanprogram} used the nonnegative-weight adversary bound $\ADV$ instead of the general adversary bound $\ADVpm$ throughout, because only for functions $f$ with $\ADV(f) = \ADVpm(f)$ had the authors found matching span programs.  \thmref{t:spanprogramSDP}, however, gives optimal span programs for every boolean function $f$.  Thus if we modify \cite[Def.~4.5]{ReichardtSpalek08spanprogram}, defining adversary-balanced formulas, to refer to $\ADVpm$ instead of $\ADV$, and if we let $\cal{S}$ be any finite gate set of boolean functions, \cite[Theorem~4.7]{ReichardtSpalek08spanprogram} becomes: 

\begin{theorem} \label{t:formulaevaluation}
There exists a quantum algorithm that evaluates an adversary-balanced formula $\varphi(x)$ over $\cal{S}$ using $O(\ADVpm(\varphi))$ input queries.  After efficient classical preprocessing independent of the input $x$, and assuming $O(1)$-time coherent access to the preprocessed classical string, the running time of the algorithm is $\ADVpm(\varphi) (\log \ADVpm(\varphi))^{O(1)}$.  
\end{theorem}

Aside from changing $\ADV$ to $\ADVpm$, the proof from~\cite{ReichardtSpalek08spanprogram} goes through entirely.  Note that layered formulas, in which gates at the same depth are the same, are a special case of adversary-balanced formulas.

\section{Correspondence between span programs and bipartite graphs} \label{s:bipartite}

In this section, we define a correspondence between span programs and weighted bipartite graphs, slightly generalizing the correspondence from~\cite{ReichardtSpalek08spanprogram}.  We also analyze the spectra of these graphs, focusing on eigenvalues near zero and eigenvectors supported on one particular ``output vertex."  The main result, \thmref{t:spanprogramspectralanalysis}, relates spectral quantities of interest to the span program witness size.  This is the key theorem that allows span programs to be evaluated on a quantum computer.  

\thmref{t:spanprogramspectralanalysis}'s proof has two main steps.  The first step, an eigenvalue-zero analysis given in \secref{s:bipartitezero}, is essentially the same as the argument in~\cite{ReichardtSpalek08spanprogram}.  
However, the second step, analyzing small, nonzero eigenvalues, is novel.  \secref{s:bipartitesmall} gives a general argument that relates properties of eigenvalue-zero eigenvectors of weighted bipartite graphs to what are in a certain sense ``{effective}" spectral gaps.  

This small-eigenvalue analysis substantially extends the proof in~\cite{ReichardtSpalek08spanprogram}.  The small-eigenvalue analysis in~\cite{ReichardtSpalek08spanprogram} only works for span programs that arise from the concatenation of constant-size span programs with constant entries, with strict balance conditions, and it breaks down when these conditions are relaxed.  For example, \cite{ReichardtSpalek08spanprogram} shows spectral gaps of $\Omega(1/\wsize P)$ away from zero, for a span program $P$, but the spectral gaps for general span programs cannot be lower-bounded in terms of the witness size.  The small-eigenvalue analysis in~\cite{ReichardtSpalek08spanprogram} is also more technically involved.  \thmref{t:spanprogramspectralanalysis} implies a simpler proof of \thmref{t:spanprogramalgorithm} and \thmref{t:formulaevaluation}, as well as for the AND-OR formula-evaluation result in~\cite{AmbainisChildsReichardtSpalekZhang07andor}.  

\smallskip

\begin{definition}
A finite, weighted, bipartite graph $G$ is specified by finite sets $T$ and $U$, and $\biadj_G \in \L(\C^U, \C^T)$ the ``biadjacency matrix."  $G$ has vertices $\{ \tau_i : i \in T \} \sqcup \{ \mu_j : j \in U \}$.  For each $i \in T$ and $j \in U$ with $\bra i \biadj_G \ket j \neq 0$, $G$ has an edge $(\tau_i, \mu_j)$ weighted by $\bra i \biadj_G \ket j$.  The weighted adjacency matrix of $G$, $A_G \in \L(\C^T \oplus \C^U)$, is 
\begin{equation}
A_G = \left( \begin{matrix} 0 & \biadj_G \\ \biadj_G^\dagger & 0 \end{matrix} \right)
\place{{\small T}}{-74mu}{18pt}
\place{{\small U}}{-30mu}{18pt}
\place{{\small T}}{0mu}{8pt}
\place{{\small U}}{0mu}{-6pt}
\end{equation}
\end{definition}

Henceforth all graphs will be finite.  
Recall from \secref{s:definitions} that $\B = \{0,1\}$.  
For a given span program, recall also the definitions $A = \sum_{i \in I} \ketbra{v_i}{i}$, $I(x) = \Ifree \cup \bigcup_{j \in [n]} I_{j, x_j}$ and $\Pi(x) = \sum_{i \in I(x)} \ketbra i i$.  Let 
\begin{equation}
\Pim(x) = \identity - \Pi(x) = \sum_{i \in I \smallsetminus I(x)} \ketbra i i
 \enspace .
\end{equation}

Now the correspondence between span programs and weighted bipartite graphs is given by:

\begin{definition}[Graphs $G_P(x)$] \label{t:spanprogramadjacencymatrix}
Let $P$ be a span program with target vector $\ket t$ and input vectors $\ket{v_i}$ for $i \in I = \Ifree \cup \bigcup_{j \in [n], b \in \B} I_{j,b}$, in inner product space $V$.  Fix an arbitrary orthonormal basis $\{ \ket k : k \in [\dim(V)] \}$ for $V$.  

Let $G_P$ be the weighted bipartite graph with $T = [\dim(V)] \sqcup I$, $U = \{ 0 \} \sqcup I$ and the biadjacency matrix
\begin{equation}
\biadj_{G_P} = \left( \begin{matrix} \ket t & A \\ 0 & \identity \end{matrix} \right)
\place{{\small $0$}}{-58mu}{18pt}
\place{{\small I}}{-23mu}{18pt}
\place{{\small V}}{0mu}{8pt}
\place{{\small I}}{0mu}{-6pt}
\end{equation}
The vertex $\mu_0$ is called the ``output vertex."

Note that for each input vector index $i \in I$, $G_P$ has two corresponding vertices, with a weight-one edge between them.  For $x \in \B^n$, let $G_P(x)$ be the same as $G_P$ except with these weight-one edges deleted for all $i \in I(x)$.  That is, $G_P(x)$ has the biadjacency matrix 
\begin{equation} \label{e:spanprogramadjacencymatrix}
\biadj_{G_P(x)} = \left( \begin{matrix} \ket t & A \\ 0 & \Pim(x) \end{matrix} \right)
\place{{\small $0$}}{-81mu}{18pt}
\place{{\small I}}{-35mu}{18pt}
\place{{\small V}}{0mu}{8pt}
\place{{\small I}}{0mu}{-6pt}
\end{equation}
\end{definition}

\defref{t:spanprogramadjacencymatrix} is a modest generalization of the correspondence between span programs and bipartite graphs given in~\cite[Sec.~2]{ReichardtSpalek08spanprogram}.  The difference is that \cite{ReichardtSpalek08spanprogram} only defines $G_P(x)$ for span programs with target $\ket t = (1, 0, 0, \ldots, 0)$.  This is not a very restrictive requirement, though, since a unitary change of basis can ensure that $\ket t = (\norm{\ket t}, 0, \ldots, 0)$.  

It will be convenient to establish some more notation. Any vector $\ket \psi \in \C^T \oplus \C^U$ can be uniquely expanded as $\ket \psi = (\ket{\psi_T}, \ket{\psi_U})$, with $\ket{\psi_T} \in \C^T$ and $\ket{\psi_U} \in \C^U$.   For the graphs $G_P(x)$, $\C^T = V \oplus \C^I$ and $\C^U = \C^{\{0\}} \oplus \C^I$, so any $\ket \psi \in \C^T \oplus \C^U$ can similarly be written $\ket \psi = \big( (\ket{\psi_{T,V}}, \ket{\psi_{T,I}}), ( \psi_{U,0}, \ket{\psi_{U,I}} ) \big)$.  Let $\ket 0 = (0, 1, 0) \in \C^T \oplus \C^{\{0\}} \oplus \C^I$ be the unit vector on vertex~$\mu_0$.  

With this notation, the eigenvalue-$\rho$ eigenvector equation of $A_{G_P(x)}$, 
\begin{equation}
\rho \ket \psi = A_{G_P(x)} \ket \psi
 \enspace ,
\end{equation}
is equivalent to the four equations 
\begin{subequations} \label{e:zeroenergyequations}
\begin{align}
\rho \ket{\psi_{T,V}} &= \psi_{U,0} \ket{t} + A \ket{\psi_{U,I}} \label{e:zeroenergyequationsbO} \\
\rho \ket{\psi_{T,I}} &= \Pim(x) \ket{\psi_{U,I}} \label{e:zeroenergyequationsbI} \\
\rho \, \psi_{U,0} &= \braket{t}{\psi_{T,I}} \label{e:zeroenergyequationsaO} \\
\rho \ket{\psi_{U,I}} &= A^\dagger \ket{\psi_{T,V}} + \Pim(x) \ket{\psi_{T,I}} \label{e:zeroenergyequationsaJ}
 \enspace .
\end{align}
\end{subequations}
%% NOTE: These equations cannot be reordered, since I have hard-coded references to them.

Our main result is: 

\begin{theorem} \label{t:spanprogramspectralanalysis}
Let $P$ be a span program and $\D \subseteq \B^n$.  Then a span program $P'$ can be constructed such that $f_{P'} = f_P$ and, for all $x \in \D$, 
\begin{itemize}
\item
If $f_P(x) = 1$, then there is an eigenvalue-zero eigenvector $\ket \psi$ of $A_{G_{P'}(x)}$ with 
\begin{equation} \label{e:spanprogramspectralanalysistrue}
\frac{\abs{ \braket 0 \psi }^2}{\norm{\ket \psi}^2}
\geq 
\frac{1}{2}
 \enspace .
\end{equation}
\item
If $f_P(x) = 0$, let $\{ \ket \alpha \}$ be a complete set of orthonormal eigenvectors of $A_{G_{P'}(x)}$, with corresponding eigenvalues $\rho(\alpha)$.  Then for any $c \geq 0$, the squared length of the projection of $\ket 0$ onto the span of the eigenvectors $\alpha$ with $\abs{\rho(\alpha)} \leq c / \wsizeD P$ satisfies 
\begin{equation} \label{e:spanprogramspectralanalysisfalse}
\sum_{\alpha : \, \abs{\rho(\alpha)} \leq c / \wsizeD P} \abs{\braket \alpha 0}^2 \leq 8 c^2 \Big(1 + \frac{1}{\wsizeD P}\Big) \leq 16 c^2
 \enspace .
\end{equation}
\end{itemize}
\end{theorem}

Roughly, Eq.~\eqnref{e:spanprogramspectralanalysisfalse} says that $A_{G_{P'(x)}}$ has an effective spectral gap around zero.  We will see in \secref{s:algorithm} below that this is strong enough for applying quantum phase estimation.  

The two main ingredients required for proving \thmref{t:spanprogramspectralanalysis}, an eigenvalue-zero analysis of $A_{G_P}(x)$ and an analysis relating eigenvalue-zero eigenvectors to the effective spectral gap.  These two ingredients are presented in \secref{s:bipartitezero} and \secref{s:bipartitesmall} below.  \secref{s:spanprogramspectralanalysisproof} will put them together to prove \thmref{t:spanprogramspectralanalysis}.  

\thmref{t:spanprogramspectralanalysis} is quite useful.  However, we will see in \secref{s:algorithm} below that for some applications, using \thmref{t:spanprogramspectralanalysis} as a black box can lead to an $O(\log n)$ overhead in the quantum query complexity.  \thmref{t:phaseestimationalgorithm} will include two quantum query algorithms.  The more specialized algorithm does not incur a logarithmic overhead, but requires that the norm of the adjacency matrix be at most a constant.  However, the span program $P'$ that \thmref{t:spanprogramspectralanalysis} outputs will not necessarily satisfy $\norm{A_{G_{P'}}} = O(1)$, so only the first algorithm applies.  Thus if one cares about saving logarithmic query overhead factors, \thmref{t:spanprogramspectralanalysis} cannot be applied as a black box.  

It is possible that the first algorithm in \thmref{t:phaseestimationalgorithm} can be improved to work without the logarithmic overhead even when $\norm{A_{G_{P'}}} = \omega(1)$.  See \conjref{t:advboundtight}.  Even if this turns out to be the case, though, there will be an important case when we cannot apply \thmref{t:spanprogramspectralanalysis} as a black box, namely, when we wish to prove upper bounds on the time complexity of the algorithm.  

For developing time-efficient quantum algorithms, other properties of the adjacency matrix besides the norm, such as the maximum degree of a vertex, also matter~\cite{ChiangNagajWocjan09simulate}.  This article is concerned primarily with the query complexity of quantum algorithms and not the time complexity.  Investigating the tradeoffs involved in designing span programs for query-optimal and nearly time-optimal quantum algorithms is an important area for further research, but is beyond our scope.  

With an eye toward these applications, though, we give a version of \thmref{t:spanprogramspectralanalysis} that applies to the graphs $G_P(x)$ directly instead of to $G_{P'}(x)$:   

\begin{theorem} \label{t:spanprogramspectralanalysisnonblackbox}
Let $P$ be a span program, and for $x \in \B^n$ let $G_P(x)$ be the weighted bipartite graph from \defref{t:spanprogramadjacencymatrix}.  Then for $x \in \B^n$: 
\begin{itemize}
\item
If $f_P(x) = 1$, let $\ket w \in \C^I$ be a witness, i.e., $A \Pi(x) \ket w = \ket t$.  Then $A_{G_P(x)}$ has an eigenvalue-zero eigenvector $\ket \psi$ with 
\begin{equation} \label{e:spanprogramspectralanalysisnonblackboxtrue}
\frac{\abs{ \braket 0 \psi }^2}{\norm{\ket \psi}^2} \geq \frac{1}{1 + \norm{\ket w}^2}
 \enspace .
\end{equation}
\item
If $f_P(x) = 0$, let $\ket{w'} \in V$ be a witness, i.e., $\braket{t}{w'} = 1$ and $\Pi(x) A^\dagger \ket{w'} = 0$.  Let $\{ \ket \alpha \}$ be a complete set of orthonormal eigenvectors of $A_{G_P(x)}$, with corresponding eigenvalues $\rho(\alpha)$.  Then for any $\Upsilon \geq 0$, 
\begin{equation} \label{e:spanprogramspectralanalysisnonblackboxfalse}
\sum_{\alpha : \, \abs{\rho(\alpha)} \leq \Upsilon} \abs{\braket \alpha 0}^2 \leq 8 \Upsilon^2 \Big( \norm{\ket{w'}}^2 + \norm{A^\dagger \ket{w'}}^2 \Big)
 \enspace .
\end{equation}
\end{itemize}
\end{theorem}

A typical application of \thmref{t:spanprogramspectralanalysisnonblackbox} will start with a span program having witnesses in the true and false cases satisfying 
\begin{equation} \label{e:wsizenewdef}
\max \big\{ \max_{x : f_P(x) = 1} \norm{\ket w}^2, \max_{x : f_P(x) = 0} (\norm{\ket{w'}}^2 + \norm{A^\dagger \ket{w'}}^2) \big\} \leq W
 \enspace ,
\end{equation}
for some $W$.  Scale the target vector down by a factor of $1/ \sqrt W$, and apply \thmref{t:spanprogramspectralanalysisnonblackbox}; Eq.~\eqnref{e:spanprogramspectralanalysisnonblackboxtrue} then holds with $1/2$ on the right-hand side, and letting $\Upsilon = c / W$ the right-hand side of Eq.~\eqnref{e:spanprogramspectralanalysisnonblackboxfalse} becomes $8 c^2$.  See \thmref{t:generalspanprogramalgorithmnonblackbox}.  

Although the upper bounds in Eqs.~\eqnref{e:spanprogramspectralanalysisnonblackboxtrue} and~\eqnref{e:spanprogramspectralanalysisnonblackboxfalse} depend on quantities, $\norm{\ket w}^2$ and $(\norm{\ket{w'}}^2 + \norm{A^\dagger \ket{w'}}^2)$, similar to the witness size, for two reasons they are not the same as the witness size.  
\begin{itemize}
\item
First, in the case $f_P(x) = 1$, $\norm{\ket w}^2$ can be greater than $\wsizex P x$ if $P$ is not strict (\defref{t:strictmonotonerealdef}), because the witness size does not count the portion of $\ket w$ supported on $\Ifree$.  
\item
Second, in the case $f_P(x) = 0$, while it is true that $\norm{A^\dagger \ket{w'}}{}^2$ can be bounded by $\wsizex P x$, the term $\norm{\ket{w'}}^2$ is not necessarily so-bounded.  This is clear because simultaneously scaling the target and all input vectors by $c > 0$ leaves the witness size invariant (\lemref{t:wsizescaling}) but multiplies $\norm{\ket{w'}}^2$ by $1/c^2$.  The effective spectral gap of $A_{G_P(x)}$ certainly should not be invariant under such scaling, and should approach zero as $c$ approaches zero.  
\end{itemize}
\thmref{t:spanprogramspectralanalysisnonblackbox} therefore motivates using $W$ in Eq.~\eqnref{e:wsizenewdef} as a new span program complexity measure.  This measure is important for developing time-efficient quantum algorithms based on span programs, as in for example \thmref{t:formulaevaluation}.  

The proof of \thmref{t:spanprogramspectralanalysisnonblackbox} will also be given below in \secref{s:spanprogramspectralanalysisproof}.

\subsection{Eigenvalue-zero spectral analysis of the graphs $G_P(x)$} \label{s:bipartitezero}

We will begin by analyzing Eqs.~\eqnref{e:zeroenergyequations} at eigenvalue $\rho = 0$.  This theorem is a straightforward extension of \cite[Theorems~2.5 and~A.7]{ReichardtSpalek08spanprogram}. 

\begin{theorem}[\cite{ReichardtSpalek08spanprogram}] \label{t:zeroenergy}
For a span program $P$ and input $x \in \B^n$, consider all the eigenvalue-zero eigenvector equations of the weighted adjacency matrix $A_{G_P(x)}$, except for the constraint at the output vertex $\mu_0$, i.e., Eqs.~(\ref{e:zeroenergyequations}) except (\ref{e:zeroenergyequationsaO}) at $\rho = 0$.  

These equations have a solution $\ket \psi$ with $\psi_{U,0} \neq 0$ if and only if $f_P(x) = 1$, and have a solution $\ket \psi$ with $\braket{t}{\psi_{T,V}} \neq 0$ if and only if $f_P(x) = 0$.  More quantitatively, let $s \in [0, \infty)^n$ be a vector of nonnegative costs, and recall from \defref{t:wsizedef} that $S = \sum_{j \in [n], b \in \B, i \in I_{j,b}} \sqrt{s_j} \ketbra i i$.  Then 
\begin{itemize}
\item 
If $f_P(x) = 1$, $A_{G_P(x)}$ has an eigenvalue-zero eigenvector $\ket \psi = (0, \psi_{U,0}, \ket{\psi_{U,I}}) \in \C^T \oplus \C^{\{0\}} \oplus \C^I$ with 
\begin{equation} \label{e:zeroenergytrue}
\frac{\abs{\psi_{U,0}}^2}{ \abs{\psi_{U,0}}^2 + \norm{S \ket {\psi_{U,I}}}^2}
\geq 
\frac{1}{1 + \wsizexS P x s}
 \enspace .
\end{equation}
\item 
If $f_P(x) = 0$, let $\ket{w'} \in V$ be an optimal witness, i.e., $\braket{t}{w'} = 1$, $\Pi(x) A^\dagger \ket{w'} = 0$ and $\norm{S A^\dagger \ket{w'}}^2 = \wsizexS P x s$ (see \defref{t:wsizedef}).  Then there is a solution $\ket \psi = (\ket{\psi_{T,V}}, \ket{\psi_{T,I}}, 0) \in V \oplus \C^I \oplus \C^U$ to Eqs.~(\ref{e:zeroenergyequations}a,b,d) at $\rho = 0$, with  
\begin{equation} \label{e:zeroenergyfalse}
\frac{ \abs{\braket{t}{\psi_{T,V}}}^2 } { \norm{\ket{\psi_{T,V}}}^2 + \norm{S \ket{\psi_{T,I}}}^2 } \geq \frac 1 {\norm{\ket{w'}}^2 + \wsizexS P x s}
 \enspace .
\end{equation}
\end{itemize}
\end{theorem}

\begin{proof}
Let $\rho = 0$.  Since $G_P(x)$ is bipartite, the $\psi_T$ terms do not interact with the $\psi_U$ terms.  In particular, Eqs.~(\ref{e:zeroenergyequations}c,d) (resp. \ref{e:zeroenergyequations}a,b) can always be satisfied by setting the $\psi_T$ (resp.~$\psi_U$) terms to zero.  Fix $s \in [0, \infty)^n$.  

Now Eqs.~(\ref{e:zeroenergyequations}a,b) are equivalent to $-\psi_{U,0} \ket t = A \ket{\psi_{U,I}}$ and $\ket{\psi_{U,I}} = \Pi(x) \ket{\psi_{U,I}}$.  If these equations have a solution with $\psi_{U,0} \neq 0$, then $- \ket{\psi_{U,I}} / \psi_{U,0}$ is a witness for $f_P(x) = 1$.  Conversely, if $f_P(x) = 1$, then let $\ket w \in \C^I$ be an optimal witness, satisfying $A \Pi(x) \ket w = \ket t$ and $\wsizexS P x s = \norm{S \ket w}^2$.  Let $\psi_{U,0} = -1$ and $\ket{\psi_{U,I}} = \Pi(x) \ket w$.  Then $\ket \psi = (0, \psi_{U,0}, \ket{\psi_{U,I}})$ satisfies Eqs.~\eqnref{e:zeroenergyequations}, and Eq.~\eqnref{e:zeroenergytrue} with equality.  

Next, assume that $\ket \psi$ solves Eq.~\eqnref{e:zeroenergyequationsaJ} with $\braket{t}{\psi_{T,V}} \neq 0$.  Then $\Pi(x) A^\dagger \ket{\psi_{T,V}} = - \Pi(x) \Pim(x) \ket{\psi_{T,I}} = 0$, so $\ket{\psi_{T,V}} / \braket{t}{\psi_{T,V}}$ is a witness for $f_P(x) = 0$.  Conversely, assume that $f_P(x) = 0$ and let $\ket{w'}$ be an optimal witness.  Let $\ket{\psi_{T,V}} = \ket{w'}$ and $\ket{\psi_{T,I}} = -A^\dagger \ket{w'}$.  Then $\ket \psi = (\ket{\psi_{T,V}}, \ket{\psi_{T,I}}, 0)$ satisfies Eqs.~(\ref{e:zeroenergyequations}a,b,d), and Eq.~\eqnref{e:zeroenergyfalse} with equality.  
\end{proof}

Note that if the costs are uniform $s = \vec 1$, then $S = \identity - \sum_{i \in \Ifree} \ketbra i i$, so $\norm{S \ket{\psi_{U,I}}}^2 \leq \norm{\ket{\psi_{U,I}}}^2$ and $\norm{S \ket{\psi_{T,I}}}^2 \leq \norm{\ket{\psi_{T,I}}}^2$.  If $P$ is also a strict span program, i.e., $\Ifree = \emptyset$, then $S = \identity$ so both these inequalities are equalities, and the denominators on the left-hand sides of Eqs.~\eqnref{e:zeroenergytrue} and~\eqnref{e:zeroenergyfalse} are, respectively, $\norm{\ket{\psi_U}}^2$ and $\norm{\ket{\psi_T}}^2$.  However, if $P$ is not strict, then Eqs.~\eqnref{e:zeroenergytrue} and~\eqnref{e:zeroenergyfalse} do not imply lower bounds on achievable $\abs{\psi_{U,0}}^2 / \norm{\ket{\psi_U}}^2$ or $\abs{\braket{t}{\psi_{T,V}}}^2 / \norm{\ket{\psi_T}}^2$.  
 
\begin{corollary} \label{t:zeroenergyuniformcosts}
Let $P$ be a span program.  Then there exists a span program $\hat P$ that computes $f_{\hat P} = f_P$, and such that, for all $x \in \B^n$, 
\begin{itemize}
\item
If $f_P(x) = 1$, then there is an eigenvalue-zero eigenvector $\ket \psi$ of $A_{G_{\hat P}(x)}$ with 
\begin{equation} \label{e:zeroenergyuniformcoststrue}
\frac{\abs{\psi_{U,0}}^2}{\norm{\ket \psi}^2}
\geq 
\frac{1}{1 + \wsizex P x}
 \enspace .
\end{equation}
\item
If $f_P(x) = 0$, then there is a solution $\ket \psi$ to all the eigenvalue-zero eigenvector equations of $A_{G_{\hat P}(x)}$, except for the constraint at vertex $\mu_0$, with  
\begin{equation} \label{e:zeroenergyuniformcostsfalse}
\frac{\abs{\braket{t}{\psi_{T,V}}}^2}{\norm{\ket \psi}^2}
\geq 
\frac{1}{1 + \wsizex P x}
 \enspace .
\end{equation}
\end{itemize}
\end{corollary}

\begin{proof}
Let $\hat P$ be the canonical span program constructed in \thmref{t:spanprogramcanonical} for costs $s = \vec 1$, with $\wsizex{\hat P}{x} \leq \wsizex P x$ for all $x \in \B^n$.  
$\hat P$ is in particular strict, so Eq.~\eqnref{e:zeroenergyuniformcoststrue} follows from Eq.~\eqnref{e:zeroenergytrue}.  

For showing Eq.~\eqnref{e:zeroenergyuniformcostsfalse}, recall from \thmref{t:spanprogramcanonical} that an optimal witness $\ket{w'}$ for $f_{\hat P}(x) = 0$ may be taken to be $\ket x$ itself, so $\norm{\ket{w'}}^2 = 1$ in Eq.~\eqnref{e:zeroenergyfalse}.  
\end{proof}

This completes the eigenvalue-zero analysis of the graphs $G_P(x)$.

\subsection{Small-eigenvalue spectral analysis of the graphs $G_P(x)$} \label{s:bipartitesmall}

\thmref{t:zeroenergy} implies in particular that when $f_P(x) = 0$, $A_{G_P(x)}$ does not have any eigenvalue-zero eigenvectors supported on the output vertex $\mu_0$.  Therefore $A_{G_P(x)}$ has some spectral gap around zero for eigenvectors overlapping $\ket 0$.  This spectral gap can be arbitrarily small, though, because $G_P(x)$ can be a very large graph and its edge weights are poorly constrained.  In fact, though, the lower bound Eq.~\eqnref{e:zeroenergyfalse} can be translated into a good lower bound on an ``effective" spectral gap.  That is, we can upper-bound the total squared support of $\ket 0$ on small-magnitude eigenvalues of $A_{G_P(x)}$.  

The main result of this section is: 

\begin{theorem} \label{t:bipartitepsdreduction}
Let $G$ be a weighted bipartite graph with biadjacency matrix $\biadj_G \in \L(\C^U, \C^T)$.  Assume that for some $\delta > 0$ and $\ket t \in \C^T$, the weighted adjacency matrix $A_G$ has an eigenvalue-zero eigenvector $\ket{\psi}$ with 
\begin{equation}
\abs{\braket{t}{\psi_T}}^2 \geq \delta \norm{\ket \psi}^2
 \enspace .
\end{equation}
Let $G'$ be the same as $G$ except with a new vertex, $\mu_0$, added to the $U$ side, and for $i \in T$ the new edge $(\tau_i, \mu_0)$ weighted by $\braket i t$.  That is, the biadjacency matrix of $G'$ is 
\begin{equation}
\biadj_{G'} = \left( \begin{matrix} \ket t & \biadj_G \end{matrix} \right)
\place{{\small $0$}}{-65mu}{12pt}
\place{{\small U}}{-27mu}{12pt}
\place{{\small T}}{3mu}{0pt}
\end{equation}

Recall that $\ket 0 = (0, 1, 0) \in \C^T \oplus \C^{\{0\}} \oplus \C^U$.  Let $\{ \ket \alpha \}$ be a complete set of orthonormal eigenvectors of $A_{G'}$, with corresponding eigenvalues $\rho(\alpha)$.  Then for all $\Upsilon \geq 0$, the squared length of the projection of $\ket 0$ onto the span of the eigenvectors $\alpha$ with $\abs{\rho(\alpha)} \leq \Upsilon$ satisfies 
\begin{equation} \label{e:bipartitepsdreduction}
\sum_{\alpha : \, \abs{\rho(\alpha)} \leq \Upsilon} \abs{\braket \alpha 0}^2 \leq 8 \Upsilon^2 / \delta
 \enspace .
\end{equation}
\end{theorem}

This theorem applies to the case of a strict span program $P$ with $f_P(x) = 0$, by letting $G = G_P(x)$ and, from Eq.~\eqnref{e:zeroenergyfalse} with $s = \vec 1$, $\delta = 1/ (\norm{\ket{w'}}^2 + \wsizex P x)$.  

To motivate our approach to proving \thmref{t:bipartitepsdreduction}, let us recall some basic properties about the eigenvalues and eigenvectors of bipartite graphs.  

\begin{proposition} \label{t:bipartitespectrum}
Let $G$ be a weighted bipartite graph with biadjacency matrix $\biadj_G$ and adjacency matrix~$A_G$.  

Assume that $\ket \psi = (\ket{\psi_T}, \ket{\psi_U}) \in \C^T \oplus \C^U$ is an eigenvalue-$\rho$ eigenvector of $A_G$, for some $\rho \neq 0$.  Then $(\ket{\psi_T}, -\ket{\psi_U})$ is an eigenvector of $A_G$ with eigenvalue $-\rho$.  Moreover, $\ket{\psi_T} = \frac{1}{\rho} \biadj_G \ket{\psi_U}$ is an eigenvector of $\biadj_G \biadj_G^\dagger$ and $\ket{\psi_U} = \frac{1}{\rho} \biadj_G^\dagger \ket{\psi_T}$ is an eigenvector of $\biadj_G^\dagger \biadj_G$, both with corresponding eigenvalues $\rho^2$.  

Conversely, if $\ket \varphi \in \C^T$ is an eigenvalue-$\lambda$ eigenvector of $\biadj_G \biadj_G^\dagger$ for $\lambda > 0$, then $\biadj_G \ket \varphi \in \C^U$ is an eigenvalue-$\lambda$ eigenvector of $\biadj_G^\dagger \biadj_G$ and $\ket{\psi_\pm} = (\ket \varphi, \pm \frac{1}{\sqrt \lambda} \biadj_G^\dagger \ket \varphi) \in \C^T \oplus \C^U$ are eigenvectors of $A_G$ with corresponding eigenvalues $\pm \sqrt{\lambda}$.  
\end{proposition}

The proof is immediate.  

Thus the spectrum of $A_G$ is symmetrical around zero, and nonzero-eigenvalue eigenvectors of the positive semi-definite matrix $\biadj_G \biadj_G^\dagger$ are in exact correspondence to symmetrical pairs of nonzero-eigenvalue eigenvectors of $A_G$.  

\propref{t:bipartitespectrum} allows us to translate the claims of \thmref{t:bipartitepsdreduction} into claims on spectral properties of positive semi-definite matrices.  We will start, though, by proving the necessary result for positive semi-definite matrices, \thmref{t:psdsquaredsupport} below.  After proving \thmref{t:psdsquaredsupport}, we will give the translation to prove \thmref{t:bipartitepsdreduction}.  

\begin{theorem} \label{t:psdsquaredsupport}
Let $X \in \L(V)$ be a positive semi-definite matrix, $\ket t \in V$ a vector, and let $X' = X + \ketbra t t$.  Let $\{ \ket \beta \}$ be a complete set of orthonormal eigenvectors of $X'$, with corresponding eigenvalues $\lambda(\beta) \geq 0$.  Assume that there exists a $\ket \varphi \in \Kernel(X)$ with $\abs{\braket t \varphi}^2 \geq \delta \norm{\ket \varphi}^2$.  Then for any $\Lambda \geq 0$, 
\begin{equation} \label{e:psdsquaredsupport}
\delta \sum_{\substack{\beta : \, \lambda(\beta) \leq \Lambda \\ \braket t \beta \neq 0}} \frac{1}{\lambda(\beta)} \abs{\braket t \beta}^2 \leq 4 \Lambda
 \enspace .
\end{equation}
\end{theorem}

\begin{proof}
The sum is well-defined, with no division by zero, because any $\ket \beta$ with $\braket t \beta \neq 0$ must have $\lambda(\beta) = \bra \beta X' \ket \beta = \bra \beta X \ket \beta + \abs{\braket t \beta}^2 > 0$.  

The key lemma for proving \thmref{t:psdsquaredsupport} is: 

\begin{lemma} \label{t:psdsquaredsupportmaximize}
Under the conditions of \thmref{t:psdsquaredsupport}, for any $\ket \xi \in V$, 
\begin{equation} \label{e:psdsquaredsupportmaximize}
\delta \abs{\braket t \xi}^2 \leq \norm{X' \ket \xi}{}^2
 \enspace .
\end{equation}
Moreover, if $\ket \xi$ is a linear combination of eigenvectors with corresponding eigenvalues at most $\kappa$, i.e., $\ket \xi = \sum_{\beta : \lambda(\beta) \leq \kappa} \braket \beta \xi \ket \beta$, then 
\begin{equation} \label{e:psdsquaredsupportmaximizeplugin}
\delta \abs{\braket t \xi}^2 \leq \kappa^2 \norm{\ket \xi}^2
 \enspace .
\end{equation}
\end{lemma}

\begin{proof}
We will write the matrices $X$ and $X'$ out in coordinates.  Fixing $\braket t \xi$, we will use straightforward calculus to minimize $\norm{X' \ket \xi}^2$.  

Let $\ket 1, \ldots, \ket m$ be a complete, orthonormal set of eigenvectors for $\Big( \identity - \frac{\ketbra t t}{\norm{\ket t}^2} \Big) X \Big( \identity - \frac{\ketbra t t}{\norm{\ket t}^2} \Big)$, with corresponding eigenvalues $a_1, a_2, \ldots, a_m$.  In the coordinates $\big( \frac{\ket t}{\norm{\ket t}}, \ket 1, \ldots, \ket m \big)$, $X$ and $X'$ are given by 
\begin{align}
X &= \left( \begin{matrix} 
a & \bar b_1 & \ldots & \bar b_m \\
b_1 & a_1 & & 0 \\
\vdots & & \ddots & \\
b_m & 0 & & a_m
\end{matrix} \right) \\
X' &= \left( \begin{matrix} 
a + \norm{\ket t}^2 & \bar b_1 & \ldots & \bar b_m \\
b_1 & a_1 & & 0 \\
\vdots & & \ddots & \\
b_m & 0 & & a_m
\end{matrix} \right)
\end{align}
where $a = \bra t X \ket t / \norm{\ket t}^2$ and $b_j = \bra{a_j} X \frac{\ket t}{\norm{\ket t}}$, for $j \in [m]$.  

By incorporating any phases into the basis vectors $\ket j$, we may assume that all $b_j \geq 0$.  Furthermore, we may assume without loss of generality that all $b_j > 0$.  Indeed, if some $b_j = 0$, then the $\ket j$ coordinate lies in a different block of $X'$ from $\ket t$, so removing this coordinate will not affect $\min_{\ket \psi} \norm{X' \ket \psi} / \abs{\braket t \psi}$.  Since $X \succeq 0$, all $a_j \geq 0$.  Moreover, if some $a_j = 0$, then since $\big(\begin{smallmatrix} a & b_j \\ b_j & 0 \end{smallmatrix}\big)$ is a (positive semi-definite) submatrix of $X$, it must be that $b_j = 0$.  Hence we may assume that $a_j > 0$ for all $j \in [m]$.  

We are given the existence of a $\ket \varphi \in \Kernel(X)$ with $\abs{\braket t \varphi}^2 \geq \delta \norm{\ket \varphi}^2$.  Let us write out this condition in coordinates.  By scaling $\ket \varphi$, we may assume that $\braket t \varphi = \norm{\ket t}$.  Thus, written in coordinates, $\ket \varphi = (1, -\frac{b_1}{a_1}, \ldots, -\frac{b_m}{a_m})$ and $\bra t X \ket \varphi = 0$ implies that 
\begin{equation} \label{e:coordinatesa}
a = \sum_{j=1}^m b_j^2 / a_j
 \enspace .
\end{equation}
The condition $\abs{\braket t \varphi}^2 \geq \delta \norm{\ket \varphi}^2$, in coordinates, is
\begin{equation}
\norm{\ket t}^2 \geq \delta \Big( 1 + \sum_{j=1}^m \Big(\frac{b_j}{a_j}\Big)^2 \Big)
 \enspace .
\end{equation}

We can now solve the minimization problem: 

\begin{claim}
\begin{equation}
\min_{\ket \xi : \, \braket t \xi = \norm{\ket t}} \norm{X' \ket \xi}^2 = \frac{ \norm{\ket t}^4 }{ 1 + \sum_j \big(\frac{b_j}{a_j}\big)^2 } \geq \delta \norm{\ket t}^2
 \enspace .
\end{equation}
\end{claim}

\begin{proof}
Since $X'$ is a symmetric matrix, we may assume that $\ket \xi$ has real coordinates.  Introduce variables $c_1, \ldots, c_m$ and let $\ket \xi = (1, c_1, \ldots, c_m)$.  For $j \in [m]$, let $\gamma_j = a_j \big(\frac{a_j}{b_j} c_j + 1\big)$.  Then 
\begin{align} \label{e:psdsquaredsupportmaximizeXprimexinorm}
\norm{X' \ket \xi}^2
&= 
\big( a + \norm{\ket t}^2 + \sum_j b_j c_j \big)^2 + \sum_j (b_j + a_j c_j)^2 \nonumber \\
&=
\bigg(a + \norm{\ket t}^2 + \sum_j \frac{b_j^2}{a_j} \Big( \frac{\gamma_j}{a_j} - 1 \Big) \bigg)^2 + \sum_j \Big( \frac{b_j}{a_j} \gamma_j \Big)^2 \nonumber \\
&= 
\bigg( \norm{\ket t}^2 + \sum_j \Big(\frac{b_j}{a_j}\Big)^2 \gamma_j \bigg)^2 + \sum_j \Big( \frac{b_j}{a_j} \Big)^2 \gamma_j^2
 \enspace ,
\end{align}
where we have substituted $c_j = \frac{b_j}{a_j}\big( \frac{\gamma_j}{a_j}-1 \big)$ and then used Eq.~\eqnref{e:coordinatesa} to cancel $a$ from the first term.  

A global minimum exists and will satisfy, for all $j \in [m]$, 
\begin{equation}\begin{split}
0
&= 
\frac{\partial}{\partial \gamma_j} \norm{X' \ket \xi}^2 \\
&=
2 \Big(\frac{b_j}{a_j}\Big)^2 \bigg( \gamma_j + \norm{\ket t}^2 + \sum_k \Big( \frac{b_k}{a_k} \Big)^2 \gamma_k \bigg)
 \enspace .
\end{split}\end{equation}
Thus we should set all $\gamma_j$ equal, $\gamma_j = \gamma$ for $j \in [m]$, where $\gamma = - \norm{\ket t}^2  / (1 + S)$ and $S = \sum_j \big(\frac{b_j}{a_j}\big)^2$.  Substituting back into Eq.~\eqnref{e:psdsquaredsupportmaximizeXprimexinorm}, $\norm{X' \ket \xi}^2$ at the minimum is 
\begin{align}
\norm{X' \ket \xi}^2 
&= ( \norm{\ket t}^2 + S \gamma )^2 + S \gamma^2 \nonumber \\
&= \norm{\ket t}^4 / (1 + S) 
 \enspace ,
\end{align}
as claimed.
\end{proof}
Eq.~\eqnref{e:psdsquaredsupportmaximize} follows.  Eq.~\eqnref{e:psdsquaredsupportmaximizeplugin} is an immediate consequence of Eq.~\eqnref{e:psdsquaredsupportmaximize}, since $\ket \xi = \sum_{\beta : \lambda(\beta) \leq \kappa} \braket \beta \xi \ket \beta$ implies $\norm{X' \ket \xi} \leq \kappa \norm{\ket \xi}$.  This completes the proof of \lemref{t:psdsquaredsupportmaximize}.  
\end{proof}

Now let us derive Eq.~\eqnref{e:psdsquaredsupport} by bootstrapping \lemref{t:psdsquaredsupportmaximize}.  We aim to bound
\begin{align}
\delta \sum_{\substack{\beta : \, \lambda(\beta) \leq \Lambda \\ \braket t \beta \neq 0}} \frac{1}{\lambda(\beta)} \abs{\braket t \beta}^2
&= 
\delta \sum_{k = 0}^\infty \sum_{\frac{\Lambda}{2^{k+1}} < \lambda(\beta) \leq \frac{\Lambda}{2^k}} \frac{1}{\lambda(\beta)} \abs{\braket t \beta}^2 \nonumber \\
&\leq 
\frac \delta \Lambda \sum_{k = 0}^\infty 2^{k+1} \sum_{\frac{\Lambda}{2^{k+1}} < \lambda(\beta) \leq \frac{\Lambda}{2^k}} \abs{\braket t \beta}^2 \nonumber \\
&= 
\frac \delta \Lambda \sum_{k = 0}^\infty 2^{k+1} \braket{t}{t_k}
 \enspace ,
\intertext{where $\ket{t_k} = \sum_{\beta : \frac{\Lambda}{2^{k+1}} < \lambda(\beta) \leq \frac{\Lambda}{2^k}} \braket \beta t \ket \beta$, the projection of $\ket t$ onto the span of the eigenvectors with eigenvalues in $\big(\frac{\Lambda}{2^{k+1}}, \frac{\Lambda}{2^k}\big]$.  Therefore $\braket{t}{t_k} = \braket{t_k}{t_k} = \abs{\braket{t}{t_k}}^2 / \norm{\ket{t_k}}^2$ when $\ket{t_k} \neq 0$, so Eq.~\eqnref{e:psdsquaredsupportmaximizeplugin} can be applied with $\ket \xi = \ket{t_k}$ and $\kappa = \Lambda / 2^k$ to continue:}
\delta \sum_{\substack{\beta : \, \lambda(\beta) \leq \Lambda \\ \braket t \beta \neq 0}} \frac{1}{\lambda(\beta)} \abs{\braket t \beta}^2
&\leq 
\frac 1 \Lambda \sum_{k = 0}^\infty 2^{k+1} \Big( \frac{\Lambda}{2^k} \Big)^2 \nonumber \\
&= 
2 \Lambda \sum_{k = 0}^\infty \frac{1}{2^k} \nonumber \\
&= 
4 \Lambda 
 \enspace ,
\end{align}
as claimed.  
%% NOTE: In general, this expansion can be made for any ratio r instead of r=1/2.  But 1/2 is best.  
\end{proof}

With \thmref{t:psdsquaredsupport} in hand, we can now apply \propref{t:bipartitespectrum} to prove \thmref{t:bipartitepsdreduction}.  

\begin{proof}[Proof of \thmref{t:bipartitepsdreduction}]
We are given an eigenvalue-zero eigenvector of $A_G$, $(\ket{\psi_T}, 0) \in \C^T \oplus \C^U$ with $\abs{\braket{t}{\psi_T}}^2 \geq \delta \norm{\ket{\psi_T}}^2$.  In particular, $\biadj_G^\dagger \ket{\psi_T} = 0$.  

An eigenvalue-zero eigenvector $\ket \zeta = (\ket{\zeta_T}, \zeta_0, \ket{\zeta_U}) \in \C^T \oplus \C^{\{0\}} \oplus \C^U$ has to satisfy 
\begin{equation}\begin{split}
0 &= \biadj_{G'} (\zeta_0, \ket{\zeta_U}) \\
&= \zeta_0 \ket t + \biadj_G \ket{\zeta_U} \enspace .
\end{split}\end{equation}
Since $\abs{\braket{t}{\psi_T}}^2 > 0$ and $\biadj_G^\dagger \ket{\psi_T} = 0$, $\ket t$ cannot lie in the range of $\biadj_G$, so $\zeta_0$ must be zero.  Thus follows the claim for $\Upsilon = 0$, that $A_{G'}$ has no eigenvalue-zero eigenvectors supported on $\mu_0$.  

Now to show Eq.~\eqnref{e:bipartitepsdreduction} for $\Upsilon > 0$, note that for each eigenvector $\ket \alpha$ of $A_{G'}$, $\rho(\alpha) \braket 0 \alpha = \bra 0 A_{G'} \ket \alpha = \braket{t}{\alpha_T}$.  Therefore 
\begin{equation}
\sum_{\alpha : \, \abs{\rho(\alpha)} \leq \Upsilon} \abs{\braket \alpha 0}^2
=
\sum_{\alpha : \, 0 < \abs{\rho(\alpha)} \leq \Upsilon} \frac{1}{\rho(\alpha)^2} \abs{\braket{t}{\alpha_T}}^2
 \enspace .
\end{equation}
Let $X' = \biadj_{G'} \biadj_{G'}^\dagger$.  Let $\{ \ket \beta \}$ be a complete set of orthonormal eigenvectors of $X'$, with corresponding eigenvalues $\lambda(\beta)$.  By \propref{t:bipartitespectrum}, each eigenvector $\ket \beta$ with $\lambda(\beta) \neq 0$ corresponds to a pair of eigenvectors of $A_{G'}$ with eigenvalues $\pm \sqrt{\lambda(\beta)}$.  The above sum therefore equals 
\begin{equation}
2 \sum_{\substack{\beta : \, 0 < \lambda(\beta) \leq \Upsilon^2}} \frac{1}{\lambda(\beta)} \abs{\braket t \beta}^2
 \enspace .
\end{equation}

Now apply \thmref{t:psdsquaredsupport} with $X = X' - \ketbra t t = \biadj_G \biadj_G^\dagger \succeq 0$, $\ket \varphi = \ket{\psi_T}$ and $\Lambda = \Upsilon^2$, to obtain the claimed upper bound of $8 \Upsilon^2 / \delta$.  
\end{proof}

\subsection{Proofs of \texorpdfstring{\thmref{t:spanprogramspectralanalysis}}{Theorem~\ref{t:spanprogramspectralanalysis}} and \texorpdfstring{\thmref{t:spanprogramspectralanalysisnonblackbox}}{Theorem~\ref{t:spanprogramspectralanalysisnonblackbox}}} \label{s:spanprogramspectralanalysisproof}

Let us now combine \thmref{t:zeroenergy} and \thmref{t:bipartitepsdreduction} to prove \thmref{t:spanprogramspectralanalysis} and \thmref{t:spanprogramspectralanalysisnonblackbox}.  The proof of \thmref{t:spanprogramspectralanalysis} will also use the canonical span program reduction, \thmref{t:spanprogramcanonical}.  

\begin{proof}[Proof of \thmref{t:spanprogramspectralanalysis}]
Let $\hat P$ be the canonical span program constructed in \thmref{t:spanprogramcanonical} for costs $s = \vec 1$, with $\wsizex{\hat P}{x} \leq \wsizex P x$ for all $x \in \B^n$.  In particular, recall that when $f_P(x) = 0$, an optimal witness $\ket{w'}$ may be taken to be $\ket x$ itself.  Also, $\hat P$ is strict, i.e., has $\Ifree = \emptyset$, so $S$ is the identity on $\C^I$.    

Let $P'$ be the same as $\hat P$ except with the the target vector scaled by a factor of $1/ \sqrt{\wsizeD P}$.  Thus $f_{P'} = f_P$ still, and, for all $x \in \D$, 
\begin{equation} \label{e:spanprogramspectralanalysiswsizeamplified}
\wsizex{P'}{x} \leq \begin{cases} 1 & \text{if $f_P(x) = 1$} \\ \wsizeD{P}^2 & \text{if $f_P(x) = 0$} \end{cases}
\end{equation}
Now, when $f_{P'}(x) = 0$, an optimal witness is $\ket{w'} = \sqrt{\wsizeD P} \ket x$.  This scaling step is known as amplification.  It was introduced by~\cite{ChildsReichardtSpalekZhang07andor} and also applied in~\cite{AmbainisChildsReichardtSpalekZhang07andor, ReichardtSpalek08spanprogram}.  

For the case $f_P(x) = 1$, the first part of \thmref{t:spanprogramspectralanalysis}, Eq.~\eqnref{e:spanprogramspectralanalysistrue}, now follows from Eqs.~\eqnref{e:zeroenergytrue} and~\eqnref{e:spanprogramspectralanalysiswsizeamplified}; since $S = \identity$, $\norm{\ket \psi}^2 = \abs{\psi_{U,0}}^2 + \norm{S \ket{\psi_{U,I}}}^2$.   

For the case $f_P(x) = 0$, let $G$ be the graph $G_P(x)$ with the output vertex $\mu_0$ and all incident edges deleted.  Thus $G$'s biadjacency matrix is the same as $\biadj_{G_P(x)}$ from Eq.~\eqnref{e:spanprogramadjacencymatrix}, except with the $\mu_0$ column deleted.  
\thmref{t:zeroenergy} implies that $A_G$ has an eigenvalue-zero eigenvector $\ket \psi = (\ket{\psi_{T,V}}, \ket{\psi_{T,I}}, 0) \in V \oplus \C^I \oplus \C^I$ satisfying 
\begin{equation}\begin{split}
\frac{ \abs{\braket{t}{\psi_{T,V}}}^2 } { \norm{\ket{\psi}}^2 } 
&\geq \frac 1 {\norm{\ket{w'}}^2 + \wsizex {P'} x} \\
&\geq \frac 1 {\wsizeD P (\wsizeD P + 1)}
\end{split}\end{equation}
by Eqs.~\eqnref{e:zeroenergyfalse} and~\eqnref{e:spanprogramspectralanalysiswsizeamplified}.  Eq.~\eqnref{e:spanprogramspectralanalysisfalse} now follows by Eq.~\eqnref{e:bipartitepsdreduction} in \thmref{t:bipartitepsdreduction} with $G' = G_P(x)$, $\Upsilon = c / \wsizeD P$ and $\delta = 1 / \big(\wsizeD P (\wsizeD P + 1)\big)$.  
\end{proof}

\begin{proof}[Proof of \thmref{t:spanprogramspectralanalysisnonblackbox}]
The idea is that we want to charge for the free input vectors of $P$.  Let $P'$ be a strict span program that is the same as $P$ except with one extra input bit, and with the free input vectors of $P$ now labeled by $(n+1, 1)$.  That is, $I_{j,b}' = I_{j,b}$ for $j \in [n]$ and $b \in \B$, but $\Ifree' = I_{n+1,0}' = \emptyset$ and $I_{n+1,1}' = \Ifree$.  Then for all $x \in \B^n$, $f_{P'}(x,1) = f_P(x)$, with the same witnesses, and $G_{P'}(x,1) = G_P(x)$.  The only difference is that in the case $f_P(x) = 1$, $\wsizex{P'}{x} = \min_{\ket w : A \Pi(x) \ket w = \ket t} \norm{\ket w}^2$ counts the portion of $\ket w$ on indices in $\Ifree$, while $\wsizex P x$ does not.  

The proof now follows the same steps as the proof of \thmref{t:spanprogramspectralanalysis}, except with $\delta = 1 / (\norm{\ket{w'}}{}^2 + \norm{A^\dagger \ket{w'}}{}^2)$ in the case $f_P(x) = 0$.  
\end{proof}

\section{Quantum algorithm for evaluating span programs} \label{s:algorithm}

In this section, we will connect quantum query algorithms to the graph spectral properties that are the conclusions of \thmref{t:spanprogramspectralanalysis} and \thmref{t:spanprogramspectralanalysisnonblackbox}.  The following theorem gives two quantum algorithms for evaluating a total or partial boolean function $f$ based on promised spectral properties of a family of graphs $\{ G(x) : x \in \D \}$, with $\D \subseteq \B^n$.  

\begin{theorem} \label{t:phaseestimationalgorithm}
Let $G = (V,E)$ be a complex-weighted graph with Hermitian weighted adjacency matrix $A_G \in \L(\C^V)$ satisfying $\bra v A_G \ket v \geq 0$ for all $v \in V$.  Let $V_{\text{input}}$ be a subset of degree-one vertices of $G$ whose incident edges have weight one, and partition $V_{\text{input}}$ as $V_{\text{input}} = \bigsqcup_{j \in [n], b \in \B} V_{j,b}$.  For $x \in \B^n$, define $G(x)$ from $G$ by deleting all edges to vertices in $\cup_{j \in [n]} V_{j, x_j}$.  Let $A_{G(x)} \in \L(\C^V)$ be the weighted adjacency of matrix of $G(x)$.  

Let $f : \D \rightarrow \B$, with $\D \subseteq \B^n$, $\mu \in V \smallsetminus V_{\text{input}}$, $\epsilon = \Omega(1)$ and $\Lambda > 0$.  Assume that for all $x \in \D$ the graphs $G(x)$ satisfy:
\begin{itemize}
\item
If $f(x) = 1$, then $A_{G(x)}$ has an eigenvalue-zero eigenvector $\ket \psi \in \C^V$ with 
\begin{equation} \label{e:phaseestimationalgorithmtrue}
\frac{\abs{\braket \mu \psi}^2}{\norm{\ket \psi}} \geq \epsilon
 \enspace .
\end{equation}
\item 
If $f(x) = 0$, let $\{ \ket \alpha \}$ be a complete set of orthonormal eigenvectors of $A_{G(x)}$, with corresponding eigenvalues $\rho(\alpha)$.  Assume that 
the squared length of the projection of $\ket \mu$ onto the span of the eigenvectors $\alpha$ with $\abs{\rho(\alpha)} \leq \Lambda$ satisfies 
\begin{equation} \label{e:phaseestimationalgorithmfalse}
\sum_{\alpha : \, \abs{\rho(\alpha)} \leq \Lambda} \abs{\braket \alpha \mu}^2 \leq \epsilon / 2
 \enspace .
\end{equation}
\end{itemize}

Let $\abst(A_G)$ be the entry-wise absolute value of $A_G$, and let $\norm{\abst(A_G)}$ be its operator norm.  Then $f$ can be evaluated with error probability at most $1/3$ using at most 
\begin{equation} \label{e:phaseestimationalgorithm}
O\left( \min \bigg\{ \frac{\norm{\abst(A_G)}}{\Lambda}, \; \frac{1}{\Lambda} \frac{\log \frac{1}{\Lambda}}{\log \log \frac{1}{\Lambda}} \bigg\} \right)
\end{equation}
quantum queries.  
\end{theorem}

The intuition behind this theorem is that $f$ can be evaluated by starting at $\ket \mu$ and ``measuring" $A_{G(x)}$ to precision $\Lambda$.  (More precisely, this is implemented by applying phase estimation to a certain unitary operator.)  Output $1$ if and only if the measurement returns $0$.  Eq.~\eqnref{e:phaseestimationalgorithmtrue} implies completeness when $f(x) = 1$, because the initial state has large overlap with an eigenvalue-zero eigenstate.  Eq.~\eqnref{e:phaseestimationalgorithmfalse} implies soundness when $f(x) = 0$.

In fact, the proof of \thmref{t:phaseestimationalgorithm} requires two quantum algorithms, one for each of the bounds in Eq.~\eqnref{e:phaseestimationalgorithm}.  
\begin{enumerate}
\item
The proof that $Q(f) = O(\norm{\abst(A_G)} / \Lambda)$ is based on Szegedy's correspondence between continuous- and discrete-time quantum walks~\cite{Szegedy04walkfocs}.  The proof is nearly the same as in~\cite[Appendix~B.2]{ReichardtSpalek08spanprogram}.  The differences are that we are only assuming an effective spectral gap in the case $f(x) = 0$, and that the graph $G$ in \thmref{t:phaseestimationalgorithm} is not required to be bipartite.  The graphs to which we apply \thmref{t:phaseestimationalgorithm} below will be bipartite, though, since they will be derived from span programs.  

This algorithm applies to the formula-evaluation applications, \thmref{t:spanprogramalgorithm}, \thmref{t:querycomplexitylimit} and \thmref{t:formulaevaluation}.  In each case, a span program $P$ is given and the algorithm run with $G = G_P$.  In addition to lower-bounding $\Lambda$, the query and time complexity bounds require showing that $\norm{\abst(A_G)} = O(1)$.  
\item
The second bound, $Q(f) = \tilde O(1/\Lambda)$, is applicable in the more typical case when we do not know an upper bound on $\norm{\abst(A_G)}$.  The idea is to apply phase estimation to $e^{i A_{G(x)}}$.  Since $A_G$ is independent of the input $x$, recent work by Cleve et al.\ shows that its norm does not matter if we can concede a logarithmic factor in the query complexity~\cite{CleveGottesmanMoscaSommaYongeMallo08discretize}.  For applying phase estimation, there is still the problem that eigenvalues can wrap around the circle, e.g., $e^{2 \pi i} = e^{0 i}$, leading to false positives.  To avoid such errors, we scale $A_{G(x)}$ by a uniformly random number $R \in (0, 144/\epsilon^2)$.   
\end{enumerate}

Although \thmref{t:phaseestimationalgorithm} refers only to query complexity, and not time complexity, the first algorithm's time complexity can also often be bounded under reasonable assumptions on $G$.  See Refs.~\cite{ReichardtSpalek08spanprogram, AmbainisChildsReichardtSpalekZhang07andor, ChiangNagajWocjan09simulate} for details.  

\smallskip

For a span program $P$, the graphs $G_P$ and $G_P(x)$ from \defref{t:spanprogramadjacencymatrix} are of the form required by \thmref{t:phaseestimationalgorithm}.  The assumptions Eqs.~\eqnref{e:phaseestimationalgorithmtrue} and~\eqnref{e:phaseestimationalgorithmfalse} for \thmref{t:phaseestimationalgorithm} are also of the same type as the conclusions of \thmref{t:spanprogramspectralanalysis} and \thmref{t:spanprogramspectralanalysisnonblackbox}.  Therefore, assuming for the moment \thmref{t:phaseestimationalgorithm}, as corollaries we obtain quantum algorithms for evaluating span programs: 

\begin{theorem}
 \label{t:generalspanprogramalgorithm}
Let $P$ be a span program and $\D \subseteq \B^n$.  Then the quantum query complexity of $f_P$ restricted to $\D$ satisfies
\begin{equation}
Q(f_P\vert_\D) = 
O\bigg( \wsizeD P \frac{\log \wsizeD P}{\log \log \wsizeD P} \bigg)
 \enspace .
\end{equation}
\end{theorem}

\begin{proof}
Set $c = 1/8$ in \thmref{t:spanprogramspectralanalysis} and apply \thmref{t:phaseestimationalgorithm} with $\mu$ the output vertex $\mu_0$ of $G_P$, 
$\epsilon = 1/2$ and $\Lambda = c / \wsizeD P$.  
\end{proof}

\begin{theorem} \label{t:generalspanprogramalgorithmnonblackbox}
Let $P$ be a span program with target vector $\ket t$ and input vectors $\ket{v_i}$ for $i \in I = \Ifree \cup \bigcup_{j \in [n], b \in \B} I_{j,b}$, in inner product space $V$.  Let $\D \subseteq \B^n$ and assume that for some $W_1, W_2 \geq 1$, 
\begin{equation}\begin{split}
\max_{x \in \D : f_P(x) = 1} \min_{\substack{\ket w \in \C^I : \\ A \Pi(x) \ket w = \ket t}} \norm{\ket w}^2 
&\leq W_1 \\
\max_{x \in \D : f_P(x) = 0} \min_{\substack{\ket{w'} \in V : \, \braket{t}{w'} = 1, \\ \Pi(x) A^\dagger \ket{w'} = 0}}(\norm{\ket{w'}}^2 + \norm{A^\dagger \ket{w'}}^2) &\leq W_2
 \enspace .
\end{split}\end{equation}

Let $P'$ be the same as $P$, except with the target vector $\ket t / \sqrt{W_1}$.  Then $f_P$ can be evaluated on inputs in $\D$ using 
\begin{equation}
O\big( \sqrt{W_1 W_2} \norm{\abst(A_{G_{P'}})} \big)
= 
O\big( \sqrt{W_1 W_2} \norm{\abst(A_{G_P})} \big)
\end{equation}
quantum queries, with error probability at most $1/3$.  
\end{theorem}

\begin{proof}
Apply \thmref{t:spanprogramspectralanalysisnonblackbox} to $P'$ with $\Upsilon = \frac{1}{4 \sqrt 2} / \sqrt{W_1 W_2}$.  Then \thmref{t:phaseestimationalgorithm}'s assumptions Eqs.~\eqnref{e:phaseestimationalgorithmtrue} and~\eqnref{e:phaseestimationalgorithmfalse} hold with $\epsilon = 1/2$ and $\Lambda = \Upsilon$.  An $O\big( \sqrt{W_1 W_2} \norm{\abst(A_{G_{P'}})} \big)$-query quantum algorithm follows.  

Finally, since $W_1 \geq 1$, $\norm{\abst(A_{G_{P'}})} \leq \norm{\abst(A_{G_P})}$.  
\end{proof}

In the rest of this section, we will prove \thmref{t:phaseestimationalgorithm}, relying heavily on~\cite{ReichardtSpalek08spanprogram} and~\cite{CleveGottesmanMoscaSommaYongeMallo08discretize}.  As sketched above, there are two parts to the proof, given in \secref{s:phaseestimationalgorithmszegedyproof} and \secref{s:phaseestimationalgorithmcontinuousproof} below.  

For $x \in \B^n$, let $O_x$ be the phase-flip input oracle defined by
\begin{equation} \label{e:ox}
O_x : \ket{b, j} \mapsto (-1)^{b \, x_j} \ket{b, j}
\end{equation}
for $b \in \B$ and $j \in [n]$.

\subsection{Algorithm using the Szegedy correspondence} \label{s:phaseestimationalgorithmszegedyproof}

\begin{proposition}[\cite{ReichardtSpalek08spanprogram}] \label{t:phaseestimationalgorithmszegedyproof}
Under the assumptions of \thmref{t:phaseestimationalgorithm}, $f$ can be evaluated with error probability at most $1/3$ using $O( \norm{\abst(A_G)} / \Lambda )$ queries to the input oracle $O_x$.  
\end{proposition}

The proof is basically the same as for the algorithm in~\cite[Appendix~B.2]{ReichardtSpalek08spanprogram}, which in turn was closely based on the algorithms in~\cite{ChildsReichardtSpalekZhang07andor, AmbainisChildsReichardtSpalekZhang07andor}.  However, the arguments in~\cite{ReichardtSpalek08spanprogram} were tied to the formula-evaluation application, whereas \propref{t:phaseestimationalgorithmszegedyproof} is in a more general setting.  In particular, \cite{ReichardtSpalek08spanprogram} could assume a spectral gap in the case $f(x) = 0$, whereas we only have Eq.~\eqnref{e:phaseestimationalgorithmfalse}, an ``effective" spectral gap.  This weaker assumption means that establishing the algorithm's soundness requires somewhat more care.  

The key technical ingredient in the proof is a theorem due to Szegedy~\cite{Szegedy04walkfocs} that we apply to relate the spectrum and eigenvectors of $A_{G(x)}$ to those of a discrete-time coined quantum walk unitary.  We use a formulation of the theorem essentially the same as given in~\cite{AmbainisChildsReichardtSpalekZhang07andor}.  However, the statement there had a minor typo (in $\ket{\alpha, \pm}$ below).  This typo did not affect their application or the application in~\cite{ReichardtSpalek08spanprogram}, but would matter for us here.  Therefore, after stating the corrected theorem, we also repeat the proof from~\cite{AmbainisChildsReichardtSpalekZhang07andor}, which was correct.  

\begin{theorem}[\cite{Szegedy04walkfocs}] \label{t:szegedization}
Let $V$ be a finite set.  For each $v \in V$, let $\ket{\varphi_v} \in \C^V$ be a length-one vector.  Define $T \in \L(\C^V, \C^V \otimes \C^V)$, $S, U \in \L(\C^V \otimes \C^V)$ and $M \in \L(\C^V)$ by 
\begin{align}
T &= \sum_{v \in V} (\ket v \otimes \ket{\varphi_v}) \bra v &
S &= \sum_{v,w \in V} \ketbra{v,w}{w,v} \label{e:szegedizationTS} \\
U &= (2 T T^\dagger - \identity) S &
M &= T^\adjoint S T = \sum_{v, w \in V} \braket{\varphi_v}{w} \braket{v}{\varphi_w} \ketbra v w \label{e:szegedizationUM}
\end{align}
Since $T^\dagger T = \identity$, $U$ is a unitary.  ($U$ is a swap followed by the reflection about the span of the vectors $\{ \ket v \otimes \ket{\varphi_v} : v \in V \}$.)  $M$ is a Hermitian matrix with $\norm M \leq 1$.  Let $\{\ket \alpha\}$ be a complete set of orthonormal eigenvectors of $M$ with respective eigenvalues $\rho(\alpha)$.  

Then the spectral decomposition of $U$ corresponds to that of $M$ as follows: Let $R_\alpha = \Span\{T \ket \alpha, S T \ket \alpha\}$.  Then $R_\alpha \perp R_{\alpha'}$ for $\alpha \neq \alpha'$; let $R = \oplus_\alpha R_\alpha$.  $U$ is $-S$ on $R^\perp$, and $U$ preserves each subspace $R_\alpha$.  

If $\abs{\rho(\alpha)} < 1$, then $R_\alpha$ is two-dimensional, and within it the eigenvectors and corresponding eigenvalues of $U$ are given by 
\begin{equation}\begin{split} \label{e:szegedizationeigenvectors}
\ket{\alpha, \pm} &= \Big(\identity - \big(\rho(\alpha) \mp i \sqrt{1 - \rho(\alpha)^2}\big) S \Big) T \ket \alpha \\
\lambda(\alpha, \pm) &= \rho(\alpha) \pm i \sqrt{1 - \rho(\alpha)^2}
 \enspace .
\end{split}\end{equation}
If $\rho(\alpha) \in \{1, -1\}$, then $S T \ket \alpha = \rho(\alpha) T \ket \alpha$, so $R_\alpha$ is one-dimensional; let $\ket{\alpha, +} = T \ket \alpha$ and $\lambda(\alpha, +) = \rho(\alpha)$ be the corresponding eigenvalue of $U$.  
\end{theorem}

\begin{proof}
	This proof is taken from~\cite{AmbainisChildsReichardtSpalekZhang07andor}.  

	First assume $\alpha \neq \alpha'$, and let us show $R_\alpha \perp R_{\alpha'}$.  Indeed, $\bra \alpha T^\dagger T \ket{\alpha'} = \braket{\alpha}{\alpha'} = 0$, as $T^\dagger T = \identity$.  Since $S^2 = \identity$, similarly, $S T \ket \alpha$ is orthogonal to $S T \ket{\alpha'}$.  Finally, $\bra \alpha T^\dagger S T \ket{\alpha'} = \bra \alpha M \ket{\alpha'} = 0$.  Therefore, the decomposition $\C^V \tensor \C^V = (\bigoplus_\alpha R_\alpha) \oplus R^\perp$ is well-defined.  
	
	$R$ is the span of the images of $S T$ and $T$.  $2T T^\dagger-1$ is $+1$ on the image of $T$ and $-1$ on its complement; therefore $U$ is $-S$ on $R^\perp$.  
	
	Finally, $T T^\dagger T = T$ and $T T^\dagger S T = T M$, so 
	\begin{align*}
		U (S T \ket \alpha) &= (2 T T^\dagger - 1) T \ket \alpha = T \ket \alpha \\
		U (T \ket \alpha) &= (2 T T^\dagger - 1)S T \ket \alpha = (2 \rho(\alpha) - S) T \ket \alpha \enspace ;
	\end{align*}
	$U$ fixes the subspaces $R_\alpha$.  
	
	For the case that $\abs{\rho(\alpha)} < 1$, let $\ket \beta = (1+ \beta S) T \ket \alpha$.  Then $U \ket \beta = (2 \rho(\alpha) + \beta) T \ket \alpha - S T \ket \alpha$ is proportional to $\ket \beta$ if $\beta (2 \rho(\alpha) + \beta) = -1$; i.e., $\beta = -\rho(\alpha) \pm i \sqrt{1-\rho(\alpha)^2}$.  Eq.~\eqnref{e:szegedizationeigenvectors} follows.
	
	If $\rho(\alpha) \in \{ -1, 1 \}$, then since $(\bra \alpha T^\dagger)(S T \ket \alpha) = \bra \alpha M \ket \alpha = \rho(\alpha)$, $T \ket \alpha = \rho(\alpha) ST \ket \alpha$.  Therefore $R_\alpha$ is one-dimensional, corresponding to a single eigenvector of~$U$ with eigenvalue $\rho(\alpha)$.
\end{proof}

We will need slightly more control over the eigenvectors $\ket{\alpha, \pm}$: 

\begin{lemma} \label{t:szegedizationnormalization}
With the setup of \thmref{t:szegedization}, for any $\ket \psi \in \C^V$, the eigenvectors $\ket{\alpha, \pm}$ with $\abs{\rho(\alpha)} < 1$ satisfy $\norm{\ket{\alpha, \pm}} = \sqrt{2 (1 - \rho(\alpha)^2)}$ and 
\begin{equation} \label{e:szegedizationnormalization}
\frac{\abs{ \bra \psi T^\dagger \ket{\alpha, \pm}}^2}{\norm{\ket{\alpha, \pm}}^2} = \frac{1}{2} \abs{ \braket \psi \alpha }^2
 \enspace .
\end{equation}
When $\abs{\rho(\alpha)} = 1$, $\norm{\ket{\alpha, +}} = 1$ and $\bra \psi T^\dagger \ket{\alpha, +} = \braket \mu \alpha$.  
\end{lemma}

\begin{proof}
Fix an eigenvector $\ket \alpha$ of $A_{G(x)}$ and let $\rho = \rho(\alpha)$.  Assume that $\abs{\rho} < 1$.  We have 
\begin{equation}\begin{split}
\norm{\ket{\alpha, \pm}}^2
&= \bra \alpha T^\dagger (\identity - e^{\pm i \arccos \rho} S)(\identity - e^{\mp i \arccos \rho} S) T \ket \alpha \\
&= \bra \alpha T^\dagger 2 ( \identity - \rho S ) T \ket \alpha \\
&= 2 ( 1 - \rho \bra \alpha T^\dagger S T \ket \alpha ) \\
&= 2 (1 - \rho^2) 
 \enspace ,
\end{split}\end{equation}
where we have used $S^2 = T^\dagger T = \identity$, $\norm{\ket \alpha} = 1$, and $T^\dagger S T = M$.  
Also, then, we compute
\begin{equation}\begin{split}
\bra \psi T^\dagger \ket{\alpha, \pm} 
&= \bra \psi T ( \identity - e^{\mp i \arccos \rho} S ) T \ket \alpha \\
&= \bra \psi T^\dagger T \ket \alpha - e^{\mp i \arccos \rho} \bra \psi T^\dagger S T \ket \alpha \\
&= \braket \psi \alpha (1 - \rho \, e^{\mp i \arccos \rho}) \\
&= \braket \psi \alpha ( 1 - \rho^2 \pm i \rho \sqrt{1-\rho^2})
 \enspace ,
\end{split}\end{equation}
so $\abs{ \bra \psi T^\dagger \ket{\alpha, \pm}}^2 = \abs{ \braket \psi \alpha }^2 (1 - \rho^2)$.  
Eq.~\eqnref{e:szegedizationnormalization} follows.   

When $\abs{\rho(\alpha)} = 1$, the claims are immediate from $\ket{\alpha, +} = T \ket \alpha$ and $T^\dagger T = \identity$.  
\end{proof}

We can now prove \propref{t:phaseestimationalgorithmszegedyproof}.  

\begin{proof}[Proof of \propref{t:phaseestimationalgorithmszegedyproof}]
Notice that if we scale $A_G$, $A_{G(x)}$ and $\Lambda$ all by $1 / \norm{\abst(A_G)}$, then both assumptions Eq.~\eqnref{e:phaseestimationalgorithmtrue} and Eq.~\eqnref{e:phaseestimationalgorithmfalse} still hold.  Therefore we will assume below that $\norm{\abst(A_G)} = 1$.  Our goal is to evaluate $f$ using $O(1/\Lambda)$ queries to the phase-flip input oracle of Eq.~\eqnref{e:ox}.  

Assume that $G$ is a connected graph; otherwise, discard all components other than the one containing the vertex $\mu$.  Therefore $\abst(A_G)$ has a single principal eigenvector $\ket \delta$, $\abst(A_G) \ket \delta = \ket \delta$, with $\braket v \delta > 0$ for all $v \in V$.  

Put an arbitrary total order ``$<$" on the vertices in $V$.  For each $v \in V$, let 
\begin{equation} \label{e:szegedizationphiv}
\ket{\varphi_v} = \frac{1}{\sqrt{\braket v \delta}} \bigg( \sqrt{\bra v A_G \ket v \braket v \delta} \ket v + \sum_{w \in V : \, w < v} \sqrt{\abs{\bra v A_G \ket w} \, \braket w \delta} \ket w + \sum_{\substack{w \in V : \, v < w \\ \bra v A_G \ket w \neq 0}} \frac{\bra w A_G \ket v}{\sqrt{\abs{\bra v A_G \ket w}}} \sqrt{\braket w \delta} \ket w \bigg)
\end{equation}
Then 
\begin{equation}\begin{split}
\norm{\ket{\varphi_v}}^2
&= \frac{1}{\braket v \delta} \sum_{w \in V} \bra v \abst(A_G) \ket w \braket w \delta \\
&= 1
 \enspace .
\end{split}\end{equation}

Therefore \thmref{t:szegedization} will apply; define $T$, $S$, $U$ and $M$ from Eqs.~\eqnref{e:szegedizationTS} and~\eqnref{e:szegedizationUM}.  
Also let $\tilde O_x$ be the unitary  
\begin{equation} \label{e:phasefliporacledef}
\tilde O_x \ket{v,w} = \begin{cases}
- \ket{v, w} & \text{if $v \in V_{j, x_j} \subseteq V_{\text{input}}$ for some $j \in [n]$} \\
\ket{v,w} & \text{otherwise}
\end{cases}
\end{equation}
One controlled call to $\tilde O_x$ can be implemented using one call to the standard phase-flip oracle $O_x$ of Eq.~\eqnref{e:ox}.  

\medskip

The algorithm has three steps: 

%\begin{center}\fbox{\begin{minipage}[l]{5in}
\begin{enumerate}
\item
Prepare the initial state $T \ket \mu$.  
\item
Run phase estimation on $W_x = i \, \tilde O_x U$, with precision $\delta_p = \frac 2 \pi \Lambda$ and error rate $\delta_e = \epsilon / 6$.  
\item
Output $1$ if the measured phase is $0$ or $\pi$.  Otherwise output $0$.
\end{enumerate}
%\end{minipage}}\end{center}

\medskip

Phase estimation on a unitary $W$ with precision $\delta_p$ and error rate $\delta_e$ requires $O(1/(\delta_p \delta_e))$ controlled applications of $W$~\cite{cemm:qalg}.  Since $\epsilon = \Omega(1)$, the query complexity of this algorithm is therefore $O(1/\Lambda)$.  It remains to prove completeness and soundness.  

Fix an input $x \in \B^n$.  For $v \in V$, let 
\begin{equation}
\ket{\varphi_v^x} = \begin{cases} \ket v & \text{if $v \in V_{j,x_j}$ for some $j \in [n]$}
\\ 
\ket{\varphi_v} & \text{otherwise} \end{cases}
\end{equation}
Apply \thmref{t:szegedization} using the vectors $\ket{\varphi_v^x}$ to define $T_x$, $U_x$ and $M_x$.  

\begin{lemma} \label{t:szegedizationoracle}
$M = A_G$ and $M_x = A_{G(x)}$.  Moreover, letting $\C^E = \Span(\{ \ket{v, w} : (v, w) \in E \}) \subseteq \C^V \otimes \C^V$ be the span of the edges of $G$, $U_x \vert_{\C^E} = \tilde O_x U \vert_{\C^E}$ and $T_x \ket \mu = T \ket \mu \in \C^E$.  
\end{lemma}

\begin{proof}
First, note that for any vertices $v, w \in V$, from Eq.~\eqnref{e:szegedizationUM} and Eq.~\eqnref{e:szegedizationphiv}, 
\begin{equation}\begin{split}
\bra v M \ket w 
&= \braket{\varphi_v}{w} \braket{v}{\varphi_w} \\
&= \bra v A_G \ket w \sqrt{\frac{\braket v \delta}{\braket w \delta}} \sqrt{\frac{\braket w \delta}{\braket v \delta}} \\
&= \bra v A_G \ket w
 \enspace .
\end{split}\end{equation}
Therefore $M = A_G$.  

Recall that $G(x)$ is the same as $G$ except with the edges to vertices in $\cup_{j \in [n]} V_{j, x_j}$ removed.  
Consider a $v \in V_{j, x_j}$.  By assumption, $v$ has a single neighbor $w \neq v$, so it must be that $\ket{\varphi_v} = \ket w$.  Since $\ket{\varphi_v^x} = \ket v$, $\bra v M_x \ket w = \braket{\varphi_v^x}{w} \braket{v}{\varphi_w^x} = 0$.  However, for all pairs $(v, w)$ that do not make an edge leaving some $V_{j, x_j}$, $\bra v M_x \ket w = \bra v M \ket w$.  Therefore $M_x = A_{G(x)}$.  

Next, we aim to show that $U_x S$ and $\tilde O_x U S$ are the same when restricted to $\C^E$.  Note that 
\begin{equation}\begin{split}
U S
&= 2 T T^\dagger - \identity_{\C^V \otimes \C^V} \\
&= 2 \sum_{v \in V} \ketbra v v \otimes \ketbra{\varphi_v}{\varphi_v} - \identity_{\C^V \otimes \C^V} \\
&= \sum_{v \in V} \ketbra v v \otimes (2 \ketbra{\varphi_v}{\varphi_v} - \identity_{\C^V})
 \enspace .
\end{split}\end{equation}
Similarly $U_x S = \sum_v \ketbra v v \otimes (2 \ketbra{\varphi_v^x}{\varphi_v^x} - \identity_{\C^V})$.  Therefore, 
\begin{align}
(U S)^\dagger U_x S 
&= \sum_v \ketbra v v \otimes \big[ (2 \ketbra{\varphi_v}{\varphi_v} - \identity)(2 \ketbra{\varphi_v^x}{\varphi_v^x} - \identity) \big] \nonumber \\
&= \sum_{v \notin \cup_j V_{j, x_j}} \ketbra v v \otimes \identity + \sum_{\substack{j \in [n], v \in V_{j, x_j} \\ w \sim v}} \ketbra v v \otimes (\identity - 2 \ketbra v v - 2 \ketbra w w) \label{e:USUxS}
 \enspace ,
\end{align}
where in the second term $w$ is $v$'s single neighbor in $G$.  On the other hand, from its definition in Eq.~\eqnref{e:phasefliporacledef}, 
\begin{equation}
\tilde O_x = \identity_{\C^V \otimes \C^V} - 2 \sum_{\substack{j \in [n] , v \in V_{j, x_j}}} \ketbra v v \otimes \identity_{\C^V}
 \enspace .
\end{equation}
By inspection, this is the same as Eq.~\eqnref{e:USUxS} on $\C^E$.  

Finally, since by assumption $\mu \notin V_{\text{input}}$, $T_x \ket \mu = \ket \mu \otimes \ket{\varphi_\mu^x} = \ket \mu \otimes \ket{\varphi_\mu} = T \ket \mu$.  $T \ket \mu \in \C^E$ by Eq.~\eqnref{e:szegedizationphiv}.  
\end{proof}

The initial state $T \ket \mu = T_x \ket \mu $ lies in $\Range(T_x) \subseteq \C^E$.  Also, $U_x$ fixes $\C^E$; in fact, it even fixes the join of the ranges of $T_x$ and $S T_x$, which could be smaller than $\C^E$.  By \lemref{t:szegedizationoracle}, $\tilde O_x U$ and $U_x$ are the same restricted to $\C^E$.  Therefore, the algorithm behaves the same as if it were running phase estimation on $i U_x$ instead of $W_x = i \tilde O_x U$.  

Based on Eq.~\eqnref{e:phaseestimationalgorithmtrue}, the algorithm is complete:

\begin{lemma}
If $x \in \D$ and $f(x) = 1$, then the algorithm outputs $1$ with probability at least $\epsilon - \delta_e = \frac 5 6 \epsilon$, where $\delta_e = \epsilon / 6$ is the phase estimation error parameter.  
\end{lemma}

\begin{proof}
Assume that $f(x) = 1$.  From Eq.~\eqnref{e:phaseestimationalgorithmtrue}, $A_{G(x)}$ has an eigenvalue-zero eigenvector $\ket \alpha \in \C^V$ with $\norm{\ket \alpha} = 1$ and $\abs{\braket \mu \alpha}^2 \geq \epsilon$.  By \thmref{t:szegedization} with $\rho(\alpha) = 0$, $U_x$ has eigenvectors $\ket{\alpha, \pm} = (1 \pm i S) T_x \ket \alpha$ with respective eigenvalues $\pm i$.  By \lemref{t:szegedizationnormalization}, these satisfy
\begin{equation}
\frac{\abs{\bra \mu T_x^\dagger \ket{\alpha, +}}^2}{\norm{\ket{\alpha, +}}} + \frac{\abs{\bra \mu T_x^\dagger \ket{\alpha, -}}^2}{\norm{\ket{\alpha, -}}}
= \abs{\braket \mu \alpha}^2
\geq \epsilon
 \enspace .
\end{equation}
Thus the algorithm measures a phase of $0$ or $\pi$, and outputs $1$, with probability at least $\epsilon - \delta_e$.  
\end{proof}

Based on Eq.~\eqnref{e:phaseestimationalgorithmfalse}, since the phase estimation precision is $\delta_p = \frac 2 \pi \Lambda$, the algorithm is sound:

\begin{lemma}
If $x \in \D$ and $f(x) = 0$, then the algorithm outputs $1$ with probability at most $\epsilon/2 + \delta_e = \frac 2 3 \epsilon$.  
\end{lemma}

\begin{proof}
Let $\{ \ket \alpha \}$ be a complete set of orthonormal eigenvectors of $A_{G(x)}$, with corresponding eigenvalues $\rho(\alpha)$.  The initial state $T \ket \mu = T_x \ket \mu$ lies in the range of $T_x$, and therefore is in the span of the eigenvectors $\{ \ket{\alpha, \pm} \}$, i.e., the space $R = \oplus_\alpha R_\alpha$ from \thmref{t:szegedization}.  The probability that the algorithm outputs $1$ is therefore at most $\delta_e$ plus 
\begin{equation}
\sum_{\substack{\ket{\alpha,b} : \\ \arg(\lambda(\alpha, b)) \in [\frac\pi 2-\delta_p, \frac\pi 2+\delta_p] \cup [-\frac\pi 2-\delta_p,-\frac\pi 2+\delta_p]}} \frac{\abs{\braket{\alpha, b}{\mu}}^2}{\norm{\ket{\alpha,b}}^2}
= 
\sum_{\alpha : \abs{\arcsin \rho(\alpha)} \leq \delta_p} \bigg( \frac{\abs{\braket{\alpha, +}{\mu}}^2}{\norm{\ket{\alpha,+}}^2} + \frac{\abs{\braket{\alpha, -}{\mu}}^2}{\norm{\ket{\alpha,-}}^2} \bigg)
\end{equation}
where in the first sum $b$ can be either $+$ or $-$, and we have used $\lambda(\alpha, \pm) = e^{\pm i \arccos \rho(\alpha)}$, so $\arg(\lambda(\alpha, \pm)) = \pm (\frac \pi 2 - \arcsin \rho(\alpha))$.  

Since $\abs{\arcsin \rho(\alpha)} \leq \frac\pi 2 \abs{\rho(\alpha)}$, and by \lemref{t:szegedizationnormalization}, the above sum is at most 
\begin{equation}
\sum_{\alpha : \abs{\rho(\alpha)} \leq \Lambda} \abs{\braket \alpha \mu}^2
 \enspace ,
\end{equation}
which is at most $\epsilon/2$ by Eq.~\eqnref{e:phaseestimationalgorithmfalse}.  
\end{proof}

Therefore, the algorithm is correct.  The constant gap $\epsilon/6$ between its completeness and soundness parameters can be amplified as usual.  
\end{proof}

\subsection{Discrete-time simulation of a continuous-time algorithm} \label{s:phaseestimationalgorithmcontinuousproof}

\begin{proposition} \label{t:phaseestimationalgorithmcontinuousproof}
Under the assumptions of \thmref{t:phaseestimationalgorithm}, $f$ can be evaluated with error probability at most $1/3$ using $O\big( \frac{1}{\Lambda} \log (\frac{1}{\Lambda}) / \log \log \frac{1}{\Lambda} \big)$ queries to the input oracle $O_x$.  
\end{proposition}

\propref{t:phaseestimationalgorithmszegedyproof} and \propref{t:phaseestimationalgorithmcontinuousproof} together prove \thmref{t:phaseestimationalgorithm}.  

To prove \propref{t:phaseestimationalgorithmcontinuousproof}, we will first give an algorithm in the continuous-time query model.  This algorithm uses the same idea as the algorithm from \propref{t:phaseestimationalgorithmszegedyproof}.  Namely, we run phase estimation on a certain unitary.  Completeness of the algorithm is derived from Eq.~\eqnref{e:phaseestimationalgorithmtrue} and soundness derived from Eq.~\eqnref{e:phaseestimationalgorithmfalse}.  

Then we simulate this continuous-query algorithm in the discrete-query model.  The key technical step is a recent result due to Cleve, Gottesman, Mosca, Somma and Yonge-Mallo, \cite{CleveGottesmanMoscaSommaYongeMallo08discretize}, that states that continuous-query algorithms can be simulated by discrete-query algorithms with only a logarithmic overhead.  We quote here a weak version of their theorem.  

\begin{theorem}[{\cite{CleveGottesmanMoscaSommaYongeMallo08discretize}}] \label{t:CGMSYization}
Suppose we are given a continuous-time query algorithm with any driving Hamiltonian $D(t)$ whose operator norm $\norm{D(t)}$ is bounded above by any $L_1$ function with respect to $t$.  (The size of $\norm{D(t)}$ as a function of the input size $N$ does not matter.)  Then there exists a discrete-time query algorithm that makes 
\begin{equation}
O \bigg( \frac{T \log \frac T \delta}{\delta \log \log \frac T \delta} \bigg)
\end{equation}
full queries and whose answer has fidelity $1-\delta$ with the output of the continuous-time algorithm.  
\end{theorem}

We will not define the continuous-time query model here; see~\cite{CleveGottesmanMoscaSommaYongeMallo08discretize} for details.  
For other applications of the model, see, e.g.,~\cite{FarhiGutmann98continuousquery, Mochon06continuousquery, fgg:and-or, ChildsCleveDeottoFarhiGutmannSpielman02walk}.  

\begin{proof}[Proof of \propref{t:phaseestimationalgorithmcontinuousproof}]
We start by presenting and analyzing the continuous-time query algorithm.  

The rough idea is to run phase estimation with precision $\Lambda$ on the unitary $e^{i A_{G(x)}}$.  Output $1$ if the estimated phase is zero, and otherwise output $0$.  This algorithm belongs in the continuous-query model, because $A_{G(x)}$ is the sum of an input-independent term $A_G$ and an oracle-dependent term
\begin{equation}
A_{G(x)} - A_G = - \sum_{\substack{j \in [n], v \in V_{j, x_j} \\ w \sim v}} ( \ketbra v w + \ketbra w v )
 \enspace .
\end{equation}
However, this algorithm would not be sound.  When $f(x) = 0$, the problem is that even though $A_{G(x)}$ has an effective spectral gap, that does not imply that there is an effective gap in the phases of the eigenvalues of $e^{i A_{G(x)}}$.  Each eigenvalue $\rho \in \R$ of $A_{G(x)}$ corresponds to the eigenvalue $e^{i \rho}$ of $e^{i A_{G(x)}}$, and therefore large eigenvalues can wrap all the way around the circle.  For example, an eigenvalue-$(2\pi)$ eigenvector of $A_{G(x)}$ is an eigenvalue-one eigenvector of $e^{i A_{G(x)}}$, which phase estimation will not distinguish from an eigenvalue-zero eigenvector of $A_{G(x)}$.  

We solve this issue by scaling $A_{G(x)}$ by a uniformly random $T \in_R (0, \tau)$, where $\tau$ is a large enough constant.  Intuitively, this means that for any eigenvector $\ket \alpha$ of $A_{G(x)}$ with eigenvector $\rho(\alpha)$, $\abs{\rho(\alpha)} > \Lambda$, there is only a small chance that $T \rho(\alpha)$ wraps around into the interval $[-\Lambda, \Lambda]$.  

We will analyze the following concrete algorithm: 

%\begin{center}\fbox{\begin{minipage}[l]{5in}	%% NOTE: If this box is removed, then add back the equation numbers
\begin{enumerate}
\item
Let $M = \lceil 12/\epsilon \rceil = O(1)$.  Let $\tau = M^2 / \Lambda$.  Let $T$ be a random variable chosen uniformly from the interval $(0, \tau)$.  
\item
Prepare the initial state 
\begin{equation}
\frac{1}{\sqrt M}\Big(\sum_{m=1}^{M} \ket m\Big) \otimes \ket \mu \in \C^{[M]} \otimes \C^V
 \enspace .
\end{equation}
\item
Apply $e^{i T \frac m M A_{G(x)}}$ to the second register, controlled by the value $m$ in the first register.  That is, apply the unitary
\begin{equation}
\sum_{m \in [M]} \ketbra m m \otimes e^{i T \frac m M A_{G(x)}} = \exp\Big(i T \sum_{m \in [M]} \frac m M \ketbra m m \otimes A_{G(x)} \Big)
 \enspace .
\end{equation}
The resulting state is 
\begin{equation}
\frac1{\sqrt M} \sum_{m \in [M]} \ket m \otimes e^{i T \frac m M A_{G(x)}} \ket \mu
 \enspace . 
\end{equation}
\item
Project the first register onto the uniform superposition $\frac1{\sqrt M} \sum_{m \in M} \ket m$.  Output $1$ if the projection succeeds, and output $0$ otherwise.  
\end{enumerate}
%\end{minipage}}\end{center}

This algorithm is essentially running a slightly simplified version of phase estimation.  We have chosen to write it out concretely, instead of using phase estimation as a black box, partly in order to illustrate that full phase estimation is unnecessary when the objective is just to decide whether or not the phase is zero.  When there is a large gap between the parameter on the right-hand side of Eq.~\eqnref{e:phaseestimationalgorithmtrue} and that on the right-hand side of Eq.~\eqnref{e:phaseestimationalgorithmfalse}, the procedure becomes especially simple.  (In fact, for our application of \thmref{t:phaseestimationalgorithm} to span programs, the gap can be made a constant arbitrarily close to one.)  A similar simplification can be made in the proof of \propref{t:phaseestimationalgorithmszegedyproof}.  

\begin{lemma}
When run with input $x \in \D$, the above procedure satisfies: 
\begin{itemize}
\item If $f(x) = 1$, then it outputs $1$ with probability at least $\epsilon$.  
\item If $f(x) = 0$, then it outputs $1$ with probability at most $3 \epsilon / 4$.  
\end{itemize}
\end{lemma}

\begin{proof}
Let $\{ \ket \alpha \}$ be a complete set of orthonormal eigenvectors of $A_{G(x)}$, with corresponding eigenvalues $\rho(\alpha)$.  
The probability that the procedure outputs $1$ is the expectation versus $T$ of 
\begin{equation}\begin{split}
\Pr\!\big[\text{output $1$} \vert T = t \big] &= 
\frac{1}{M^2} \Big\|{ \sum_{m \in [M]} e^{i t \frac m M A_{G(x)}} \ket \mu }\Big\|^2 \\
&= 
\frac{1}{M^2} \Big\| \sum_\alpha \sum_{m \in [M]} e^{i t \frac m M \rho(\alpha)} \braket \alpha \mu \ket \alpha  \Big\|^2 \\
&= 
\frac{1}{M^2} \sum_\alpha \Big\lvert{ \sum_{m \in [M]} e^{i t \frac m M \rho(\alpha)} }\Big\rvert^2 \abs{\braket \alpha \mu}^2
\end{split}\end{equation}

When $f(x) = 1$, we have from Eq.~\eqnref{e:phaseestimationalgorithmtrue} that $\sum_{\alpha : \rho(\alpha) = 0} \abs{\braket \alpha \mu}^2 \geq \epsilon$, so $\Pr\!\big[\text{output $1$} \vert T = t \big] \geq \epsilon$, regardless of $t$.  

For the case $f(x) = 0$, we split the sum over $\alpha$ into a sum over those $\alpha$ with $\abs{\rho(\alpha)} \leq \Lambda$ and a sum over those $\alpha$ with $\abs{\rho(\alpha)} > \Lambda$.  By Eq.~\eqnref{e:phaseestimationalgorithmfalse}, the first sum is at most $M^2 \epsilon / 2$: 
\begin{equation} \label{e:soundnesssum}
\Pr\!\big[\text{output $1$} \vert T = t \big] 
\leq
\frac \epsilon 2 + \frac 1 {M^2} \sum_{\alpha : \, \abs{\rho(\alpha)} > \Lambda} \Big\lvert{ \sum_{m \in [M]} e^{i t \frac m M \rho(\alpha)} }\Big\rvert^2 \abs{\braket \alpha \mu}^2
 \enspace .
\end{equation}
Now use 
\begin{equation}
\Big\lvert{ \sum_{m \in [M]} e^{i t \frac m M \rho(\alpha)} }\Big\rvert^2
=
M + \sum_{\substack{l, m \in [M] \\ l \neq m}} e^{i t \frac{l-m}{M} \rho(\alpha)}
\end{equation}
and, for $\rho(\alpha) \neq 0$, 
\begin{equation}
\Ex_T \big[e^{i T \frac{l-m}{M} \rho(\alpha)} \big]
=
\frac{e^{i \tau \frac{l-m}{M} \rho(\alpha)} - 1}{i \tau \frac{l-m}{M} \rho(\alpha)}
 \enspace .
\end{equation}
Substituting back into Eq.~\eqnref{e:soundnesssum} gives 
\begin{equation}\begin{split}
\Pr\!\big[\text{output $1$}\big]
&\leq
\frac \epsilon 2 + \frac 1 {M^2} \sum_{\alpha : \, \abs{\rho(\alpha)} > \Lambda} \bigg( M + \sum_{\substack{l, m \in [M] \\ l \neq m}} \frac{e^{i \tau \frac{l-m}{M} \rho(\alpha)} - 1}{i \tau \frac{l-m}{M} \rho(\alpha)} \bigg) \abs{\braket \alpha \mu}^2 \\
&=
\frac \epsilon 2 + \frac 1 M \sum_{\alpha : \, \abs{\rho(\alpha)} > \Lambda} \bigg( 1 + \frac{1}{\tau \rho(\alpha)} \sum_{\substack{l, m \in [M] \\ l > m}} \frac{e^{i \tau \frac{l-m}{M} \rho(\alpha)} - e^{-i \tau \frac{l-m}{M} \rho(\alpha)}}{i (l-m)} \bigg) \abs{\braket \alpha \mu}^2 \\
&\leq
\frac \epsilon 2 + \frac 1 M \sum_{\alpha : \, \abs{\rho(\alpha)} > \Lambda} \bigg( 1 + \frac{2 M^2}{\tau \rho(\alpha)} \bigg) \abs{\braket \alpha \mu}^2 
 \enspace ,
\end{split}\end{equation}
where in the last step we have (loosely) bounded the sum over $l$ and $m$.  Now use $\sum_\alpha \abs{\braket \alpha \mu}^2 = 1$, $\tau = M^2 / \Lambda$ and $M \geq 12/\epsilon$ to conclude 
\begin{equation}\begin{split}
\Pr\!\big[\text{output $1$}\big]
&\leq 
\frac \epsilon 2 + \frac 3 M \\
&\leq \frac{3 \epsilon}{4}
 \enspace ,
\end{split}\end{equation}
as claimed.  
\end{proof}

Therefore, the above procedure is correct.  It remains to show that it can be simulated using $O( \tau \log \tau / \log \log \tau)$ queries to the input oracle $O_x$ from Eq.~\eqnref{e:ox}.  The difficulty is simulating $e^{i t H(x)}$ for a $t \in (0, \tau)$, where 
\begin{equation}
H(x) = \sum_{m \in [M]} \frac m M \ketbra m m \otimes A_{G(x)}
 \enspace .
\end{equation}
In the language of physics, $e^{i t H(x)}$ corresponds to applying the time-independent Hamiltonian $H(x)$ for a time $t$.  

Let 
\begin{equation}
D =  \sum_{m \in [M]} \frac m M \ketbra m m \otimes A_G
 \enspace ,
\end{equation}
the ``driving Hamiltonian."  $D$ is independent of the input $x$, so $e^{i t D}$ can be implemented without querying $O_x$.  
Let the ``query Hamiltonian" be 
\begin{equation} \label{e:myqueryHamiltonian}
\begin{split}
H_x
&= \sum_{m \in [M]} \frac m M \ketbra m m \otimes (A_{G(x)} - A_G) \\
&= - \sum_{m \in [M]} \frac m M \ketbra m m \otimes \sum_{\substack{v \in \cup_{j \in [n]} V_{j, x_j} \\ w \sim v}} ( \ketbra v w + \ketbra w v )
 \enspace ,
\end{split}\end{equation}
where we have used that $G(x)$ differs from $G$ only in the deletion of the weight-one edges leaving the vertices $v \in \cup_{j \in [n]} V_{j, x_j}$.  
For an arbitrary $t \in \R$, $e^{i t H_x}$ can be implemented using at most two queries to the oracle~$O_x$.  

Then 
\begin{equation}
H(x) = D + H_x
 \enspace .
\end{equation}
The problem for taking the exponential is that the two terms $D$ and $H_x$ do not commute.

We will now apply \thmref{t:CGMSYization}, which very roughly can be thought of as an asymmetric Lie-Trotter expansion of the exponential.  A minor difference between our setting and the one in~\cite{CleveGottesmanMoscaSommaYongeMallo08discretize}, though, is that they assume a more restricted form for the query Hamiltonian.  For querying a $k$-bit string $y$ with the first bit fixed to $y_1 = 0$, they assume the query Hamiltonian is  
\begin{equation}
\tilde H_y = \sum_{j \in [k]} y_j \ketbra j j = \sum_{j \in [k] : \, y_j = 1} \ketbra j j
 \enspace .
\end{equation}
Unfortunately, our query Hamiltonian $H_x$ is not of this required form.  In order to apply \thmref{t:CGMSYization} as a black box, we need to put $H_x$ into this form for some string $y$.  Each bit of $y$ will be a fixed function of exactly one bit of $x$, and therefore a discrete phase-flip query on $y$ can be simulated with one application of~$O_x$.  

First of all, note that \thmref{t:CGMSYization} still holds if the query Hamiltonian is of the form 
\begin{equation} \label{e:theirqueryHamiltonian}
\tilde H_y' = \sum_{j \in [k]} y_j g(j) \ketbra j j
 \enspace ,
\end{equation}
where $g$ is any fixed function $[k] \rightarrow \{-1, 1\}$.  That is, signs are allowed.  Indeed, then 
\begin{equation}
\tilde H_y' = - \sum_{j \in [k] : \, g(j) = -1} \ketbra j j + \sum_{j \in [k] : \, g(j) = 1} y_j \ketbra j j + \sum_{j \in [k] : \, g(j) = -1} (1-y_j) \ketbra j j
\end{equation}
The first term can be moved into the driving Hamiltonian, since it does not depend on $y$, and the remaining terms are of the form of $\tilde H_{y'}$ on an input $y'$ that equals $y$ except with the bits $\{ j \in [k] : g(j) = -1 \}$ complemented.  

Let us now translate our query Hamiltonian $H_x$ into the form of Eq.~\eqnref{e:theirqueryHamiltonian}.  
For $m \in [M]$, let 
\begin{align}
D^m &= \ketbra m m \otimes A_G \\
H_x^m &= -\ketbra m m \otimes \sum_{\substack{v \in \cup_{j \in [n]} V_{j, x_j} \\ w \sim v}} ( \ketbra v w + \ketbra w v ) \label{e:dividedphaseestimationqueryHamiltonian} \\
H^m(x) &= D^m + H_x^m = \ketbra m m \otimes A_{G(x)}
 \enspace .
\end{align}
Then $H(x) = \sum_{m \in [M]} \frac m M H^m(x)$, so 
\begin{equation} \label{e:dividedphaseestimation}
e^{i t H(x)} = \prod_{m \in [M]} \exp\!\big(i t \frac m M H^m(x)\big)
\end{equation}
since the different terms $H^m(x)$ commute pairwise.  

The term $H_x^m$ is nearly of the form Eq.~\eqnref{e:theirqueryHamiltonian}.  It can be put in that form by changing basis.  For $j \in [n]$, $b \in \B$ and $v \in V_{j, b} \subseteq V_{\text{input}}$, with neighbor $w$, write
\begin{equation}
\ketbra v w + \ketbra w v = \ketbra{vw{+}}{vw{+}} - \ketbra{vw{-}}{vw{-}}
 \enspace ,
\end{equation}
where $\ket{vw{\pm}} = \frac1{\sqrt 2}(\ket v \pm \ket w)$.  Use two bits of $y$, with values $x_j$ and $\bar x_j$, to get the terms $\mp \ketbra m m \otimes \ketbra{vw\pm}{vw{\pm}}$ from Eq.~\eqnref{e:dividedphaseestimationqueryHamiltonian} into the form of Eq.~\eqnref{e:theirqueryHamiltonian}.  

Overall, therefore $y$ has $\abs{V_{j, b}}$ copies of bit $x_j$ and $\abs{V_{j, b}}$ copies of the complement $\bar x_j$, for all $j \in [n]$, $b \in \B$.  Thus $k$, the length of $y$, is $2 \abs{V_{\text{input}}}$.  

Finally, apply \thmref{t:CGMSYization}, with accuracy parameter $\delta = \frac{\epsilon}{12 M} = \Omega(1)$, $M$ times, once for each of the terms in Eq.~\eqnref{e:dividedphaseestimation}.  The total query complexity is $O(M \tau \log (\tau) / \log \log \tau) = O(\frac 1 \Lambda \log (\frac 1 \Lambda) / \log \log \frac 1 \Lambda)$, as desired.  The total error introduced in the simulation is at most $M \delta$, so the gap between the completeness and soundness parameters of the final algorithm is at least $\epsilon/4 - 2 \cdot \epsilon / 12 = \epsilon / 12$.  This constant gap can be amplified as usual.  
\end{proof}

\section{The general quantum adversary bound is nearly tight for every boolean function}

We can now prove the main result of this paper, that for any total or partial boolean function $f$ the general adversary bound on the quantum query complexity is tight up to a logarithmic factor.  

\begin{theorem} \label{t:querycomplexitytight}
For any function $f : \D \rightarrow \{0,1\}$, with $\D \subseteq \{0,1\}^n$, the bounded-error quantum query complexity of $f$, $Q(f)$, satisfies 
\begin{align}
Q(f) &= \Omega(\ADVpm(f)) \\
\intertext{and}
Q(f) &= O\bigg(\ADVpm(f) \, \frac{\log \ADVpm(f)}{\log \log \ADVpm(f)}\bigg)
 \enspace .
\end{align}
\end{theorem}

\begin{proof}
The lower bound is a special case of \thmref{t:advquerycomplexity}, and is due to H{\o}yer, Lee and {\v S}palek~\cite{HoyerLeeSpalek07negativeadv}.  

As already sketched in \secref{s:introduction}, 
for the upper bound, use the semi-definite program from \thmref{t:spanprogramSDP} with uniform costs $s = \vec 1$ to construct a span program $P$ computing $f_P\vert_\D = f$, with $\wsizeD P = \ADVpm(f)$.  Then apply \thmref{t:generalspanprogramalgorithm} to obtain a bounded-error quantum query algorithm that evaluates $f$.  
\end{proof}

By using binary search and standard error reduction, \thmref{t:querycomplexitytight} can be extended to cover functions with larger codomain~\cite{Lee09private}: 

\begin{theorem} \label{t:querycomplexitytightnonbinary}
For any function $f : \D \rightarrow [m]$, with $\D \subseteq \{0,1\}^n$, $Q(f)$ satisfies
\begin{align}
Q(f) &= \Omega(\ADVpm(f)) \\
\intertext{and}
Q(f) &= O\bigg(\ADVpm(f) \, \frac{\log \ADVpm(f)}{\log \log \ADVpm(f)} \log(m) \log \log m \bigg)
 \enspace .
\end{align}
\end{theorem}

\begin{proof}
The lower bound is again due to~\cite{HoyerLeeSpalek07negativeadv}.  To derive the upper bound, first let us show: 

\begin{lemma} \label{t:adversaryprojection}
For finite sets $\D \subseteq C^n$, $E$ and $F$, let $f: \D \rightarrow E$ and $g : E \rightarrow F$.  Let $s \in [0, \infty)^n$.  Then 
\begin{align}
\ADV_s(g \circ f) &\leq \ADV_s(f) \\
\ADVpm_s(g \circ f) &\leq \ADVpm_s(f)
 \enspace .
\end{align}
\end{lemma}

\begin{proof}
For $x, y \in \D$, $f(x) = f(y)$ implies $g(f(x)) = g(f(y))$.  Therefore if $\Gamma$ is an adversary matrix for $g \circ f : \D \rightarrow F$, then $\Gamma$ is also an adversary matrix for $f$.  The conclusions follow by \defref{t:adversarydef} for the adversary bounds.  
\end{proof}

In order to evaluate $f$, apply standard binary search using $\lceil \log_2 m \rceil$ steps.  In each step, there is some division of the range $g : [m] \rightarrow \{0,1\}$.  By \lemref{t:adversaryprojection}, $\ADVpm(g \circ f) \leq \ADVpm(f)$.  Therefore by \thmref{t:querycomplexitytight}, $g \circ f$ can be evaluated with error at most $1/3$, using \begin{equation}
O\bigg(\ADVpm(f) \, \frac{\log \ADVpm(f)}{\log \log \ADVpm(f)}\bigg)
\end{equation}
queries.  Repeat this $O(\log \log m)$ times in order to reduce the error probability to $1/(3 \lceil \log m \rceil)$.  Then by the union bound, the entire procedure has a probability of error at most $1/3$.  
\end{proof}

We do not have a result for the case of a non-binary input alphabet.  Of course the input can be encoded into binary, so that \thmref{t:querycomplexitytight} applies.  However, this encoding might increase $\ADVpm$ significantly.

\section{Open problems} \label{s:conclusion}

We have shown that for any boolean function $f$, the general adversary bound $\ADVpm(f)$ is a tight lower bound on the bounded-error quantum query complexity $Q(f)$, up to a logarithmic factor.  In proving this statement, we have also shown that quantum algorithms, judged by query complexity, and span programs, judged by witness size, are equivalent computational models for evaluating boolean functions, again up to a logarithmic factor.  

Among the corollaries, \thmref{t:formulaevaluation} gives an optimal quantum algorithm for evaluating adversary-balanced formulas over any finite boolean gate set.  For example, the formula's gate set may be taken to be all functions $\{0,1\}^n \rightarrow \{0,1\}$ with $n \leq 1000$.  This formula-evaluation algorithm exploits the ease of composing span programs.  The main unresolved problem here is how best to evaluate unbalanced formulas, aiming for optimal query complexity and near-optimal time complexity.  

Span programs may also be useful for developing other quantum algorithms.  They have a rich mathematical structure, and their potential has not been fully explored.  One possible approach is to study the general adversary bound for more problems.  For example, studying the Barnum/Saks/Szegedy semi-definite program for quantum query complexity~\cite{BarnumSaksSzegedy03adv} has led to improved zero-error algorithms for Ordered Search~\cite{ChildsLandahlParrilo06orderedsearch}.  The $\ADVpm$ SDP is simpler than the SDP in~\cite{BarnumSaksSzegedy03adv}, and \thmref{t:adversarydual} gives a new, simpler form for the dual SDP, for boolean functions.  Although this SDP is still exponentially large, the simplifications may ease the inference of structure from numerical investigations.  For the Ordered Search problem in particular, Childs and Lee have closely characterized $\ADVpm$~\cite{ChildsLee07orderedsearch}.  This result will not necessarily be useful for developing an Ordered Search algorithm because the codomain is not boolean and \thmref{t:querycomplexitytight} has a logarithmic overhead.  
A variation of this problem, Least-Significant-Bit Ordered Search, has boolean codomain, but is of less practical interest.  

The nonnegative-weight adversary bound $\ADV$ is often easy to approximate.  If this bound is close to $\ADVpm$, then perhaps a solution to Eq.~\eqnref{e:adversarydual}, the SDP dual to the $\ADV$ SDP, can also be turned into a quantum walk algorithm.  However, the span program framework will not apply for the analysis.  

This article has focused on query complexity, but \thmref{t:phaseestimationalgorithm} is more than an information-theoretic statement.  It gives explicit algorithms whose time complexity can be analyzed, as in \thmref{t:formulaevaluation} for formula evaluation.  \propref{t:reducedtensorproductcompose}, \thmref{t:spanprogramspectralanalysisnonblackbox} and \thmref{t:generalspanprogramalgorithmnonblackbox} are pertinent results, but more techniques are needed for developing span programs $P$ such that $\norm{\abst(A_{G_P})} = O(1)$ and for which the quantum walk reflections from Szegedy's \thmref{t:szegedization} can be implemented efficiently.  

It is an interesting problem to consider functions with non-binary input alphabet and non-boolean codomain.  The three main theorems, \thmref{t:spanprogramSDP}, \thmref{t:spanprogramspectralanalysis} and \thmref{t:phaseestimationalgorithm}, may extend to cover partial functions with domain in $[k]^n$ and $k = O(1)$.  When the codomain is not boolean, we would like to strengthen \thmref{t:querycomplexitytightnonbinary}.  The natural approach is to define generalized canonical span programs and extend \lemref{t:spanprogramSDPconstruction} to characterize the optimal generalized witness size of $f : C^n \rightarrow E$~\cite{ReichardtSpalek09generalizedspanprograms}.  Although this may lead to new quantum query algorithms, it will be insufficient for obtaining provably optimal or near-optimal algorithms for non-binary input alphabets, since the SDP in Eq.~\eqnref{e:generaladversarydual} is {not} always equal to $\ADVpm$; see Eq.~\eqnref{e:generaladversarydualgeneral}.  Moreover, there are functions $[3]^2 \rightarrow [3]$ for which both $\ADVpm$ and the SDP in Eq.~\eqnref{e:generaladversarydual} compose strictly sub-multiplicatively, which indicates that the formula-evaluation problem for non-boolean gate sets is more complicated.  

One might ask whether the classical query complexity of evaluating a span program $P$ on inputs in $\D$ can be related to the witness size $\wsizeD P$.  A polynomial dependence is not possible, though, since there is only a polynomial relationship between quantum and classical query complexities for total functions~\cite{Simon97promise, BealsBuhrmanCleveMoscaWolf98}.  

Finally, we conjecture that the logarithmic overhead can be removed from \thmref{t:querycomplexitytight}.  
An analogous conjecture may hold in the continuous-time query model~\cite{FarhiGutmann98continuousquery, Mochon06continuousquery, CleveGottesmanMoscaSommaYongeMallo08discretize}.  

\begin{conjecture} \label{t:advboundtight}
For any function $f : \D \rightarrow \{0,1\}$, with $\D \subseteq \{0,1\}^n$, the general quantum adversary bound is tight: 
\begin{equation}
Q(f) = \Theta(\ADVpm(f))
 \enspace .
\end{equation}
\end{conjecture}

\section*{Acknowledgements}
I would like to thank Robert {\v S}palek and Troy Lee for helpful discussions and feedback.  I would also like to thank Andris Ambainis and Stephanie Wehner for sharing their ideas on the Hamming-weight threshold functions.

\bibliographystyle{alpha}
\bibliography{andor}

\appendix

\section[Optimal span programs for the Hamming-weight threshold functions]{Optimal span programs for the Hamming-weight threshold \\ functions} \label{s:thresholdappendix}

In this appendix, span programs with optimal witness size are given for the Hamming-weight threshold functions.  Additionally, optimal span programs are given for those Hamming-weight interval functions for which the nonnegative-weight adversary bound equals the general adversary bound.  The motivation is to show an explicit and nontrivial span program construction.  The main technique, recursive composition of symmetrized span programs, may be useful for other constructions.  

Surprisingly, the optimal span programs are simply derived from span programs for AND and OR gates, composed in a certain symmetrical manner and with optimized weights.  The optimal span programs for Threshold 2 of 3 and Threshold 2 of 4 given in~\cite{ReichardtSpalek08spanprogram} did not have this form.  

The proofs are simple calculations.  After presenting the span programs, we compute their witness sizes, compute the best nonnegative-weight adversary bounds (by giving adversary matrices and solutions to the dual formulation), and show by perturbations of these matrices that the general adversary bound is strictly greater in those cases where the span program witness size does not match the nonnegative-weight adversary bound.  

\begin{definition}
The Hamming-weight threshold function $\threshold{l}{n} : \{0, 1\}^n \rightarrow \{0,1\}$ is defined by 
\begin{equation}
\threshold{l}{n}(x) = \begin{cases}
1 & \text{if $\abs{x} \geq l$} \\
0 & \text{otherwise}
\end{cases}
\end{equation}
where $\abs{x} = \sum_{i=1}^{n} x_i$ is the Hamming weight of $x$.  

The Hamming-weight interval function $\interval{l}{m}{n} : \{0,1\}^n \rightarrow \{0,1\}$ is defined by 
\begin{equation}
\interval{l}{m}{n}(x) = \begin{cases}
1 & \text{if $l \leq \abs{x} \leq m$} \\
0 & \text{otherwise}
\end{cases}
\end{equation}
\end{definition}

Note that $\threshold{l}{n} = \interval{l}{n}{n}$ and, for all $x \in \{0,1\}^n$, $\interval{l}{m}{n}(x)$ is the conjunction $\threshold{l}{n}(x) \wedge \threshold{n-m}{n}(\bar x)$, where $\bar x$ is the bitwise complement of $x$.  Also, $\interval{l}{m}{n}(x) = \interval{n-m}{n-l}{n}(\bar x)$, which allows us to assume without loss of generality that $\abs{\frac{n}{2} - m} \leq \abs{\frac{n}{2} - l}$.  

\begin{theorem} \label{t:maintheorem}
For the interval function $\interval{l}{m}{n}$, assume that $\abs{\frac{n}{2} - m} \leq \abs{\frac{n}{2} - l}$.  Then  
\begin{equation}
\ADV(\interval{l}{m}{n}) = \begin{cases}
\sqrt{ (m+1)(n-m) + \frac{m(n-l+1)}{(m-l+1)^2} } & \text{if $l > 0$} \\
\sqrt{ (m+1)(n-m) } & \text{if $l = 0$}
\end{cases}
 \enspace .
\end{equation}
If $l \in \{0, 1, m \}$, then $\ADVpm(\interval{l}{m}{n}) = \ADV(\interval{l}{m}{n})$; otherwise $\ADVpm(\interval{l}{m}{n}) > \ADV(\interval{l}{m}{n})$.  

There exists a span program $P_{l,m}^{n}$ computing $f_{P_{l,m}^{n}} = \interval{l}{m}{n}$, with witness size 
\begin{equation} \label{e:maintheoremwsize}
\wsize{P_{l,m}^n} \leq \sqrt{ (m+1)(n-m) + \frac{l(n-l+1)}{m-l+1} }
 \enspace .
\end{equation}
This witness size matches $\ADV(\interval{l}{m}{n})$, and hence is optimal, for $l \in \{0, 1, m \}$, i.e., in those cases where $\ADVpm(\interval{l}{m}{n}) = \ADV(\interval{l}{m}{n})$.  
\end{theorem}

Our span program construction for the case $l = 0$ and $m = n-2$ has been influenced by a family of constructions due to Ambainis that come arbitrarily close to optimality~\cite{Ambainis08private}.  

We will use the following notation.  For $i \in [n] = \{1, 2, \ldots, n\}$, let $e^i = 0^{i-1} 1 0^{n-i} \in \{0,1\}^n$ be the bit string with a $1$ only in position $i$, and for $x \in \{0,1\}^n$, let $i \in x$ mean $x_i = 1$ and $i \notin x$ mean $x_i = 0$.  Let $\oplus$ denote the bitwise exor operation.  

For computing the nonnegative-weight adversary bounds, we will use a dual formulation that is a simplified version of Eq.~\eqnref{e:adversarydual}: 
\begin{theorem}[\cite{SpalekSzegedy04advequivalent}] \label{t:dualadversarydef}
Let $f: \{0,1\}^n \rightarrow \{0,1\}$.  Then 
\begin{equation} \label{e:dualadversarydef}
\ADV(f) = \min_{\{p_x\}} \max_{x, y : f(x) \neq f(y)} \frac{1}{ \sum_{i : x_i \neq y_i} \sqrt{p_x(i) p_y(i)} }
 \enspace ,
\end{equation}
where the first minimization is over distributions $p_x$ on $[n]$ for each $x \in \{0,1\}^n$.  
\end{theorem}

\subsection{Span programs for the threshold functions $\threshold{l}{n}$}

\begin{proposition} \label{t:thresholdspanprogram}
For $l \in [n]$, there exists a span program $P_{l}^{n}$ computing $f_{P_{l}^{n}} = \threshold{l}{n}$, with witness size 
\begin{equation} \label{e:thresholdspanprogram}
\wsize{P_{l}^{n}} \leq \sqrt{ l (n-l+1) }
 \enspace .
\end{equation}
\end{proposition}

\begin{proof}
The proof is by induction in $l$.  For the base case, $l=1$, $\threshold{1}{n}$ is the OR function, for which an optimal span program has $V = \C$, target vector $\ket t = 1$ and, for $i \in I = [n]$, input vector $\ket{v_i} = 1$ labeled by $(i,1)$.  For this span program, the witness size for inputs $x$ of Hamming weight $\abs{x} = j \geq 1$ is $1/j$, achieved by $\ket w = \frac{1}{j} \sum_{i \in x} \ket i$, and the witness size for $x = 0^n$ is $n$.  

For $i \in [n]$, let $x_{-i} = x_1 \ldots \widehat{x_i} \ldots x_n \in \{0,1\}^{n-1}$ be the string $x$ with the $i$th bit removed.  
For $l > 1$, the span program for $\threshold{l}{n}$ can be built recursively, by expanding out the formula 
\begin{equation} \label{e:thresholdexpansion}
\threshold{l}{n}(x_1, \ldots, x_n) = \bigvee_{i=1}^{n} \big( x_i \wedge \threshold{l-1}{n-1}(x_{-i}) \big)
 \enspace .
\end{equation}
By induction, let $P_{l-1}^{n-1}$ be an optimal span program for $\threshold{l-1}{n-1}$, over a vector space of dimension $d$ with target vector $\ket{t'} = (1, 0, \ldots, 0)$, and with witness sizes $1$ for inputs of Hamming weight $l-1$ and witness sizes $(l-1)(n-l+1)$ for inputs of Hamming weight $l-2$.  We construct span program $P_{l}^{n}$ over the vector space $V = \C \oplus (\C^n \tensor \C^d)$, of dimension $1 + n d$.  Let the target vector be $\ket t = (1,0)$.  For the $i$th term in Eq.~\eqnref{e:thresholdexpansion}, add the following ``block" of input vectors: $(1, \sqrt{l-1} \ket i \tensor \ket{t'})$ labeled by $(i,1)$, and $(0, \ket i \tensor \ket{v_j})$ for each input vector $\ket{v_j}$ of $P_{l-1}^{n-1}$ on $x_{-i}$.  

The span program $P_{l}^{n}$ indeed computes $\threshold{l}{n}$.  For computing the witness size of $P_{l}^{n}$, note that all input bits are symmetrical, so it suffices to consider inputs of the form $x = 1^j 0^{n-j}$.  
\begin{itemize}
\item In the true case, $j \geq l$, consider the witness $\ket w$ with weight $1/j$ on each of the input vectors $(1, \sqrt{l-1} \ket i \tensor \ket{t'})$ for $i \in [j]$ and then an optimal witness, of squared length at most $(\sqrt{l-1}/j)^2 \cdot \wsizex{P_{l-1}^{n-1}}{x_{-i}}$ within each of those $\threshold{l-1}{n-1}$ span program blocks.  The witness size is 
\begin{equation} \label{e:thresholdspanprogramtruecase}
\norm{\ket w}^2 = \frac{1}{j^2} \sum_{i \in x} \big(1 + (l-1) \wsizex{P_{l-1}^{n-1}}{x_{-i}}\big) \leq 1
 \enspace .
\end{equation}
\item In the false case, $j < l$, let the witness vector $\ket{w'} \in V$ orthogonal to the available input vectors be $\ket{w'} = \big(1, -\frac{1}{\sqrt{l-1}} \sum_{i \in x} \ket i \tensor \ket{w'_i} \big)$.  Here $\ket{w'_i}$ is an optimal witness vector for the span program $P_{l-1}^{n-1}$ on $x_{-i}$, i.e., orthogonal to the available input vectors and with $\braket{t'}{w'_i} = 1$.  Then $\braket{t}{w'} = 1$ and 
\begin{align}
\norm{A^\dagger \ket{w'}}^2 
 &= \sum_{i \notin x} 1 + \sum_{i \in x} \frac1{l-1} \norm{A^\dagger \ket{w'_i}}^2 \nonumber \\
 &= (n-j) + \sum_{i \in x} \frac1{l-1} \wsizex{P_{l-1}^{n-1}}{x_{-i}} \nonumber \\
 &\leq (n-j) + j(n-l+1) \nonumber \\
 &= n + j(n-l) \nonumber \\
 &\leq l (n - l + 1)
 \enspace ,
\end{align}
where in the two inequalities we have used $\wsizex{P_{l-1}^{n-1}}{x_{-i}} \leq (l-1)(n-j+1)$ and $j \leq l-1$, respectively.  
\end{itemize}
Thus $\wsize{P_{l}^{n}} \leq \sqrt{l (n-l+1)}$.  
\end{proof}

Letting $\size{P}$ be the number of input vectors of a span program~$P$~\cite{KarchmerWigderson93span}, note that $\size{P_{1}^{n}} = n$ and $\size{P_{l}^{n}} = n (1 + \size{P_{l-1}^{n-1}})$, which is exponential in $l$.  For example, for the three-majority function $\threshold{2}{3}$, $\size{P_{2}^{3}} = 9$.  This size is not optimal, even among span programs with optimal witness size.  

In \propref{t:intervalspanprogram} below, we will require slightly finer control over the threshold span program witness sizes: 

\begin{claim} \label{t:thresholdspanprogramwsizetrue}
On an input $x$ of Hamming weight $\abs{x} = j \geq l$, the span program $P_{l}^{n}$ constructed in \propref{t:thresholdspanprogram} satisfies
\begin{equation} \label{e:thresholdspanprogramwsizetrue}
\wsizex{P_{l}^{n}}{x} \leq \frac{1}{j-l+1}
 \enspace .
\end{equation}
\end{claim}

\begin{proof}
By induction in $l$.  The base case, $l=1$, was already considered as the base case for the induction in the proof of \propref{t:thresholdspanprogram}.  For $l > 1$, apply Eq.~\eqnref{e:thresholdspanprogramtruecase} and the induction assumption to get 
\begin{align}
\wsizex{P_{l}^{n}}{x}
&\leq \frac{1}{j^2} \sum_{i \in x} \big(1 + (l-1) \wsizex{P_{l-1}^{n-1}}{x_{-i}}\big) \\
&\leq \frac{1}{j} \big(1 + \frac{l-1}{j-l+1} \big) \nonumber \\
&= \frac{1}{j-l+1}
 \enspace . \qedhere
\end{align}
\end{proof}

\subsection{Span programs for the interval functions $\interval{l}{m}{n}$}

\begin{proposition} \label{t:intervalspanprogram}
There exists a span program $P_{l,m}^{n}$ computing $f_{P_{l,m}^{n}} = \interval{l}{m}{n}$, with witness size 
\begin{equation} \label{e:intervalspanprogram}
\wsize{P_{l,m}^n} \leq \sqrt{ (m+1)(n-m) + \frac{l(n-l+1)}{m-l+1} }
\end{equation}
when $\abs{\frac{n}{2} - m} \leq \abs{\frac{n}{2} - l}$.  
\end{proposition}

\begin{proof}
We use $\interval{l}{m}{n}(x) = \threshold{l}{n}(x) \wedge \threshold{n-m}{n}(\bar x)$ and combine the span programs $P_{l}^{n}$ for $\threshold{l}{n}$ and $P_{n-m}^{n}$ for $\threshold{n-m}{n}$ from \propref{t:thresholdspanprogram}.  

Let $V'$ and $V''$ be the vector spaces for $P_{l}^{n}$ and $P_{n-m}^{n}$, with target vectors $\ket{t'}$ and $\ket{t''}$, respectively.  As in \propref{t:thresholdspanprogram}, scale the target vectors so that witness sizes in the true cases are at most $1$ and in the false cases are at most $l(n-l+1)$ or $(n-m)(m+1)$ for $P_{l}^{n}$ and $P_{n-m}^{n}$, respectively.  Let $V = V' \oplus V''$ be the vector space for $P_{l,m}^{n}$, with target vector 
\begin{equation}
\ket t = \big( \sqrt{l(n-l+1)} \ket{t'}, \sqrt{(n-m)(m+1)} \ket{t''} \big) \in V
 \enspace .
\end{equation}
The input vectors for $P_{l,m}^{n}$ are exactly the input vectors of $P_{l}^{n}$ on input $x$ in the first component of $V$ and the input vectors of $P_{n-m}^{n}$ on input $\bar x$ in the second component of $V$.  This way, $f_{P_{l,m}^{n}} = 1$ if and only if both component span programs evaluate to true, so indeed $f_{P_{l,m}^{n}} = \interval{l}{m}{n}$.

Note that all input bits are symmetrical, so the witness size of $P_{l,m}^{n}$ on an input $x$ depends only on $j = \abs{x}$.  
\begin{itemize}
\item In the true case, $l \leq j \leq m$, the witness size is the sum of the squared lengths for witnesses for the two component span programs, from \claimref{t:thresholdspanprogramwsizetrue},  
\begin{align}
\wsizex{P_{l,m}^{n}}{x} 
&= l (n-l+1) \wsizex{P_{l}^{n}}{x} + (n-m)(m+1) \wsizex{P_{n-m}^{n}}{\bar x} \nonumber \\
&\leq \frac{l (n-l+1)}{j-l+1} + \frac{(n-m)(m+1)}{m-j+1}
 \enspace .
\intertext{As the above expression is convex up in $j \in [l, m]$, it is maximized for $j \in \{l, m\}$.  Since $\abs{\frac{n}{2} - m} \leq \abs{\frac{n}{2} - l}$, $l (n-l+1) \leq (m+1)(n-m)$, so $j = m$ is the worst case:}
\wsizex{P_{l,m}^{n}}{x} 
&\leq \frac{l (n-l+1)}{m-l+1} + (m+1)(n-m)
 \enspace .
\end{align}
\item In the false case, either $j < l$ or $j > m$, and we aim to show $\wsizex{P_{l,m}^{n}}{x} \leq 1$.  Take first the case $j > l$.  Consider a witness vector $\big(\frac{1}{\sqrt{l(n-l+1)}} \ket{w'}, 0\big) \in V$ where $\ket{w'}$ is an optimal witness vector to $f_{P_{l}^{n}}(x) = 1$.  The witness size is $\frac{1}{l(n-l+1)} \wsizex{P_{l}^{n}}{x} \leq 1$.  The case $j > m$ is dealt with symmetrically.  \qedhere
\end{itemize}
\end{proof}

\subsection{Adversary bounds for the interval functions $\interval{l}{m}{n}$}

\begin{proposition} \label{t:intervaladversary}
For the interval function $\interval{l}{m}{n}$ with $\abs{\frac{n}{2} - m} \leq \abs{\frac{n}{2} - l}$,   
\begin{equation}
\ADV(\interval{l}{m}{n}) = 
\begin{cases}
\sqrt{ (m+1)(n-m) + \frac{m(n-l+1)}{(m-l+1)^2} } & \text{if $l > 0$} \\
\sqrt{ (m+1)(n-m) } & \text{if $l = 0$}
\end{cases}
\end{equation}
In particular, $\ADV(\threshold{l}{n}) = \sqrt{l(n-l+1)}$.  
\end{proposition}

\begin{proof}
There are two steps to the proof.  First we give an adversary matrix $\Gamma$ that achieves for each $i \in [n]$ $\norm{\Gamma} / \norm{\Gamma \circ \Delta_i} =  \sqrt{ (m+1)(n-m) + \frac{m(n-l+1)}{(m-l+1)^2} }$ if $l > 0$, or $\sqrt{(m+1)(n-m)}$ if $l = 0$.  By \defref{t:adversarydef}, this lowers bounds $\ADV(\interval{l}{m}{n})$.  Second, we give a matching solution to the dual formulation of the nonnegative-weight adversary bound of \thmref{t:dualadversarydef}, in order to upper-bound $\ADV(\interval{l}{m}{n})$.  

Let
\begin{equation} \label{e:intervaladversarymatrix}
\Gamma = \sum_{x : \abs{x} = m} \ket x \Bigg( \sum_{i \notin x} \bra{x \oplus e^i} + c \sum_{\substack{y : \abs{y} = l-1\\ \abs{x \oplus y} = m-l+1}} \bra{y} \Bigg)
 \enspace ,
\end{equation}
where $c$ is to be determined.  For the case $l = 0$, the second term above is zero, so set $c = 0$.  

Then for each $i \in [n]$, let $\Gamma_i = \Gamma \circ \Delta_i$, so 
\begin{equation}
\begin{split}
\Gamma_i
 &= \sum_{x, y : x_i \neq y_i} \bra{x} \Gamma \ket y \\
 &= \sum_{\substack{x : \abs{x} = m\\ i \notin x}} \ketbra{x}{x \oplus e^i} + c \sum_{\substack{x : \abs{x} = m\\ i \in x}} \sum_{\substack{y : \abs{y} = l-1\\ \abs{x \oplus y} = m-l+1 \\ i \notin y}} \ketbra{x}{y}
  \enspace .
\end{split}
\end{equation}
Then 
\begin{equation}
\Gamma_i^\dagger \Gamma_i
 = \sum_{\substack{y : \abs{y} = m+1\\ i \in y}} \ketbra{y}{y} + c^2 \sum_{\substack{y, y', x\\ \abs{y} = \abs{y'} = l-1, \abs{x} = m\\ \abs{x \oplus y} = \abs{x \oplus y'} = m-l+1\\ i \in x, i \notin y, i \notin y'}} \ketbra{y}{y'} %\\
 \enspace .
\end{equation}
Thus $\Gamma_i^\dagger \Gamma_i$ is the direct sum of two matrices, for $l > 0$.  The first term above clearly has norm one, and we want to choose $c$ as large as possible so the second term also has norm one.  Now the eigenvector with largest eigenvalue for the second sum is, by symmetry, $\ket \psi = \sum_{x : \abs{x} = l-1, i \notin x} \ket x$, with eigenvalue $c^2 \binomial{n-l}{m-l} \binomial{m-1}{l-1}$.  Thus let 
\begin{equation} \label{e:intervaladversarymatrixc}
c = \big[ \binomial{n-l}{m-l} \binomial{m-1}{l-1} \big]^{-1/2}
\end{equation}
so $\norm{\Gamma_i} = 1$.  

Let us determine the norm of $\Gamma$.  We have 
\begin{equation} \label{e:intervaladversarymatrixnorm}
\norm{\Gamma \Gamma^\dagger} = \frac{ \bra{\psi_{m}} \Gamma \Gamma^\dagger \ket{\psi_{m}} }{ \braket{\psi_{m}}{\psi_{m}} }
 \enspace ,
\end{equation}
where $\ket{\psi_{m}} = \sum_{x : \abs{x} = m} \ket x$, $\norm{\ket{\psi_{m}}}^2 = \binomial{n}{m}$.  Then 
\begin{equation}
\begin{split}
\Gamma^\dagger \ket{\psi_{m}}
 &= \sum_{x : \abs{x} = m} \bigg[ \sum_{i \notin x} \ket{x \oplus e^i} + c \sum_{\substack{y : \abs{y} = l-1\\ \abs{x \oplus y} = m-l+1}} \ket y \bigg] \\
 &= \sum_{y : \abs{y} = m+1} (m+1) \ket y + c \sum_{y : \abs{y} = l-1} \binomial{n-l+1}{m-l+1} \ket y
\end{split}
\end{equation}
so 
\begin{equation}
\begin{split}
\norm{\Gamma \Gamma^\dagger}
 &= \frac{1}{\binomial{n}{m}} \bigg( (m+1)^2 \binomial{n}{m+1} + c^2 \binomial{n}{l-1} \binomial{n-l+1}{m-l+1}^2 \bigg) \\
 &= \begin{cases}
 (m+1)(n-m) + \frac{ m(n-l+1) }{ (m-l+1)^2 } & \text{if $l > 0$} \\
 (m+1)(n-m) & \text{if $l = 0$}
 \end{cases}
\end{split}
\end{equation}
This gives the desired lower bound on $\ADV(\interval{l}{m}{n})$.  

Next, we need to show a matching upper bound on $\ADV(\interval{l}{m}{n})$, using \thmref{t:dualadversarydef}.  For each $x$, we need a distribution $p_x$ on $[n]$.  For a function $f$ that is symmetrical under permuting the input bits, we look for distributions such that $p_x(i)$ depends only on whether $x_i = 0$ or $1$ and moreover its values in these cases depends only on $\abs{x}$.  Thus for $i = 0, 1, \ldots, n$, we fix a $p_i$, $0 \leq p_i \leq 1/i$ (with $p_0 = 0$) and set $p_i' = (1 - i p_i) / (n - i) \geq 0$ (with $p_n' = 0$).  Letting $p_i$ and $p_i'$ be the probabilities of $1$ and $0$ bits, respectively, when $\abs{x} = i$, Eq.~\eqnref{e:dualadversarydef} gives  
\begin{equation}
\ADV(f) 
\leq \min_{\{ p_i \}} \max_{\substack{x, y\\ f(x) \neq f(y)}} \bigg( \sum_{i : x_i = 1, y_i = 0} \sqrt{p_{\abs{x}} p_{\abs{y}}'} + \sum_{i : x_i = 0, y_i = 1} \sqrt{ p_{\abs{x}}' p_{\abs{y}} } \bigg)^{-1}
 \enspace .
\end{equation}
Fixing $\abs{x} = i$ and $\abs{y} = j$, the inner maximum is achieved by $x = 1^i 0^{n-i}$ and $y = 1^j 0^{n-j}$ because these strings have the fewest differing bits.  Thus the above bound simplifies to 
\begin{equation} \label{e:dualadvboundsymmetrical}
\ADV(f) 
\leq \min_{\{ p_i \}} \max_{\substack{ i < j\\ f(1^i 0^{n-i}) \neq f(1^j 0^{n-j})}} \bigg( (j-i) \sqrt{p_i' p_j} \bigg)^{-1}
 \enspace .
\end{equation}

Now specialize from symmetrical functions down to the Hamming-weight interval function $f = \interval{l}{m}{n}$.  For $i \geq m+1$, we should clearly set $p_i$ as large as possible, i.e., set $p_i = 1/i$, while for $i < l$ we should set $p_i'$ as large as possible, i.e., $p_i' = 1/(n-i)$.  

First consider the case $l = 0$.  Then we should set $p_i' = 1/(n-i)$ for all $i \leq m$ in order to minimize the expression in Eq.~\eqnref{e:dualadvboundsymmetrical}.  This gives 
\begin{equation}\begin{split}
\ADV(\interval{0}{m}{n}) 
&\leq \max_{\substack{0 \leq i \leq m\\ m+1 \leq j \leq n}} \frac{\sqrt{(n-i)j}}{j-i} \\
&= \sqrt{(m+1)(n-m)}
 \enspace ,
\end{split}\end{equation}
where the maximum is achieved at $i = m$, $j = m+1$.  

Now assume $l > 0$.  It turns out that there is some freedom in the choice of $p_i$ for $l \leq i < m$.  For $i = l, \ldots, m$, choose $p_i$ to balance the $(l-1,i)$ and $(i, m+1)$ terms above, i.e., setting 
\begin{equation}
\bigg( {(i-l+1)\sqrt{p_{l-1}' p_i}} \bigg)^{-1} = \bigg( {(m-i+1) \sqrt{p_i' p_{m+1}}} \bigg)^{-1}
 \enspace .
\end{equation}
Since $p_{l-1}' = 1 / (n-l+1)$ and $p_{m+1} = 1 / (m+1)$, this gives 
\begin{equation}
p_i = \frac{1}{ i + \frac{(m+1)(n-i)}{n-l+1} \big( \frac{i-l+1}{m-i+1} \big)^2 }
 \enspace .
\end{equation}
Substituting this value for $p_i$ back in, the $(l-1,i)$ and $(i,m+1)$ terms are both the square root of 
\begin{equation}
f(n, l, m, i) := \frac{i(n-l+1)}{(i-l+1)^2} + \frac{(n-i)(m+1)}{(m-i+1)^2}
 \enspace .
\end{equation}
The case $i = m$ gives the bound we are aiming for.  We claim that this is the worst case, i.e., that $f(n, l, m, i) \leq f(n, l, m, m)$ when $\abs{\frac{n}{2} - m} \leq \abs{\frac{n}{2} - l}$.  

First note that 
\begin{equation}
\frac{\partial^2}{\partial i^2} f(n, l, m, i) = 2 \frac{(n-l+1)(i+2l-2)}{(i-l+1)^4} + 2 \frac{(m+1)(3n-i-2m-2)}{(m-i+1)^4} > 0
 \enspace .
\end{equation}
Thus it suffices to check that $f(n, l, m, l) \leq f(n, l, m, m)$.  Indeed, 
\begin{equation}
f(n, l, m, m) - f(n, l, m, l) = \frac{(n-m-l)\big((m-l+1)^3-1\big)}{(m-l+1)^2}
 \enspace .
\end{equation}
Note that $l \leq m$.  The above difference is clearly $\geq 0$ if $m \leq \frac{n}{2}$.  If $m > \frac{n}{2}$, then the assumption $\abs{\frac{n}{2} - m} \leq \abs{\frac{n}{2} - l}$ implies that $l < \frac{n}{2}$ and $m - \frac{n}{2} \leq \frac{n}{2} - l$, i.e., $m + l \leq n$; so again the above difference is $\geq 0$.  
\end{proof}

\begin{proposition}
For the interval function $\interval{l}{m}{n}$, $\ADV(\interval{l}{m}{n}) < \ADVpm(\interval{l}{m}{n})$ if and only if $l \notin \{ 0, 1, m, n-1, n \}$.  
\end{proposition}

\begin{proof}
\def\Gammaeps {{\Gamma^{(\epsilon)}}}
\def\Gammaepsi {{\Gamma_i^{(\epsilon)}}}
\newcommand{\depsilon}[1]{{\frac{\partial}{\partial \epsilon} {#1} \big\vert_{\epsilon = 0}}}
\newcommand{\ddepsilon}[1]{{\frac{\partial^2}{\partial \epsilon^2} {#1} \big\vert_{\epsilon = 0}}}

For $l \in \{ 0, 1, m, n-1, n \}$, $\ADV(\interval{l}{m}{n}) = \ADVpm(\interval{l}{m}{n})$ because \propref{t:intervalspanprogram} gave a span program $P_{l,m}^{n}$ with witness size $\wsize{P_{l,m}^{n}} = \ADV(\interval{l}{m}{n})$, and by \thmref{t:wsizeadvbound}, $\wsize{P_{l,m}^{n}} \geq \ADVpm(\interval{l}{m}{n})$.  

Otherwise, assume that $2 \leq l \leq m-1$ and $\abs{\frac{n}{2} - m} \leq \abs{\frac{n}{2}-l}$.  We will show that a perturbation of the adversary matrix $\Gamma$ from Eqs.~\eqnref{e:intervaladversarymatrix} and~\eqnref{e:intervaladversarymatrixc} in the proof of \propref{t:intervaladversary} increases $\norm{\Gamma} / \norm{\Gamma \circ \Delta_i}$ for each $i \in [n]$.  The perturbation we consider will be in the direction of 
\begin{equation}
\Lambda = \sum_{x : \abs{x} = m-1} \ket x \bigg( \sum_{\substack{y : \abs{y} = l-1\\ \abs{x \oplus y} = m-l}} \bra{y} - \delta \sum_{\substack{y : \abs{y} = l-1\\ \abs{x \oplus y} = m-l+2}} \bra{y} \bigg)
 \enspace ,
\end{equation}
where $\delta > 0$ will be determined later.  Let $\Gammaeps = \Gamma + \eps \Lambda$.  Let $\Lambda_i = \Lambda \circ \Delta_i$ and $\Gammaepsi = \Gammaeps \circ \Delta_i$.  

First of all, note that $\depsilon{\norm{\Gammaeps} / \norm{\Gammaepsi}} = 0$.  Indeed, 
\begin{equation}
\depsilon{\norm{\Gammaeps}}
= \frac{1}{2 \norm{\Gamma}} \depsilon{\norm{\Gammaeps^\dagger \Gammaeps}}
 \enspace .
\end{equation}
However, $\Gammaeps^\dagger \Gammaeps = \Gamma^\dagger \Gamma + \epsilon^2 \Lambda^\dagger \Lambda$ since $\Gamma^\dagger \Lambda = \Lambda^\dagger \Gamma = 0$.  Thus $\depsilon{\norm{\Gammaeps}} = 0$, and similarly $\depsilon{\norm{\Gammaepsi}} = 0$.  

Therefore, we need to compute $\ddepsilon{\norm{\Gammaeps} / \norm{\Gammaepsi}}$.  Now 
\begin{equation}\begin{split}
\ddepsilon{\norm{\Gammaeps}}
&= \frac{1}{2\norm{\Gamma}} \ddepsilon{\norm{\Gamma^\dagger \Gamma + \epsilon^2 \Lambda^\dagger \Lambda}} \\
&= \frac{1}{\norm{\Gamma}} \depsilon{ \norm{\Gamma^\dagger \Gamma + \epsilon \Lambda^\dagger \Lambda} }
 \enspace .
\end{split}\end{equation}
By the Perron-Frobenius theorem, $\Gamma^\dagger \Gamma$ has a unique eigenvalue of largest magnitude, and it is nondegenerate.  Letting $\ket \psi$ be the corresponding eigenvector, we have by nondegenerate perturbation theory
\begin{equation} \label{e:ddepsilonperturbation}
\ddepsilon{\norm{\Gammaeps}}
= \frac{1}{\norm{\Gamma}} \frac{ \norm{\Lambda \ket \psi}^2 }{ \norm{\ket \psi}^2 }
 \enspace .
\end{equation}
For $j = 0, 1, \ldots, n$, let $\ket{\psi_j} = \sum_{x : \abs{x} = j} \ket x$, with $\norm{\ket{\psi_j}}^2 = \binomial{n}{j}$.  By Eq.~\eqnref{e:intervaladversarymatrixnorm}, we may take 
\begin{align}
\ket \psi
&= \Gamma^\dagger \ket{\psi_m} \nonumber \\
&= \sum_{x : \abs{x} = m} \Bigg( \sum_{i \notin x} \ket{x \oplus e^i} + c \sum_{\substack{y : \abs{y} = l-1\\ \abs{x \oplus y} = m-l+1}} \ket y \Bigg) \nonumber \\
&= (m+1) \ket{\psi_{m+1}} + c \binomial{n-l+1}{m-l+1} \ket{\psi_{l-1}}
 \enspace ,
\intertext{so}
\Lambda \ket\psi 
&= c \binomial{n-l+1}{m-l+1} \ket{\psi_{m-1}} \big[ \binomial{m-1}{m-l} - \delta \binomial{m-1}{m-l-1}(n-m+1) \big]
 \enspace .
\end{align}
Substituting this into Eq.~\eqnref{e:ddepsilonperturbation}, 
\begin{align}
\ddepsilon{\norm{\Gammaeps}}
&= \frac{1}{\norm{\Gamma}} \frac{ c^2 \binomial{n-l+1}{m-l+1}^2 \binomial{n}{m-1} \big[ \binomial{m-1}{m-l} - \delta \binomial{m-1}{m-l-1}(n-m+1) \big]^2 }{ (m+1)^2 \binomial{n}{m+1} + c^2 \binomial{n-l+1}{m-l+1}^2 \binomial{n}{l-1} } \nonumber \\
&= \frac{1}{\norm{\Gamma}} \frac{\binomial{m}{l-1} \binomial{n-l+1}{m-l+1}}{\big( \frac{(m+1)(n-m)}{n-l+1} + \frac{m}{(m-l+1)^2} \big)(n-m+1)} \bigg( 1 - \frac{(m-l)(n-m+1)}{l} \delta \bigg)^2
 \enspace .
\end{align}

Unlike $\Gamma^\dagger \Gamma$, $\Gamma_i^\dagger \Gamma_i$ has a degenerate principal eigenspace.  This principal eigenspace is spanned by $\ket \phi = \sum_{\substack{x : \abs{x} = l-1 \\ i \notin x}} \ket x$ and $\ket{\phi'} = \sum_{\substack{x : \abs{x} = m+1 \\ i \in x}} \ket x$.  Since $\Lambda_i \ket{\phi'} = 0$, we have by degenerate perturbation theory 
\begin{equation} \label{e:ddepsiloniperturbation}
\begin{split}
\ddepsilon{\norm{\Gammaepsi}}
&= \frac{1}{\norm{\Gamma_i}} \depsilon{ \norm{\Gamma_i^\dagger \Gamma_i + \epsilon \Lambda_i^\dagger \Lambda_i} } \\
&= \frac{1}{\norm{\Gamma_i}} \frac{ \norm{\Lambda_i \ket \phi}^2 }{ \norm{\ket \phi}^2 }
\end{split}\end{equation}
Recall that $\norm{\Gamma_i} = 1$, and note that $\norm{\ket \phi}^2 = \binomial{n-1}{l-1}$.  Then
\begin{equation}
\Lambda_i \ket \phi 
= \big[ \binomial{m-2}{m-l-1} - \delta \binomial{m-2}{m-l-2}(n-m+1) \big] \sum_{\substack{x : \abs{x} = m-1 \\ i \in x}} \ket x
 \enspace .
\end{equation}
Substituting into Eq.~\eqnref{e:ddepsilonperturbation}, 
\begin{align}
\ddepsilon{\norm{\Gammaepsi}}
&= \frac{ \binomial{n-1}{m-2} }{ \binomial{n-1}{l-1} } \big[ \binomial{m-2}{m-l-1} - \delta \binomial{m-2}{m-l-2}(n-m+1) \big]^2 \nonumber \\
&= \binomial{m-2}{l-1} \binomial{n-l}{m-l-1} \bigg( 1 - \frac{(l-1)(n-m+1)}{m-l} \delta \bigg)^2
 \enspace .
\end{align}

Now set $\delta = \frac{m-l}{(l-1)(n-m+1)}$; recall that $l \geq 2$ so the denominator is nonzero.  We get $\ddepsilon{\norm{\Gammaepsi}} = 0$ while $\ddepsilon{\norm{\Gammaeps}} > 0$.  Thus $\ddepsilon{{\norm{\Gammaeps}}/{\norm{\Gammaepsi}}} > 0$, so $\ADV(\interval{l}{m}{n}) < \ADVpm(\interval{l}{m}{n})$.  
\end{proof}

\section{Examples of composed span programs} \label{s:compositionexamples}

In order to illustrate the different methods of span program composition used in \thmref{t:spanprogramcomposition} and \propref{t:reducedtensorproductcompose}, in this appendix we give examples of span program direct-sum composition (\defref{t:directsumcomposedef}), tensor-product composition (\defref{t:tensorproductcomposedef}), and reduced-tensor-product composition (\defref{t:reducedtensorproductcomposedef}).  For presenting the examples, we use the correspondence from \defref{t:spanprogramadjacencymatrix} between span programs and bipartite graphs.   

Our examples will use the following monotone span programs for fan-in-two AND and OR gates: 

\begin{definition} \label{t:andorspanprogramdef}
Define span programs $P_{\AND}$ and $P_{\OR}$ computing $\AND$ and $\OR$, $\B^2 \rightarrow \B$, respectively, by 
\begin{align}
P_{\AND}: &&
\ket t &= \left( \begin{matrix} \alpha_1 \\ \alpha_2 \end{matrix} \right) ,\; 
&\ket{v_1} &= \left( \begin{matrix} \beta_1 \\ 0 \end{matrix} \right) ,\; &\ket{v_2} &= \left( \begin{matrix} 0 \\ \beta_2 \end{matrix} \right) \\
P_{\OR}: &&
\ket t &= \delta ,\; &\ket{v_1} &= \epsilon_1 ,\;& \ket{v_2} &= \epsilon_2 
\end{align}
for parameters $\alpha_j, \beta_j, \delta, \epsilon_j > 0$, $j \in \{1,2\}$.  
Both span programs have $I_{1,1} = \{1\}$, $I_{2,1} = \{2\}$ and $\Ifree = I_{1,0} = I_{2,0} = \emptyset$.  
Let $\alpha = \sqrt{\alpha_1^2 + \alpha_2^2}$.  
\end{definition}

Now let $\varphi : \B^n \rightarrow \B$ be a size-$n$ AND-OR formula in which all gates have fan-in two.  By composing the span programs of \defref{t:andorspanprogramdef} according to $\varphi$, we obtain a span program $P_\varphi$ computing $\varphi$.  The particular composed span program $P_\varphi$ will depend on what composition method is used.  \figref{f:reducedtensorproductcompositionexamples} gives several examples of tensor-product and reduced-tensor-product composition.  Much like a canonical span program, the structure of the reduced-tensor-product-composed span program is related to the set of ``maximal false" inputs to $\varphi$.  
\figref{f:exampledirectsum} compares reduced-tensor-product composition to direct-sum composition, as well as to the graphs used in the AND-OR formula-evaluation algorithms of Refs.~\cite{AmbainisChildsReichardtSpalekZhang07andor, fgg:and-or}.  Although these algorithms did not use the span program framework, the graphs they use do correspond to span programs, built essentially according to direct-sum composition of $P_{\AND}$ and $P_{\OR}$.  The small-eigenvalue spectral analysis in \thmref{t:bipartitepsdreduction} simplifies their proofs.  

Although not shown here, the different composition methods can also be combined.  Hybrid-composed span programs will be analyzed in~\cite{Reichardt09andorfaster}.  

Typical parameter choices for $P_{\AND}$ and $P_{\OR}$ are given by: 

\begin{claim} \label{t:andorspanprogramwsize}
With the parameters in \defref{t:andorspanprogramdef} set to 
\begin{align}
\alpha_j &= (s_j / s_p)^{1/4} & \beta_j &= 1 \\
\delta &= 1 & \epsilon_j &= (s_j / s_p)^{1/4}
 \enspace ,
\end{align}
where $s_p = s_1 + s_2$, the span programs $P_{\AND}$ and $P_{\OR}$ satisfy: 
\begin{align}\begin{split}
\wsizexS{P_{\AND}}{x}{(\sqrt{s_1}, \sqrt{s_2})}
&=
\begin{cases}
\sqrt{s_p} & \text{if $x \in \{11, 10, 01\}$} \\
\frac{\sqrt{s_p}}{2} & \text{if $x = 00$}
\end{cases} \\
\wsizexS{P_{\OR}}{x}{(\sqrt{s_1}, \sqrt{s_2})}
&=
\begin{cases}
\sqrt{s_p} & \text{if $x \in \{00, 10, 01\}$} \\
\frac{\sqrt{s_p}}{2} & \text{if $x = 11$}
\end{cases}
\end{split}\end{align}
\end{claim}

It can be seen as a consequence of De Morgan's laws and span program duality (\lemref{t:dualspanprogram}) that $\wsizexS{P_{\AND}}{x}{(\sqrt{s_1}, \sqrt{s_2})} = \wsizexS{P_{\OR}}{\bar x}{(\sqrt{s_1}, \sqrt{s_2})}$ in \claimref{t:andorspanprogramwsize}.  

\begin{figure}
\centering
\begin{tabular}{c@{$\quad$}c}
\subfigure[$x_1 \vee x_2$]{\label{f:orgate}\raisebox{.2in}{\includegraphics[scale=1]{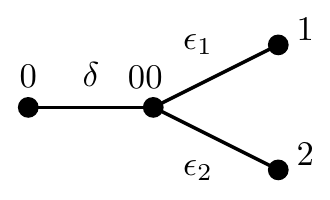}}}&
\subfigure[$x_1 \wedge x_2$]{\label{f:andgate}\includegraphics[scale=1]{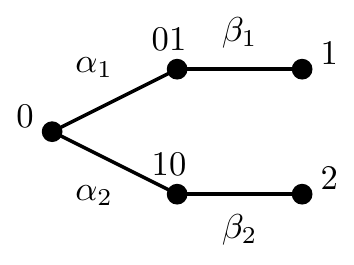}}\\
\subfigure[$(x_1 \wedge x_2) \vee x_3$]{\label{f:or_and}\raisebox{.05in}{\includegraphics[scale=1]{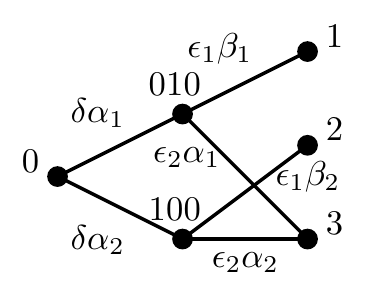}}}&
\subfigure[$(x_1 \vee x_2) \wedge x_3$]{\label{f:and_or}\includegraphics[scale=1]{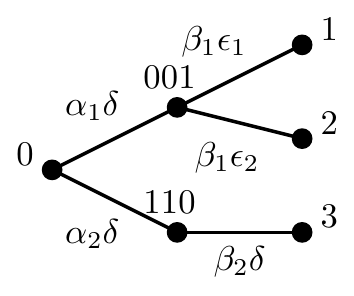}}
\\\multicolumn{1}{c}{\subfigure[$\big((x_1 \wedge x_2) \vee x_3\big) \wedge x_4$]{\label{f:and_or_and}\includegraphics[scale=1]{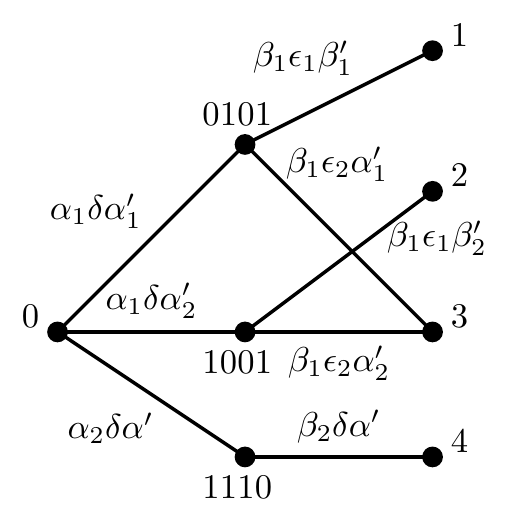}}}&
\subfigure[$\big((x_1 \wedge x_2) \vee x_3\big) \wedge x_4$]{\label{f:and_or_and_tensor}\includegraphics[scale=1]{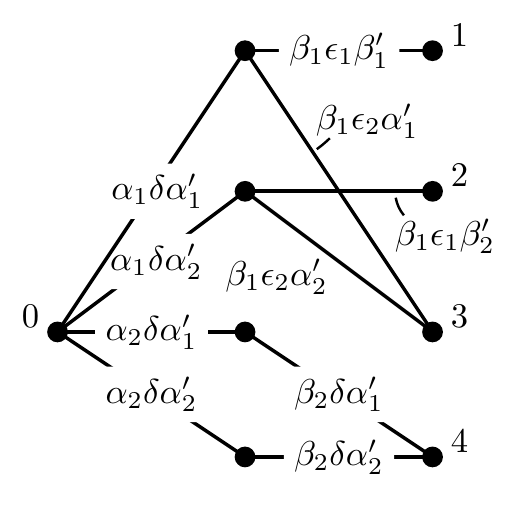}}
\end{tabular}
\caption{In (a) and (b) are given the graphs $G_{P_{\OR}}$ and $G_{P_{\AND}}$, respectively, according to \defref{t:spanprogramadjacencymatrix}.  Parts (c) and (d) show tensor-product compositions of these span programs, which are also the reduced-tensor-product compositions.  Part (e) shows the reduced-tensor-product composition of the span programs for a larger formula.  Notice that for reduced-tensor-product composition, the structure of the graph changes locally as each additional gate is composed onto the end of the formula, e.g., going from (d) to (e).  However, composing additional gates has a nonlocal effect on edge weights.  In each graph, the output vertex is labeled $0$ and the input vertices are labeled by $[n]$.  Similarly to canonical span programs, \defref{t:spanprogramcanonicaldef}, the other vertices are labeled by the maximal false inputs to the formula; notice in each example that a vertex labeled with input $x$ is connected exactly to those input bits $j \in [n]$ with $x_j = 0$.  Part~(f) shows a span program for the same formula as part (e), except built using tensor-product composition.  The vertex $1110$ has been unnecessarily duplicated.  
In (e) and (f), there are two $\AND$ gates; the primed variables refer to the $P_{\AND}$ span program coefficients for $x_1 \wedge x_2$.} \label{f:reducedtensorproductcompositionexamples}
\end{figure}

\begin{figure}
\centering
\begin{tabular}{c@{$\quad$}c}
\multicolumn{1}{c}{\subfigure[]{\label{f:formula}\includegraphics[scale=1]{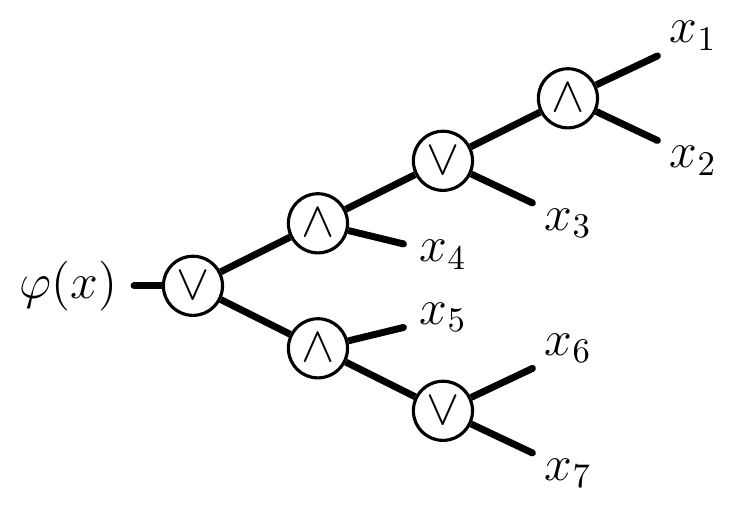}}} &
\subfigure[]{\includegraphics[scale=1]{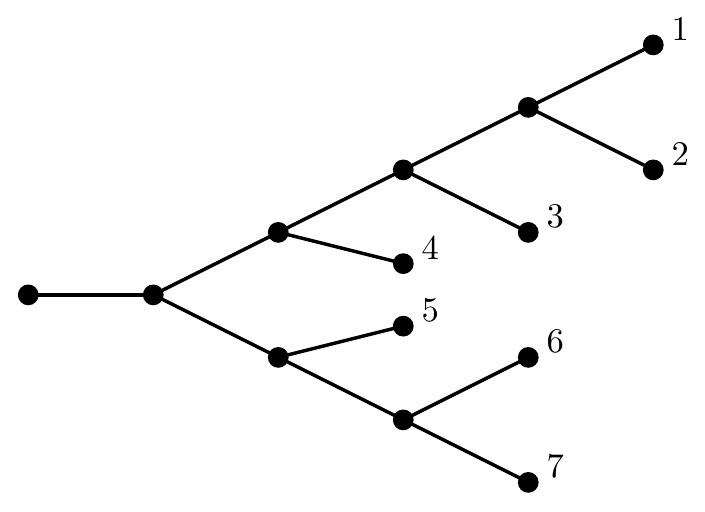}} \\
\subfigure[]{
$\place{\includegraphics[scale=1]{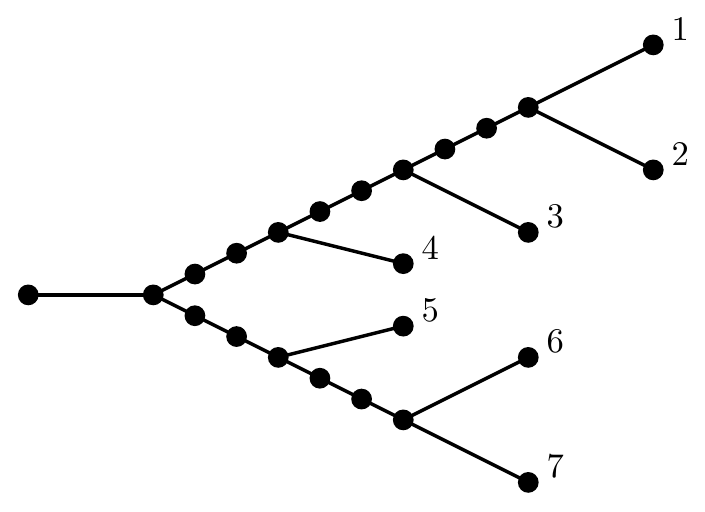}}{0mu}{53.5pt}$	%% centered vertically manually
}
& 
\subfigure[]{\label{f:tensorproductgraph}\includegraphics[scale=1]{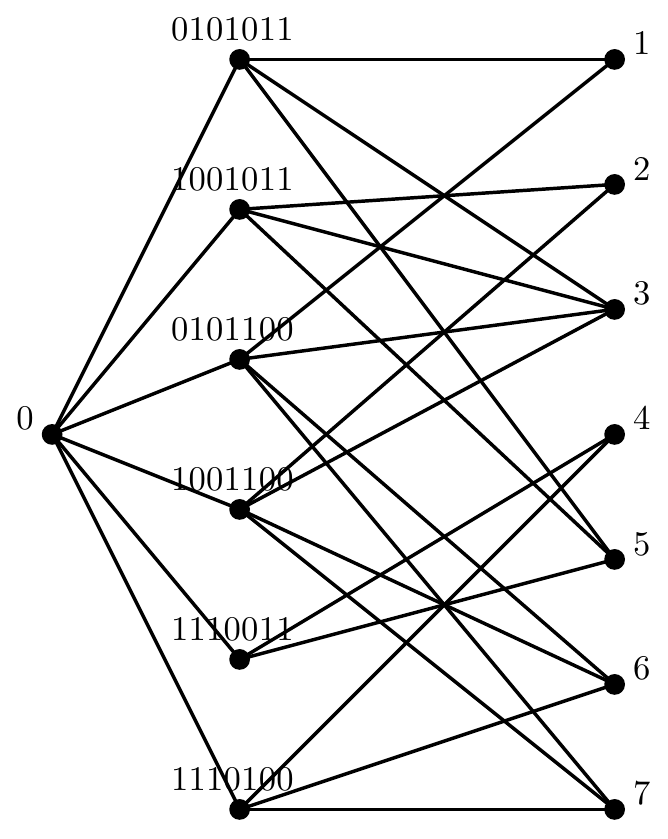}}
\end{tabular}
\caption{Consider the AND-OR formula $\varphi(x) = \big( [ (x_1 \wedge x_2) \vee x_3 ] \wedge x_4 \big) \vee \big( x_5 \wedge [x_6 \vee x_7] \big)$, represented as a tree in (a).  Part~(b) shows the graph on which~\cite{AmbainisChildsReichardtSpalekZhang07andor} runs a quantum walk in order to evaluate $\varphi$.  The graph is essentially the same as the formula tree.  The weight of an edge from child $v$ to parent $p$ is the $1/4$ power of the ratio $s_v / s_p$ of sizes of the subformula rooted at $v$ to that rooted at $p$, as in \claimref{t:andorspanprogramwsize}.  The only exception is the weight of the edge to the root, which is set to $1/n^{1/4}$ for amplification, as in \thmref{t:spanprogramspectralanalysis} and \thmref{t:generalspanprogramalgorithmnonblackbox}.  
Part~(c) shows the graph one obtains by from direct-sum composition of $P_{\AND}$ and $P_{\OR}$.  It is the same as in (b), except with two weight-one edges inserted above each internal gate.  These edges can be interpreted as pairs of NOT gates that cancel out.  Including them would slow the~\cite{AmbainisChildsReichardtSpalekZhang07andor} algorithm down only by a constant factor.  Part (d) shows a span program derived from the same formula using reduced-tensor-product composition only.  Vertices are labeled using the same convention as in \figref{f:reducedtensorproductcompositionexamples}.  Even though every gate has fan-in two, graph vertices can have exponentially large degree.  
} \label{f:exampledirectsum}
\end{figure}

\end{document}